\documentclass[11pt]{article}

\usepackage{ifxetex,ifluatex}
\newif\ifxetexorluatex
\ifxetex
  \xetexorluatextrue
\else
  \ifluatex
    \xetexorluatextrue
  \else
    \xetexorluatexfalse
  \fi
\fi

\ifxetexorluatex
  \usepackage{fontspec}
\else
  \usepackage[T1]{fontenc}
  \usepackage[utf8]{inputenc}
  \usepackage{lmodern}
\fi
\usepackage{amsthm}
\usepackage{paithan}
\usepackage{xcolor}
\usepackage{mathtools}
\usepackage{tikz}
\usepackage{enumitem}
\usetikzlibrary{arrows.meta}
\usetikzlibrary{shapes, fit, backgrounds}

\usepackage{braket}

\theoremstyle{definition}
\newtheorem{definition}{Definition}[section]
\newtheorem{theorem}{Theorem}[section]
\newtheorem{corollary}{Corollary}[section]
\newtheorem{lemma}{Lemma}[section]
\newtheorem{property}{Property}
\newtheorem{observation}{Observation}[section]
\newtheorem{openquestion}{Open Question}[section]
\newtheorem{conjecture}{Conjecture}[section]
\newtheorem{proposition}{Proposition}[section]
\providecommand{\floor}[1]{\left \lfloor #1 \right \rfloor}
\usepackage[hidelinks]{hyperref}
\usepackage[capitalise,noabbrev]{cleveref}
\usepackage{geometry}
\usepackage{setspace}
\def\scaleAmt{1.3}

\def\QuantumLift{{\cal Q}}
\def\filter{{\mbox{\rm filter}}}

\clearpage
\title{Quantum Combinatorial Games:\\ Structures and Computational Complexity}
\author{Kyle W. Burke
\\Plymouth State\and
Matthew T. Ferland
\\ USC \and Shang-Hua Teng\thanks{
Supported by the Simons Investigator Award for fundamental \& curiosity-driven research and NSF grant CCF-1815254.}\\ USC}

\date{}

\vspace{-0.5in}
\begin{document}
\maketitle
\begin{singlespace}

\begin{abstract}
In the past a few decades, a variety of fields---from  information theory to economic game theory to cryptography---have explored what happens when quantum decisions are allowed in classical frameworks or as a computational mechanism to improve classical algorithms. Recently, a standardized framework was proposed for introducing {\em quantum-inspired moves} in mathematical games with perfect information and no chance.

The beauty of quantum games---succinct in representation, rich in structures, explosive in complexity, dazzling for visualization, and sophisticated for strategical reasoning---has drawn us to play concrete games full of subtleties and to characterize abstract properties pertinent to complexity consequence.
Going beyond individual games, we
  explore the tractability of {\em quantum combinatorial games}
  as whole, and address fundamental questions including:

\begin{itemize}
\item {\bf Quantum Leap in Complexity}: {\em Are there polynomial-time solvable combinatorial games whose quantum extensions are intractable?}

\item {\bf Quantum Collapses in Complexity}: {\em Are there \cclass{PSPACE}-complete combinatorial games whose quantum extensions fall to the lower levels of the polynomial-time hierarchy?}

\item {\bf Quantumness Matters}: {\em How do outcome classes and strategies change under quantum moves? Under what conditions doesn't quantumness matter?}

\item{\bf PSPACE Barrier for Quantum Leap}: {\em Can quantum moves launch
\cclass{PSPACE} games into the {\em outer} polynomial space?}
\end{itemize}

We show that quantum moves not only enrich the structure of combinatorial games, but also impact their computational complexity.
In settling some of these basic questions, we characterize both the powers and limitations of quantum moves as well as the superposition of game configurations that they create.
Our constructive proofs---both on the leap of complexity in concrete \ruleset{Quantum Nim} and \ruleset{Quantum Undirected Geography} and on the continuous  collapses,  in the quantum setting, of complexity in
abstract \cclass{PSPACE}-complete games  to each level of the polynomial-time hierarchy---illustrate
  the striking computational landscape
  over quantum games and highlight
  surprising turns 
with unexpected 
quantum impact.
Particularly, our affirmative answer to the second question provides a complexity-theoretical refutation to the following seemingly intuitive conjecture:
{\em  Given that players have more options to choose strategically, combinatorial games is always as least as hard (computationally) in the quantum setting than in the classical setting}?
Our studies also enable us to identify several elegant open questions fundamental to quantum combinatorial game theory (QCGT).
\end{abstract}

\newpage

\clearpage
\pagenumbering{arabic}
\setcounter{page}{1}

\section{Introduction}
Quantum games have---in some fashion---existed for over a decade now.
Quantum \ruleset{Tic-Tac-Toe} was introduced by Allan Goff as an educational tool for learning quantum mechanics~\cite{goff2006quantum}.
This game's popularity made it subject to further study, including exploring the entire game tree~\cite{2011takumiquantum}.
Quantum \ruleset{Chess} is another well known exploration of introducing quantum elements to perfect information games \cite{akl2010importance}.
There have been other ways to introduce quantum elements to combinatorial
games.  Take, for example, a recent result for \ruleset{Cops and Robbers} \cite{glos2019role}. 
In economic game theory, researchers have also explored what happens when agents are allowed to make quantum decisions and how quantum decisions impact game and economic equilibria~\cite{eisert1999quantum}.


\subsection{Quantum Extension of Combinatorial Games}
In this paper, we study the complexity-theoretical impacts of {\em quantum-inspired moves} to combinatorial games, and properties of quantum game structures and dynamics.
Combinatorial games are mathematical games with perfect information and no chance. There is a rich area of mathematics and theoretical computer science devoted to studying these games and their elegant and inherent structural properties \cite{WinningWays:2001}.  As we will see in this paper, it is these structural properties that create a unique quantum game system, different than the systems mentioned above.

In a nutshell, {\em quantum moves}---unlike deterministic classical moves---can create superpositions of normal game positions.
Going beyond individual games,  our study focuses on a systematic quantum extension of general combinatorial games.
Therefore, several of our results
 apply to quantum extensions of both classical combinatorial games
 ---such as \ruleset{Nim} \cite{Bouton:1901}, \ruleset{Node-Kayles} \cite{DBLP:journals/jcss/Schaefer78}, \ruleset{Geography} \cite{PapadimitriouBook:1994}, \ruleset{Hex} \cite{NashHex,Gale:1979}, \ruleset{QSAT}/\ruleset{QBF} and \ruleset{Avoid-True} \cite{DBLP:journals/jcss/Schaefer78}---and games of our own creation, like \ruleset{Atropos} \cite{Burke2008,10.5555/1714301}.

Intuitively, games with superpositions of moves and superpositions of positions have richer structures than their classical counterparts, and consequently, may be more challenging to play optimally.
Our studies aim to characterize the impact of
  quantumness to the underlying computational problem of determining winnability and strategies.

Mathematically, the formalization of quantum moves and
 their induced superposition of game positions enjoys great subtlety.
We will follow the recent general framework introduced by Dorbec and Mhalla~\cite{dorbec2017toward} for turning combinatorial games
into quantum variants.\footnote{We would like to point out that by {\em quantum combinatorial games}, we don't mean games defined on {\em continuous} qubits and unitary transformations as in quantum computing.
Rather, quantum combinatorial games draw on quantum concepts of superposition and entanglement to enhance the classical combinatorial games with a ``quantum-inspired'' {\em discrete} framework of moves and game boards. So, ``{\em quantum-inspired combinatorial games}''
might be the more precise and literal name for these games.}
By applying this framework to some well-known classical games, we will introduce a few new, practically inspired quantum games, and demonstrate that they enjoy several desirable complexity-theoretical properties.

\vspace{0.1in}
\noindent {\sc Quantum Game Dynamics}:
Without getting into the subtleties,
in the quantum extension of a combinatorial game, each player can make a classic or a quantum move each turn.
A {\em classical move} is a move a player can make normally in the game, resulting in a new game position. (E.g., in \ruleset{Go}, placing a white stone on a given point.)
A {\em quantum move} is a {\em superposition} of a set of multiple classical moves.

Like a quantum move, a {\em quantum position} is a superposition of a set of classical game positions.
We call one possible position in a superposition
  a {\em realization}.
The dynamic of a quantum 
combinatorial game
  is then a sequence of quantum game positions---superposition of realizations---induced by players' moves.
Moves---either classical or quantum---can only be made when it is {\em legal in at least one realization} in the overall superposition.
Whenever a move is made which is illegal in a realization, that realization can no longer be factored in for making future moves, and we say that the realization {\em collapses}.
So, as a quantum game progresses, the number of realizations in the positions can go up and down.
When all of a player's moves would cause a realization to collapse, we call it a {\em terminal} realization for that player.
Under the normal-play\footnote{We do not apply quantumness to Mis\`{e}re play or Scoring games, the usual alternatives to Normal Play.  These seem like excellent areas for future work.} convention for combinatorial games, the game ends and the next player loses on a quantum game position when all its realizations are terminal for that player.

\vspace{0.1in}
\noindent {\sc An Example of Quantum Impact}:
We use the popular game \ruleset{Nim}\footnote{Also known in ancient China as Jian Shi Zi (picking pebbles).
The game starts with a collection of piles of pebbles.
Two (or multiple players)
takes turns picking at least one pebbles from one of the piles.
Under normal play, the player taking the last pebble wins the game.} \cite{Bouton:1901} to provide a concrete example.
In the classical setting, a well-known losing position for the current player is the configuration consisting of a pair of two-pebble piles.
We denote this game position by the two-dimensional vector $(2,2)$,
The current player cannot win at this position because regardless of whether they pick one or two from a pile,
the other player can simply match
the action in the other pile to win
either immediately in the case when the current player takes two pebbles, or in one more round, after each player picks one further pebble from the piles.

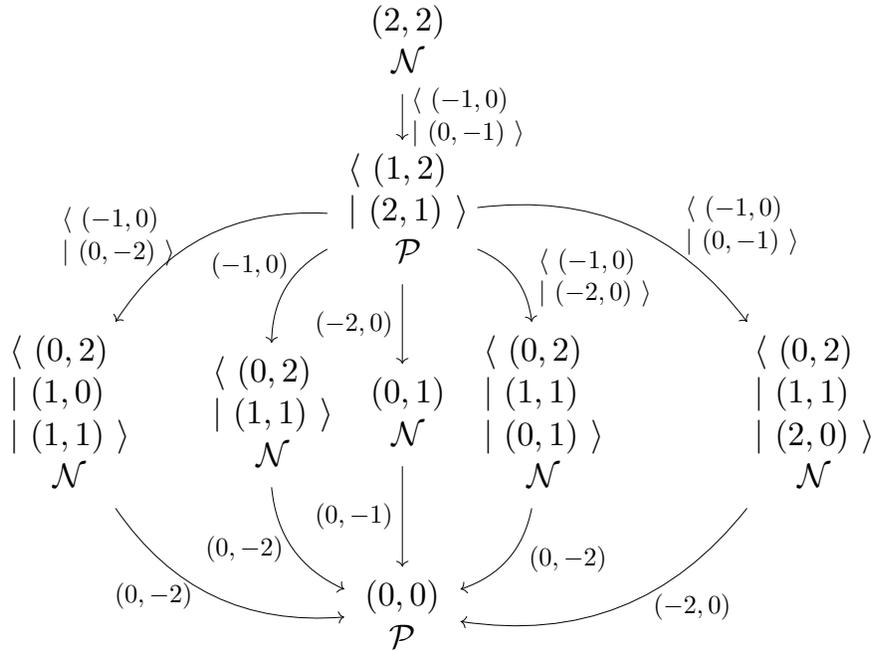
\begin{figure}[ht]
\begin{center}
  \def\scaleAmt{1.3}
  \scalebox{.9}{
  \begin{tikzpicture} [node distance = 1cm]
    \node (start) at (0, .5)  {\scalebox{\scaleAmt}{
        \begin{tabular}{@{}c@{}}
            $(2,2)$ \\
            \multicolumn{1}{c}{$\outcomeClass{N}$}
        \end{tabular}}};
    \node (first) at (0, -2) {\scalebox{\scaleAmt}{
        \begin{tabular}{@{}l@{}c@{}}
            $\langle\ (1,2)$\\ $|\ (2,1)\ \rangle$ \\
            \multicolumn{1}{c}{$\outcomeClass{P}$}
        \end{tabular}}};
    \node (secondA) at (-5, -5) {\scalebox{\scaleAmt}{
        \begin{tabular}{@{}l@{}l@{}c@{}}
            $\langle\ (0,2)$\\ $ |\ (1, 0)$\\ $|\ (1,1)\ \rangle$ \\
            \multicolumn{1}{c}{$\outcomeClass{N}$}
        \end{tabular}}};
    \node (secondB) at (-2, -5) {\scalebox{\scaleAmt}{
        \begin{tabular}{@{}l@{}c@{}}
            $\langle\ (0,2)$\\ $|\ (1,1)\ \rangle$\\
            \multicolumn{1}{c}{$\outcomeClass{N}$}
        \end{tabular}}};
    \node (secondC) at (0, -5) {\scalebox{\scaleAmt}{
        \begin{tabular}{@{}c@{}}
            $(0, 1)$ \\ $\outcomeClass{N}$
        \end{tabular}}};
    \node (secondD) at (2, -5) {\scalebox{\scaleAmt}{
        \begin{tabular}{@{}l@{}l@{}c@{}}
            $\langle\ (0,2)$\\ $|\ (1,1)$\\ $|\ (0,1)\ \rangle$ \\
            \multicolumn{1}{c}{$\outcomeClass{N}$}
        \end{tabular}}};
    \node (secondE) at (6, -5) {\scalebox{\scaleAmt}{
        \begin{tabular}{@{}l@{}l@{}c@{}}
            $\langle\ (0,2)$\\ $|\ (1,1)$\\ $|\ (2,0)\ \rangle$ \\
            \multicolumn{1}{c}{$\outcomeClass{N}$}
        \end{tabular}}};
    \node (zero) at (0, -8) {\scalebox{\scaleAmt}{
        \begin{tabular}{@{}c@{}}
            $(0,0)$\\ $\outcomeClass{P}$
        \end{tabular} } };

    \path[->]
        (start) edge [] node [right, text width=1cm] {$\langle\ (-1, 0)$\\ $|\ (0, -1)\ \rangle$} (first)
        (first) edge [bend right] node [left, text width = 1cm, xshift = -28pt] {$\langle\ (-1, 0)$\\ $|\ (0, -2)\ \rangle$} (secondA)
        (first) edge [bend right] node [above, xshift = -15pt] {$(-1, 0)$} (secondB)
        (first) edge [] node [left] {$(-2, 0)$} (secondC)
        (first) edge [bend left] node [right, text width = 1cm, xshift = 6pt] {$\langle\ (-1, 0)$\\ $|\ (-2, 0)\ \rangle$} (secondD)
        (first) edge [bend left] node [right, xshift = 20pt, text width = 1cm] {$\langle\ (-1,0)$\\ $|\ (0, -1)\ \rangle$} (secondE)
        (secondA) edge [bend right] node [left, xshift = -5pt] {$(0,-2)$} (zero)
        (secondB) edge [bend right] node [left] {$(0,-2)$} (zero)
        (secondC) edge [] node [left] {$(0, -1)$} (zero)
        (secondD) edge [bend left] node [right, xshift = 5pt] {$(0, -2)$} (zero)
        (secondE) edge [bend left] node [right, xshift = 10pt] {$(-2, 0)$} (zero)
        ;

  \end{tikzpicture}}
\end{center}
  \caption{Winning strategy for Next player in  \ruleset{Quantum Nim} $(2,2)$, showing that quantum moves impact the game's outcome.  (There are four additional move options from $\braket{(1,2)\ |\ (2,1)}$ that are not shown because they are symmetric to moves given.)}
  \label{fig:nim22-intro}
\end{figure}

From the \ruleset{Quantum Nim} position $(2,2)$, in addition to the classical moves of picking one or two pebbles from a pile, the current play can consider a quantum move formed by the {\em superposition} of taking one pebble from either pile.
We denote this superposition of classical moves by:
$$\langle (-1,0)\ |\  (0,-1) \rangle.$$
This {\em quantum move} creates a {\em quantum game position} consisting of a superposition of two \ruleset{Nim} positions
$$\langle (1,2)\ | \ (2,1)\rangle.$$
Subsequent quantum moves may further increase the number of realizations in the game's superposition state.
For example, if the other player makes the following
the quantum move:
$$\langle (-1,0)\ |\  (0,-2) \rangle $$
then the next game position will be a superposition of three classical positions:
$$ \langle (0,2)\ | \ (1,0)\ |\ (1,1)\rangle.$$
(The position $(2, -1)$ is not a legal \ruleset{Nim} position, so it is not included in the resulting superposition.)

Remarkably, the quantum move $\langle (-1,0)\ |\  (0,-1) \rangle$
 is in fact a winning move on \ruleset{Quantum Nim} at $(2,2)$.
For the complete game tree of \ruleset{Quantum Nim} after that move from $(2,2)$, see Figure \ref{fig:nim22-intro}.
Thus, quantum moves not only enrich players' strategical domains, but also can alter the outcome of combinatorial games.
The sophisticated interactions (between superpositions of moves and superpositions of
pure classical game positions) and the potential explosiveness (in the complexity of quantum configurations) make quantum games
a fascinating
for computational complexity studies.

\vspace{0.1in}
\noindent {\sc Quantum Flavors}:
Dorbec and Mhalla\cite{dorbec2017toward} postulate five quantum variants for extending classical rulesets,
regarding possible restrictions on
the interaction between classical moves and quantum positions.
We now list them from least to most restrictive.

\begin{itemize}
\item In variant D, there are no additional restrictions.
In other words, at each turn, a player can play
  either a classical or quantum move.

\item In variant C', a player can make a classical move only if it would cause no non-terminal realization to collapse.

\item In variant C, a player can make a classical move only if it would cause no realization to collapse.

\item In variant B, a player can make a classical move only if there are no valid quantum moves available
(i.e., superposition of moves in which each is feasible for some realization(s)).

\item In variant A, a player can never make classical moves.
\end{itemize}


In this paper, to emphasize their structural
 implication to quantum rulesets, we will refer to these variants interchangeably as {\em quantum flavors}, {\em flavors of quantum extension}, or simply {flavors}.
In the definitions above, quantum flavors D and A are straightforward
  and B at its core is an
  extension of the \ruleset{Tic Tac Toe} ruleset brought over to other combinatorial games.
Notice that
flavor C' is essentially combining the additional relaxations into restrictions from B and C.

We will primarily be focusing on quantum flavor D,
  in which classical moves can be viewed as special cases of superposition moves.
This variant is perhaps the most natural quantum extension of combinatorial games.
Most of our proofs can be applied or adapted to other flavors.
Our complexity-theoretical studies of quantum  combinatorial games also illustrate some subtle differences of these quantum flavors.

\subsection{Highlights and Organization of the Paper}



\vspace{0.1in}
\noindent Combinatorial games are fun to play.
They can  also be intellectually challenging to play optimally.
Part of the reason behind the challenge---from computational  perspectives---is that the basic underlying {\em decision problem} for determining the outcome of many games
is {\em intractable}.
For various games,  unless $\cclass{PSPACE} = \cclass{P}$---a remarkable refuting of the conventional wisdom in the theory of computing\footnote{\cclass{PSPACE} represents the set of all decision problems that can be solved in polynomial space; \cclass{P} represents the set of all decision problems solvable in polynomial time.}---no algorithm can always determine the game's outcome, from an arbitrary game position, in time polynomial in the descriptive size of these games \cite{DBLP:journals/jcss/Schaefer78,PapadimitriouBook:1994,LichtensteinSipser:1980,EvenTarjanHex,Reisch:1981,Burke2008,10.5555/1714301}.
The rapid exponential growth in decision time
makes search for winning strategies challenging
even when the game boards are with moderate-sizes, such as 19 by 19 for \ruleset{Go} or 14 by 14 for Nash's ``optimal'' \ruleset{Hex} size \cite{NashHex}.

Thus, the computational
  intractability of {\sc Decision of Outcome}
 and the relevant search problem of winning moves by querying game positions (representing possible intermediate game configurations in the processing of playing the games) provides a strong indicator of the richness and intricacy of combinatorial games.
{\em Elegant} combinatorial games with
  {\em simple rules} yet
  {\em intractable complexity} are
  the gold standard for combinatorial game design \cite{Eppstein}. 
 Particularly, in the digital age
  powered by AI,
  the lack of an efficient algorithm for optimally playing a combingatorial
  game gives human players a fighting
   chance to compete against computer programs,
   and, in the meantime,
   the deep challenge also gives
   computer programs
   a reason to learn and improve \cite{Burke2008,10.5555/1714301}.

While many popular combinatorial games---such as \ruleset{Hex}, \ruleset{Node-Kayles}, \ruleset{Geography}, and \ruleset{Go}---are \cclass{PSPACE}-intractable \cite{DBLP:journals/jcss/Schaefer78,PapadimitriouBook:1994,LichtensteinSipser:1980,EvenTarjanHex,Reisch:1981},
some seemingly nontrivial games are
 polynomial-time solvable (with elegant math and algorithms).
Two examples are
  are \ruleset{Nim} \cite{Bouton:1901} and
  \ruleset{Undirected Geography} \cite{DBLP:journals/tcs/FraenkelSU93}.\footnote{ \ruleset{Geography} is an impartial game defined by a directed graph and a starting vertex in the graph. During the game, two players take turns to advance a path in the graph from the starting vertex.
In the normal-play version, the player who cannot not make a further advancement---i.e., all out-neighbors of the current vertex are occupied---loses the game.
\ruleset{Undirected Geography} is a special family of \ruleset{Geography} where the underlying graphs are undirected.
Determining the outcome of \ruleset{Geography} on general directed graphs  is \cclass{PSPACE}-complete \cite{LichtensteinSipser:1980}.
However, cycles of directed edges play a crucial role in the gadgets for establishing the \cclass{PSPACE}-hardness.
Indeed, both \ruleset{Undirected Geography} and
\ruleset{Directed Acyclic Geography}---i.e., \ruleset{Geography} on DAGs---are polynomial-time solvable games.\cite{{DBLP:journals/tcs/FraenkelSU93}}}
For both combinatorial games, the underlying {\sc Decision of Outcome} and {\sc Search for Winning Moves} have polynomial-time algorithmic solutions in terms of their descriptive sizes (i.e., respectively, the number of bits specifying the starting piles in \ruleset{Nim} and the number of nodes of the
 graphs defining \ruleset{Undirected Geography}).

Our research program in this paper started with and expanded upon the following 
question. 

\begin{itemize}
\item {\bf Quantum Leap in Complexity} {\em Is there a polynomial-time solvable combinatorial game whose quantum extension is intractable?}
\end{itemize}

In our pursuit of this possible leap in the {\em quantumized  complexity} of combinatorial games, we have also obtained an affirmative answer to the following fundamental question:

\begin{itemize}
\item {\bf Quantum Collapses in Complexity}: {\em Are there \cclass{PSPACE}-complete combinatorial games whose quantum extensions fall to the lower levels of polynomial hierarchy?}
\end{itemize}

For the first question, we have focused on two concrete games for identifying potential leap in the quantumized complexity:
\begin{enumerate}
\item Is \ruleset{Quantum Nim} polynomial-time solvable?
\item Is \ruleset{Quantum Undirected Geography} polynomial-time solvable?
\end{enumerate}

Understanding these two important, classically polynomial-time solvable combinatorial games in the quantum setting has been instrumental to our research.
The process and findings have shed light on the expressive power of superposition of moves and quantum game positions.
The concrete structures of their game dynamics have provided us with much-needed 
insight into the otherwise-abstract mechanisms of deep alternation in quantum combinatorial games.

\ruleset{Nim} and \ruleset{Undirected Geography}  are polynomial-time solvable in the classical setting for different mathematical reasons.
The solution for \ruleset{Nim}
  uses a polynomial-time computable algebraic expression for the outcome and winning strategy \cite{Bouton:1901}.
On the other hand, the solution for \ruleset{Undirected Geography} relies on a matching theory for characterizing the outcome and winning strategy \cite{DBLP:journals/tcs/FraenkelSU93}.
Other than their apparent structural differences---one based on integer vectors and another based on graph structures---these two games also differ in one fundamental aspect:
The game tree of \ruleset{Undirected Geography} is {\em polynomially short}, while the game tree for \ruleset{Nim} is {\em exponentially tall} in their respective descriptive complexity.
We will prove later in Section  \ref{Sec:PSpaceUpperBound} that the height of game trees can have a crucial impact on the
quantumized complexity.

\vspace{0.1in}
\noindent{\sc A Taxonomy of Quantumized Complexity}:
We prove that both \ruleset{Nim} and \ruleset{Undirected Geography} make a {\em\bf complexity leap} over \cclass{NP} in the {\bf\em quantum world}, by showing that 
{\sc Decision of Outcome}
for both quantum games at a querying superposition becomes intractable.

We have been drawn further into these quantum regions because of their intriguing computational complexity landscape.
Due to fact that the number of realizations in the resulting quantum superposition could be
 exponential in number of quantum moves
 leading up to the superposition,
  one could face, in the {\sc Decision of Outcome} for playing quantum games, configurations far more complex than in the classical setting.
This potential explosiveness in game configurations---although encodable with polynomially succinct representations---contributes to the mathematical and computational challenges in quantum combinatorial game theory.
It also introduces significant difficulty
 in practical implementation and visualization of quantum games.
Thus, quantum games come with their own character beyond classical games, introducing naturally multifaceted complexity measures. 

 In Section \ref{Sec:Taxonomy},
  we introduce a taxonomy of complexity measures for quantum combinatorial games, which uses
  the ``degree of quantumness'' of the querying
  superposition as
  a natural parameter for
  the classification of quantumized complexity.
As a first step, we attempt to capture the behaviors of quantum games with several baseline quantumized complexity landmarks, measuring the worst-case complexity for determining the outcome when allowing quamtum moves from:
\begin{itemize}
    \item a classical position,
    \item a superposition with a polynomial number of realizations (shorthand: {\em poly-wide quantum position})
    \item a superposition obtained from a classical position after making a given sequence of polynomial number of quantum/classical moves
 (shorthand: \emph{poly-moves after classical start}),
    \item a poly-wide superposition reachable from a classical position after a sequence of quantum/classical moves (shorthand: \emph{reachable poly-wide quantum position from a classical start}),
 and
    \item a superposition obtained from a poly-wide
 quantum position after making a given sequence of poly-moves
 (shorthand: {\em poly-moves after poly-wide quantum start}).
 \end{itemize}


\vspace{0.1in}
\noindent{\sc Exploring Quantum Undirected Geography}:
Of the two quantum games, we have obtained a more detailed picture for \ruleset{Quantum Undirected Geography}.
Our contrasting yet complementing
 characterization 
of its solvability reveals a dramatic complexity
landscape.

First, we prove  that
  it is \cclass{PSPACE}-complete to determine the game outcome of
  \ruleset{Quantum Undirected Geography}
  at a poly-wide quantum position.
We also show that this \cclass{PSPACE}-completeness result continues to hold for
quantum positions created by poly-moves after a classical start.
Both of these proofs (in Section \ref{Sec:PSPACECompleteQUG}) reduce from \ruleset{Quantum (Directed) Geography}.  The proofs construct gadgets to simulate directed edges by utilizing superpositions of multiple realizations and their dynamic interaction with quantum/classical moves.
In Section \ref{Sec:PSPACEComplete},
we reduce the \cclass{PSPACE}-complete \ruleset{Geography}  to \ruleset{Quantum Geography} by exploiting graph localities to ``regulate'' the impact of quantum moves.

We then complement these complexity results
  with an algorithmic one, highlighting the delicacy of this quantum combinatorial game.
In Section \ref{Sec:PClassicalStartQUG},
we design a polynomial-time algorithm
  for determining whether \ruleset{Quantum Undirected Geography} with a classical starting position is winnable.
The mathematical statement of our result is in fact simple: In quantum flavor $D$, the classical winner (called them `Hero') remains the quantum winner,
 provided the Hero is vigilant.
In addition, the Hero can always find a winning classical move, in polynomial-time, whether or not the classical loser (called them the 'Villain') makes quantum moves.
Perhaps unexpectedly, the proof of correctness requires careful graph theory and algorithm design.
Our proof extends the matching theory for classical \ruleset{Undirected Geography} \cite{DBLP:journals/tcs/FraenkelSU93}, through
{\em realization-collapsing moves}
and {\em graph contraction}, to find the Hero's winning move at any intermediate position to fend off the Villain's quantum attacks.
As long as the Hero plays judiciously with consistency from the start, they can keep the
Villain in check. 
This algorithmic result provides a remarkable contrast to the earlier mentioned complexity result that
\ruleset{Quantum Undirected Geography} at quantum position reachable from a classical position after poly-moves is \cclass{PSPACE}-hard.\footnote{For the winner of a classical position in \ruleset{Quantum Undirected Geography}, it is literally, ``if they snooze, not only could they lose the game, but they may also not able to efficiently determine whether they would lose or could win after certain number of less-strategically steps have been made.''}


\vspace{0.1in}
\noindent{\sc Quantum Leap in Nim Complexity}:
In Section \ref{Sec:Nim},
we prove that,
deciding the outcome of
\ruleset{Quantum Nim} at a given {\em poly-wide} quantum position is at least $\Sigma_2$-hard.
In other words, determining whether a poly-wide superposition is winnable in \ruleset{Quantum Nim} is above the second level of the polynomial-time hierarchy.

Note that \ruleset{Quantum Nim} remains an impartial game.
Thus, by the celebrated Sprague-Grundy theorem \cite{Sprague:1936,Grundy:1939}, \ruleset{Quantum Nim} at poly-wide quantum position can be mathematically characterized by classical \ruleset{Nim} (i.e., using ``nimbers'', or ``Grundy values'').
Thus, our result demonstrates that
any such ``Sprague-Grundy-reduction'' from
\ruleset{Quantum Nim} to the classical \ruleset{Nim} is not likely to be  computationally efficient.

Our proof takes several turns, each of which provides us with new insight.
First, in Section
\ref{Sec:QATToQNim},
by exploiting the expressive power of
superpositions of \ruleset{Nim} positions,
we design an efficient encoding of \ruleset{Quantum Avoid True}---the quantum extension of Schaefer's impartial logic game,
\ruleset{Avoid True} \cite{DBLP:journals/jcss/Schaefer78}\footnote{Each instance of \ruleset{Avoid True} is defined by a positive CNF, in which two players take turns setting Boolean variables to \texttt{true}, provided that they can avoid making the overall CNF evaluate to \texttt{true}.}---by \ruleset{Quantum Nim}
at poly-wide superposition.
This encoding represents a significant step in using superpositions of \ruleset{Nim} positions---individually being polynomial-time solvable—to characterize intractable combinatorial game positions.
In the classical setting, because \ruleset{Avoid True} is \cclass{PSPACE}-complete  \cite{DBLP:journals/jcss/Schaefer78} and
\ruleset{Nim} is polynomial-time solvable, any reduction from \ruleset{Avoid True} to \ruleset{Nim}  must incur an exponential explosion (unless \cclass{PSPACE} = \cclass{P}). In contrast, our reduction efficiently characterizes each \ruleset{Quantum Avoid True} position by a poly-wide superposition of \ruleset{Nim} positions, each with a polynomial descriptive size.
These  \ruleset{Nim} positions are individually simple; each is an instance of \ruleset{Boolean Nim}, consisting of piles with only one or zero pebbles.

Therefore, our polynomial-time reduction from \ruleset{Quantum Avoid True} to \ruleset{Quantum Nim} sheds light sharply on the structural power of quantum moves and game positions. Our proof, in particular, illustrates the capacity of \ruleset{Quantum Nim} positions to succinctly characterize the complex logical alternations intrinsic to \ruleset{Quantum Avoid True}.

Second, because our reduction is not from \ruleset{Avoid True} to \ruleset{Quantum Nim},\footnote{The reduction can encode every \ruleset{Avoid True} position
 with a poly-wide
 \ruleset{Quantum Nim} positions, and
if the \ruleset{Quantum Nim} can somehow ``forbids''---figuratively and mathematically---
 quantum moves from this position onwards,
then the ``classical continuation'' of
  the \ruleset{Quantum Nim} position
  precisely captures the alternation
  in \ruleset{Avoid True}.
  However,  quantumness does matter in the resulting \ruleset{Quantum Nim}.
Hence, instead of preserving \ruleset{Avoid True}'s alternation structure,
  the reduction preserves
  \ruleset{Quantum Avoid True}'s game tree.
  } but rather, and critically, from \ruleset{Quantum Avoid True} to
\ruleset{Quantum Nim},
we have to establish intractability
of \ruleset{Quamtum Avoid True}.
The need to understand the computational complexity of \ruleset{Quantum Avoid True} in order to determine the intractability of \ruleset{Quantum Nim} also brings us to a family of fundamental questions, centered around
an intuitive conjecture:

\begin{quote}
{\em The  quantum lift of any combinatorial game is always at least as hard  (computationally) to play as the combinatorial game itself}.
\end{quote}

or more concretely:
\begin{quote}
{\em Any \cclass{PSPACE}-complete combinatorial
    game remains a \cclass{PSPACE}-hard game in the quamtum setting.}
\end{quote}

Had these conjectures been true, we would have proved that ``\ruleset{Quantum Nim} is a \cclass{PSPACE}-hard game at poly-wide superposition.''
In an exciting turn of events, these questions together with Schaefer's brilliant reduction from \ruleset{Partition-Free QBF}\footnote{In this two-player game defined by a QBF, the goal of the True player is to make the QBF \texttt{true} and the goal of False is to make the QBF \texttt{false}.
Boolean variables of the QBF are {\em partitioned} into two sets, one for each player. At their turn, players can ``freely'' set any of their unassigned variables.} to \ruleset{Avoid True} \cite{DBLP:journals/jcss/Schaefer78} has led us to realize that various variants of \cclass{PSPACE}-complete \ruleset{QBF} games experience {\em collapses in quantumized complexity} to the second level of polynomial-time hierarchy.
We will have more discussion on \cclass{PSPACE}-complete games with  ``quantum collapse'' shortly.

The reduction from partizan  \ruleset{Partition-Free QBF} to impartial \ruleset{Avoid True} \cite{DBLP:journals/jcss/Schaefer78}
illuminates promising pathways for us
after our \ruleset{Quantum Nim} encoding of
\ruleset{Quantum Avoid True}.
In Section \ref{Sec:QuantumSchaefer}, we lift Schaefer's reduction to show that \ruleset{Quantum Partition-Free QBF} can be efficiently encoded by \ruleset{Quantum Avoid True}.
This confidence-boosting proof---one of the two main technical results of Section \ref{Sec:Nim}---crucially uses the very mathematical properties that cause the ``quantum collapses'' in \ruleset{Quantum Partition-Free QBF}:
\begin{itemize}
\item {\bf Player True's Winning Freedom}: If \ruleset{Quantum Partition-Free QBF} is winnable by the True player then they have "complete freedom", meaning they can ignore their False's choices and win by making their own classical moves.
\item {\bf Player False's Winning Freedom}: If \ruleset{Quantum Partition-Free QBF} is winnable by the False player, then they have the ``liberty'' 
to just make quantum assignments of $\langle \texttt{true}\ |\  \texttt{false}) \rangle$ for their variables, in any order they like, regardless of what True does.
\end{itemize}
However, while quantumness introduces more freedom to the winning design for \ruleset{Quantum Partition-Free QBF} due to the partizan nature of the game,
it also provides the losing player
wide latitude and power to interfere
the winning player's action in the impartial
\ruleset{Quantum Avoid True}.
Thus, the proof of reduction requires a delicate argument to counter this ``unwelcome'' quantum impact.

The intractability of \ruleset{Quantum Nim} then follows from the result in Section \ref{Sec:QuantumCollapses}, showing that although quantum moves collapse the \cclass{PSPACE}-hardness of \ruleset{Partition-Free QBF}, the
\ruleset{Quantum Partition-Free QBF} retains a certain degree of intractability at the second level of the polynomial-time hierarchy.

 In our reduction,
our \ruleset{Quantum Nim} game
uses poly-wide superposition of perhaps the simplest \ruleset{Nim} positions:
All of the \ruleset{Nim} positions in our encoding are from the family of \ruleset{Boolean Nim}.
So, our $\Sigma_2$-hard intractability proof
of \ruleset{Quantum Nim} holds for \ruleset{Quantum Boolean Nim}
and can be extended broadly to
the quantum extension of
\ruleset{Subtraction} games, a family of
arithmetic strategy games widely used to teach young children about mathematical thinking.\footnote{For example, Los Angeles Math Circle summer school assigns subtraction games with subtraction set $\{1,2\}$ to first-grade students.}
In Section \ref{Sec:NimEncoding}, we also formulate a general ``Robust Binary-Nim Encoding" properties sufficient for supporting
a reduction from \ruleset{Quantum Nim} superposition, implying that evaluating
poly-wide superposition in
several natural games
including  \ruleset{Brussel Sprouts}, \ruleset{Go} and \ruleset{Domineering} is $\Sigma_2$-hard in the quantum setting.



\vspace{0.1in}
\noindent{\sc Quantumness Matters}:
The quantum leap in the computational complexity of \ruleset{Nim} and \ruleset{Undirected Geography} as well as quantum collapse in the intractability of \ruleset{QBF} further 
 illuminate the rich impact of
 quantum moves to the structure
 of combinatorial games.
In Section \ref{Sec:QuantumMatter},
  we formalize two notions of
  ``{\bf\em quantumness matters}''
  to capture whether or not having quantum
  options alters optimal strategies
  or game outcomes.
We illustrate these concepts with concrete examples, 
some of which are later used as gadgets 
for complexity-theoretical analysis.
These examples also serve as illustrations for the formal notation of the {\bf\em quantum lift of combinatorial games} introduced in Section \ref{Sec:Notations}.
On the other hand, our proof establishing the polynomial-time solvability of  \ruleset{Quantum Undirected Geography} from a classical start
illustrates some fundamental aspects of when ``quantumness doesn't matter.''
One of the conditions that we have identified in Section \ref{Sec:QuantumMatter}
also contributes instrumentally to our proof extending Schaefer's  \ruleset{Partition-Free QBF} to  \ruleset{Avoid True} reduction to  the quantum setting.

\vspace{0.1in}
\noindent{\sc Quantum Transformation of \cclass{PSPACE}-Complete Games: Complexity Preservation and Collapse}:
In Section \ref{sec:Reductions},
we study the quantumized complexity of various classical \cclass{PSPACE}-complete combinatorial games.

In Section \ref{Sec:QuantumCollapses},
 focusing on variants of combinatorial games built on the complexity-theoretical
 Quantified Boolean formula (QBF) problem, we identify a family of classical \cclass{PSPACE}-complete games subject to
  {\em quantum collapse}, namely, the complexity of their quantum lift falls
  from \cclass{PSPACE}-complete to
  the second level in the polynomial-time
  hierarchy.
While the proof is extremely simple, this result illustrates the subtle impact of quantum moves to the structure and complexity of combinatorial games.
The result provides a complexity-theoretical refutation to conjecture that the quantum extension of combinatorial games is always at least as hard (computationally) to play as the original ones.

In fact, these QBF-based games have richer stories in the quantum world, revealing the sensitive
interactions between quantum moves and logic positions.
Combinatorial games commonly use ``normal play'', meaning that a player loses {\em if and only if} they are unable to make a move on one of their turns.
There are several ways to transform QBF-games into normal play games by introducing new game moves for the ``end-of-QBF'' condition.
All of these are classically equivalent:
They are polynomial-time reducible to each other, and
\cclass{PSPACE}-complete.

However, our analysis exhibits the fundamental difference among various
\ruleset{QBF} variants in the quantum setting, illustrating
they are dependent
  on the explicit move definitions
  in the ruleset.
We also show that these differences have complexity consequences.
In Section \ref{Sec:QuantumCollapses}, we discuss
  various \ruleset{QBF} variants.
While we give intractability results
  for all these natural \ruleset {Quantum QBF} variants, some  ``end-of-QBF'' transformations
  contribute to the quantum collapse in complexity.  This is critical to our intractability analysis of \ruleset{Quantum Nim}, and provides the needed technical reduction to complete the proof of the leap in \ruleset{Nim}'s quantumized complexity from \cclass{P} to the second level of the polynomial-time hierarchy.

In Section \ref{Sec:Depth},
 motivated by celebrated complexity characterizations of Lander's  \cclass{NP}-intermediate problems \cite{Lander} and Ko's intricate, meticulous separation of levels of
the polynomial-time hierarchy \cite{KoPHSeparation},\footnote{In his famous result in computational complexity,
Ko proved the following:
For any integer $k > 0$, there exists an oracle $O_k$ such that
{\em relative to} $O_k$, the polynomial-time hierarchy collapses to exactly
to level $k$, that is,
{\em relative to} $O_k$, \cclass{PSPACE} is equal to level $k$ of the polynomial-time hierarchy, while level  $(k-1)$  of the polynomial-time hierarchy
is strictly contained in level $k$ of the polynomial-time hierarchy.}
we further extend our proof to show that for any integer $k>0$, there exists a \cclass{PSPACE}-complete game whose quantum extension is complete exactly for level $k$ in the polynomial-time hierarchy.

We then show, in Section \ref{Sec:PSPACEComplete}, that several natural \cclass{PSPACE}-complete combinatorial games---including
\ruleset{Geograhy}, \ruleset{Node-Kayles}, and \ruleset{Snort}---preserve their \cclass{PSPACE}-hardness in the quantum setting.
Unlike \ruleset{Avoid True}, these games enjoy a great deal of ``locality'' instrumental to
our gadgets  to limit the impact of quantum moves to set up reduction from the classical versions to their quantum counterparts.

\vspace{0.1in}
\noindent{\sc A Quantum Leap Complexity Barrier}:
Our analysis of quantumness matters and discoveries of quantum leaps, collapses, and preservation in the complexity of combinatorial games also lead us to the following basic question:

\begin{itemize}
\item {\bf Quantum Leap over PSPACE?}:
{\em Is there a \cclass{PSPACE}-solvable combinatorial game whose
    quantumized complexity jumps  to above \cclass{PSPACE}}?
\end{itemize}

In order to identify such games, we have initially taken
  an intuitive path by considering, as potential candidates, some
  of our favorite \cclass{PSPACE}-complete games such as \ruleset{QSAT},
  \ruleset{Avoid true}, \ruleset{Hex}, \ruleset{Go},  \ruleset{Geography}, \ruleset{Node Kayles}, \ruleset{Atropos},
{\em etc}, {\em etc}.
Upon examining this diverse set of games built on logic, geometric, graph-theoretical, and topological formulations, we use a common feature:
unlike \ruleset{Nim}, these games are all {\em polynomially-short}; the height of their game trees is polynomially bounded by their descriptive complexity.
We also observe that
  while quantum moves potentially increase the complexity of combinatorial games---particularly,  they can introduce game positions with an exponential number of realizations---quantum lift does not increase the height of game trees.

Preservation of game-tree height turns out to be a key limitation of quantum moves with crucial computational consequences.
Through a careful algorithmic treatment of superpositions with a potentially exponential number of realizations, we show
in Section \ref{Sec:PSpaceUpperBound}
that there is a \cclass{PSPACE} barrier limiting quantum leaps of  all polynomially-short games.
In other words, for a large family of combinatorial games, including all \cclass{PSPACE}-complete games mentioned above, the strengthened quantum power in moves---although capable of affecting the game outcomes and creating quantum game boards with superposition of exponentially many realizations---does not elevate the complexity of the combinatorial games beyond \cclass{PSPACE}.

\vspace{0.1in}
\noindent{\sc Extensions, Conjectures, and Open Questions}:
During our exploration of quantum combinatorial games, we have encountered   far more questions than we can answer.
Even seemingly simple quantum games, such as \ruleset{Quantum Trilean Nim} with classical start, in which each pile has at most two pebbles (i.e., three values: 0, 1, 2) remain unresolved.
Even the basic question of
whether there is a
polynomial-time algorithm to
solve \ruleset{Quamtum Undirected Geography} at a superposition with two realizations has escaped us.

Thus, we identify several elegant open questions fundamental to quantum combinatorial games.
The most outstanding question, perhaps, is the
complexity of \ruleset{Quantum Nim}.
As one of the very first quantum games that we have studied, the complexity of Quantum Nim remains elusive for precise characterization.
While \ruleset{Nim} is a polynomial-time solvable game, it is not a polynomially-short game.
Thus, the  \cclass{PSPACE}-barrier that we establish for quantumized polynomially-short games, does not directly apply to \ruleset{Quantum Nim}, even with a classical start.
On the other hand, while we are able to show that there is a quantum leap of \ruleset{Nim}'s complexity over \cclass{NP} to level two of polynomial-time hierarchy,
the magnitude of quantum leap, measured either by  levels of polynomial-time hierarchy or
by its degree above \cclass{PSPACE}, and the dependency on the ``degree of quantumness'' in game's starting position,
remains unknown.
Thus, for its degree of intractability,
  we still don't know if it is an \cclass{PSPACE}-hard game.
On the other hand, for its upper bound, we still don't know if \ruleset{Quantum Nim}  is a \cclass{PSPACE}-solvable game, although---by generating its game tree---we can show that it is an \cclass{EXPSPACE}-solvable game.
To put it simply, the complexity of \ruleset{Quantum Nim} is wide open.

Our characterization of quantumized complexity of polynomially-short games now places \ruleset{Quantum Nim} as a potentially viable candidate for quantum lift over \cclass{PSPACE}.
Other candiates are various combinatorial games
 based on structures based on polynomial-local search (PLS).
In general, we will benefit from a better understanding of
the connection between quantumized complexity and the height of game trees, and, in particular, the impact of the degree-of-succinctness of a game’s representation—which is often the source of superpolynomial game-tree height—to both classical and quantumized complexity.

Another concrete open question is whether a stronger
 dichotomy characterization holds for
 \ruleset{Undirected Geography}, that is,
 whether or not
 it is intractable to determine the
 game outcome of
 \ruleset{Quantum Undirected Geography}
  at a given superposition
  with a {\em constant number} of realizations
  as well as at poly-wide superpositions reachable from classical positions.
In general, it is a wonderful research question to establish a sharper mathematical characterization on
the dependency of a game’s structure and complexity on the degree-of-quantumness in starting positions and reachability from classical positions.

In Section \ref{Sec:FinalRemarks},
 in addition to sharing some of our conjectures,
 we will also discuss  several potential extensions of quantum combinatorial games.  These extensions may be more manageable for human play.
It is our desire to
 set up some elegant quantum games for us and more people to explore.

\section{Combinatorial Games and Their Quantum Lift}\label{Sec:Notations}

In this section, we will first review some mathematical backgrounds of combinatorial game theory (CGT).
Then, building on these classical formulations, we introduce basic notation and concepts in quantum combinatorial game theory (QCGT) essential
to the presentation of this paper.

\def\cp{h}

\subsection{Classical Combinatorial Games: Abstract Rulesets and Concrete Instances}

Combinatorial games are {\em succinct abstract rulesets} for {\em deep interactions}, which define the domains of game positions and transitional moves from position to position
\cite{WinningWays:2001,SiegelCGT:2013}.
Some games, such as \ruleset{Go} \cite{LichtensteinSipser:1980}, \ruleset{Chess} \cite{DBLP:journals/jct/FraenkelL81}, and \ruleset{Checkers}, are well played in practice; some games, such as \ruleset{Nim} and \ruleset{Mancala}, are taught in Math Circles/summer camps to primary school students; while many others, including theorized practical games and abstract formulations based on logic, geometry, and graph theory,
 such as \ruleset{QSAT}, \ruleset{Avoid True}, \ruleset{Geography}, \ruleset{Nodel-Kayles},  \ruleset{Hex}, \ruleset{Atropos} {\em etc, etc}, are studied in mathematics and computer science.

In practice,  board games such as \ruleset{Go} and \ruleset{Hex} may often have a fixed scale and starting position.
Mathematically (in CGT), each game  name---including well-known ones---usually 
represents an infinite family of concrete game instances, each of which is {\em instantiated} by a {\em starting game position} from which players can take turns making transitional decisions according to the same general {\em ruleset} of the game.
 For structural and computational studies, the game instances can be further classified
into subfamilies,  parameterized either by  a natural dimension characterizing games' {\em scales}\footnote{i.e., size of the board in \ruleset{Hex} and \ruleset{Go} or number of piles in \ruleset{Nim}.}  or by a natural ``game-board frame'' formulating the sub-domains of positions/moves.\footnote{i.e. the graph on which \ruleset{Node-Kayles} and \ruleset{Geography} is played.}
Each fixed parameterization defines a sub-domain of game positions, alphabet of moves, and transitional rules for its subfamily of game instances (for a focused analysis).
In this regard, practically popular games, such as the (19 by 19) \ruleset{Go} and (Nash's 14 by 14) \ruleset{Hex}, are game instances with a fixed starting position (i.e.,  empty boards in these cases) and from a particular subfamily (i.e., respectively, of dimension 19 for the former and 14 for the latter) under the umbrella of the mathematical
rulesets \ruleset{Go} and \ruleset{Hex}, each of which represents infinity many instances definable over any dimension and starting  position (i.e., partially filled boards).

Formally, the ruleset of a combinatorial game can be defined by four components,
denoted by
${\cal R} = ({\cal B}, {\cal S}, \rho_L, \rho_R)$, where:
\begin{itemize}
 \item ${\cal B}$ is a set of all possible game positions.
    \item ${\cal S}$ is an alphabet of possible moves for players over positions in ${\cal B}$, \footnote{Combinatorial Game Theory experts may wonder why a separate alphabet is needed to describe moves, instead of just using states from ${\cal B}$ that are the results of those moves.
The authors posit that the moves are more natural in the context of quantum moves.  Indeed, the choice of moves used in a ruleset determines how each move interacts with components in a quantum superposition.  Additionally, if we're using the resulting state as the move in quantum flavor $D$,
then based on the current formulation, any player can collapse a superposition on their turn whenever they could make any move, because they can choose the desired position directly.  This means that choosing a superposition gives the opponent more options.  Thus, it's never strategically better to choose a superposition and the game values of all classical positions are the same.  The use of move descriptors allows players to make classical moves that are legal on multiple realizations.  For further discussion on this topic, see Section \ref{Sec:QuantumLift}.
}
    \item $\rho_{L} : {\cal B} \times {\cal S} \rightarrow {\cal B} \cup \{NULL\}$ is the move function for (the player) Left,
    \item $\rho_{R} : {\cal B} \times {\cal S} \rightarrow {\cal B} \cup \{NULL\}$ is the move function for Right.
\end{itemize}

If $\rho_L = \rho_R$, then both players always have the same options.
We refer to these rulesets as \emph{impartial}.
Games that aren't impartial are known as \emph{partizan}.
When looking at an impartial game, we will  ignore the $L$ and $R$ identifiers and just refer to the function with the singular $\rho$. 

Given a ruleset ${\cal R}$,
 each game position $b_0\in {\cal B}$
 defines a (concrete) {\em game instance} (or a {\em game} for short) ${\cal R}[b_0]$, with
 $b_0$ as the starting position.
To play 
  ${\cal R}[b_0]$,
  two players $L$ and $R$ take turns
  in selecting their moves.
Each turn, the current player,
$\cp \in \{L, R\}$, chooses to make a move from the current position,
  $b$ (initially set to $b_0$).
The move they choose, $\sigma \in \Sigma$, must be \emph{feasible}, meaning that it must result in a non-$NULL$ value: $\rho_\cp(b, \sigma) \neq NULL$.
The results of these feasible moves for $\cp$ from $b$,
$$\{b'\ |\ \exists\ \sigma \in \Sigma \text{ and } \rho_\cp(b, \sigma) \neq NULL\  \}$$
are the \emph{options} of $\cp$ from $b$.
We will focus on the {\em normal-playing convention} in this paper.
For normal-playing games, the current player $\cp$ loses the game, if they have no feasible option at the current position.
We refer to such a position as a \emph{terminal} position for $\cp$.
Throughout the paper, we will use $\bar{\cp}$ to denote the other player, that is, if $\cp = L$ then $\bar{\cp} = R$ and if $\cp = R$ then $\bar{\cp} = L$.

Thus, we can apply the move function $\rho_\cp$ to
  determine the position
  after a sequence of moves from a given position $b\in B$.
We use the default, $\forall b\in B$, $\rho_\cp(b,\emptyset) := b$; and
inductively,  for any sequence of moves,
$s = \sigma_T \circ ... \circ \sigma_1$:
$$ \rho_\cp(b,s) := \rho_\cp(\rho_{\bar{\cp}}(b,\sigma_{T-1}\circ ... \circ \sigma_1),\sigma_T). $$

For each $b\in B$ and $h\in \{L,R\}$, define:
\begin{eqnarray}
\mbox{\rm Reachable}_h(b) := \{b'\ |\  \exists s \in \Sigma^*, b' = \rho_h(b,s)\}
\end{eqnarray}
$\mbox{\rm Reachable}_h(b)$ is the set of all positions reachable, according to the ruleset, when the game is started at position $b$ with $h$ as the current player.
Let
\begin{eqnarray}
\mbox{\rm Reachable}(b) := \mbox{\rm Reachable}_L(b)\cup \mbox{\rm Reachable}_L(b)
\end{eqnarray}

Note that game ${\cal R}[b_0]$ is only played on positions reachable from $b_0$; players need not make decision on any position not reachable from $b_0$.
Thus, while the sets ${\cal B}$ and ${\cal S}$ for abstract rulesets are generally infinite,
game instances are usually played on finite subset of ${\cal B}$, and hence only involves finite subset of moves from ${\cal S}$.

\begin{definition}[Self-Contained Domain of Positions]
We say $B\subseteq {\cal B}$ is {\em closed under move functions} $\{\rho_L,\rho_R\}$ if
$\rho_h(b,\sigma)\in B\cup \{NULL\}$, $\forall b\in B, \sigma\in {\cal S}, h\in \{L,R\}$.
\end{definition}

\begin{proposition}
In any ruleset ${\cal R} =({\cal B},{\cal S},\rho_L,\rho_R)$, for any $b\in {\cal B}$,
$\mbox{\rm Reachable}(b)$ is
 closed under
move functions $\{\rho_L,\rho_R\}$.
\end{proposition}

In our computational study of combinatorial games,
  we use the following formal description
 of game instances.
\begin{definition}[Description of Game Instances]
For a ruleset ${\cal R} =({\cal B},{\cal S},\rho_L,\rho_R)$ and a position $b_0\in {\cal B}$,
 the description of game instance ${\cal R}[b_0]$
 is given by five components
 denoted here by, $Z=(B,\Sigma,\rho_{L}, \rho_{R},b_0)$, where:

\begin{itemize}
 \item $b_0$ is the starting position of the game. 

\item
$B\ni b_0$
is a set,  closed under operations
$\{\rho_L,\rho_R\}$,
of possible positions containing all potential game positions {\em reachable} when applying the ruleset to play the game from the starting position $b_0$.\footnote{It is not uncommon that the computational problem of determining whether a given position can be reached from another one according to the ruleset is \cclass{NP}-hard. Thus, while ideally we want to define $B$ as all reachable positions from $b_0$, it is often set it to be the subdomain of positions of the subfamily of focus that this game instance belongs to, as the larger set may actually have more efficient description. }
    \item $\Sigma\subseteq {\cal S}$ is an alphabet of possible moves of ${\cal R}$ for players over positions in $B$.
\end{itemize}
\end{definition}
As will become clear later, we define
 game instances using subsets $B\subseteq {\cal B}$ and $\Sigma\subseteq {\cal S}$ in order to formalize the notion of games' {\em
 descriptive size} for our algorithmic and complexity-theoretical studies.
As mathematical objects, we can simply say $Z :=
{\cal R}[b_0]$.
In the mathematical formulation of game instances, we assume the following property of {\em complete information}.

\begin{property}[Complete-Information]
\label{prop:CompleteInformation}
Each description $Z=(B,\Sigma,\rho_{L}, \rho_{R},b_0)$ of  the game instance
${\cal R}[b_0]$ satisfies the property that $B$ is closed under move functions $\{\rho_L,\rho_R\}$, and
$\mbox{\rm Reachable}(b_0) \subseteq B.$
\end{property}

\subsection{Examples}
We highlight the notation with three example games.

\begin{itemize}
\item \ruleset{Nim} is a classic game that is the origin of Combinatorial Game Theory \cite{Bouton:1901,ONAG:2001,WinningWays:2001,Sprague:1936,Grundy:1939}.
In this mathematical game of numbers, there are a set of piles that each have a certain value.
Players alternate turns reducing the value of any pile to any smaller natural number.
A player loses when there are no piles with value greater than 0 remaining.
Suppose the game has $m$ piles of values $n_1,...,n_m$, respectively.
Then, by labeling the piles by integers $1$ to $m$, we can represent this game position by a vector as $[n_1,...,n_m]$.
So, as an abstract CGT ruleset,
the domain of game positions, ${\cal B}_{\ruleset{Nim}}$ can be expressed by the set of all non-negative vectors.
The set of moves ${\cal S}_{\ruleset{Nim}}$
can be expressed by the set of all vectors in which all but one entry has value zero and the non-zero entry has a negative integer value,
specifying the amount of intended reduction to the corresponding pile.
For $b\in {\cal B}_{\ruleset{Nim}}$ and $\sigma\in {\cal S}_{\ruleset{Nim}}$, the move function is essentially vector addition:
\begin{itemize}
\item If $b$ and $\sigma$ has different dimensions, then return NULL;
\item otherwise return $b+\sigma$ if it is non-negative, and NULL if $b+\sigma$ has a negative entry.
\end{itemize}

To illustrate our notation of concrete game instance, we will look at the simple position of when there are 2 piles of 2.
In this game, the reachable set has nine possible positions: 
$$[2,2],[2,1][1,2],[2,0][0,2],[1,1],[1,0],[0,1],[0,0].$$

It may seem like there should be less, since, for example, $[1, 2]$ is \textit{symmetric} to $[2, 1]$, and $[2, 0]$ is symmetric to
$[0, 2]$, and so on.
Symmetric in this context means that the game trees (see below) are exactly the same, other than using a different label in the $\rho$ function.
This is perfectly acceptable (and used by researchers in CGT) for analysis of the classical game.
As we shall see later, the quantum setting, however, is sensitive to this labeling of piles, so it is critical that we don't apply any "obvious" pruning.
We will show this when discussing quantum extension to \ruleset{Nim}.
\end{itemize}

\begin{itemize}
\item An  \ruleset{Geography} instance is typically defined by a directed graph $G = (V, E)$ and a starting vertex $v \in V$ (on which a toke is placed).
Players alternate moving the token to a vertex reachable through an out-degree edge.
Players are unable to enter any vertex that the token previously traversed.
So mathematically, in a \ruleset{Geography} position, vertices in the graph also have a Boolean state, indicating whether the vertex is available to the token.
When a player is unable to move, they lose.
{
}

\end{itemize}

\begin{itemize}
\item In \ruleset{Hex}, we have a source node $s$ and a sink node $t$.
Players alternate turns labeling vertices in the graph $L$ and $R$, respectively, by the two players.
Once a vertex is labeled, it can't be relabeled.
A player wins when there is a path from $s$ to $t$ using only vertices labeled with their respective label.
In \cref{fig:hex}, we have a sample board:

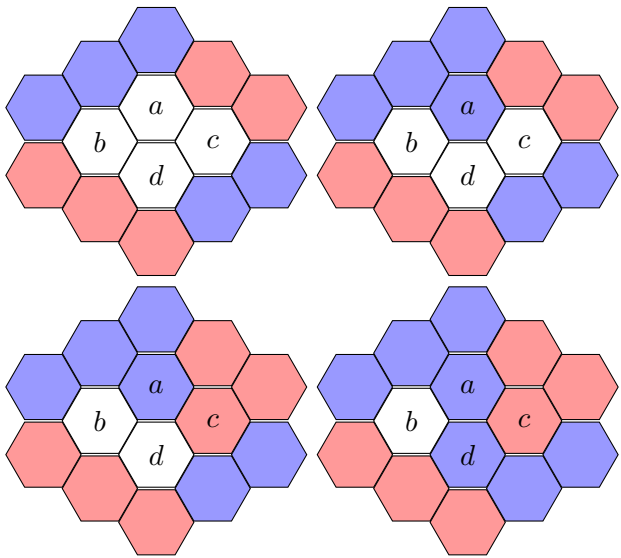
\begin{figure}[h]
  \begin{center}
    \begin{tikzpicture} [node distance = .9cm]
        \tikzstyle{hex} = [shape=regular polygon,regular polygon sides=6,minimum size=1cm, draw,inner sep=0,anchor=south]
        \tikzstyle{blue} = [fill=white!60!blue]
        \tikzstyle{red} = [fill=white!60!red]

        \node[hex, blue] (top) {};
        \node[hex] (a) [below of = top] {$a$};
        \node[hex] (d) [below of = a] {$d$};
        \node[hex, blue] (leftBlue) at (-.75, -.45) {};
        \node[hex] (b) [below of = leftBlue] {$b$};
        \node[hex, blue] (leftBlue2) at (-1.5, -.9) {};
        \node[hex, red] (leftRed1) [below of = leftBlue2] {};
        \node[hex, red] (leftRed2) [below of = b] {};
        \node[hex, red] (leftRed3) [below of = d] {};
        \node[hex, red] (rightRed1) at (.75, -.45) {};
        \node[hex] (c) [below of = rightRed1] {$c$};
        \node[hex, blue] (rightBlue2) [below of = c] {};
        \node[hex, red] (rightRed2) at (1.5, -.9) {};
        \node[hex, blue] (rightBlue2) [below of = rightRed2] {};

    \end{tikzpicture}
    \begin{tikzpicture} [node distance = .9cm]
        \tikzstyle{hex} = [shape=regular polygon,regular polygon sides=6,minimum size=1cm, draw,inner sep=0,anchor=south]
        \tikzstyle{blue} = [fill=white!60!blue]
        \tikzstyle{red} = [fill=white!60!red]

        \node[hex, blue] (top) {};
        \node[hex, blue] (a) [below of = top] {$a$};
        \node[hex] (d) [below of = a] {$d$};
        \node[hex, blue] (leftBlue) at (-.75, -.45) {};
        \node[hex] (b) [below of = leftBlue] {$b$};
        \node[hex, blue] (leftBlue2) at (-1.5, -.9) {};
        \node[hex, red] (leftRed1) [below of = leftBlue2] {};
        \node[hex, red] (leftRed2) [below of = b] {};
        \node[hex, red] (leftRed3) [below of = d] {};
        \node[hex, red] (rightRed1) at (.75, -.45) {};
        \node[hex] (c) [below of = rightRed1] {$c$};
        \node[hex, blue] (rightBlue2) [below of = c] {};
        \node[hex, red] (rightRed2) at (1.5, -.9) {};
        \node[hex, blue] (rightBlue2) [below of = rightRed2] {};

    \end{tikzpicture}

    \vspace{.1cm}

    \begin{tikzpicture} [node distance = .9cm]
        \tikzstyle{hex} = [shape=regular polygon,regular polygon sides=6,minimum size=1cm, draw,inner sep=0,anchor=south]
        \tikzstyle{blue} = [fill=white!60!blue]
        \tikzstyle{red} = [fill=white!60!red]

        \node[hex, blue] (top) {};
        \node[hex, blue] (a) [below of = top] {$a$};
        \node[hex] (d) [below of = a] {$d$};
        \node[hex, blue] (leftBlue) at (-.75, -.45) {};
        \node[hex] (b) [below of = leftBlue] {$b$};
        \node[hex, blue] (leftBlue2) at (-1.5, -.9) {};
        \node[hex, red] (leftRed1) [below of = leftBlue2] {};
        \node[hex, red] (leftRed2) [below of = b] {};
        \node[hex, red] (leftRed3) [below of = d] {};
        \node[hex, red] (rightRed1) at (.75, -.45) {};
        \node[hex, red] (c) [below of = rightRed1] {$c$};
        \node[hex, blue] (rightBlue2) [below of = c] {};
        \node[hex, red] (rightRed2) at (1.5, -.9) {};
        \node[hex, blue] (rightBlue2) [below of = rightRed2] {};

    \end{tikzpicture}
    \begin{tikzpicture} [node distance = .9cm]
        \tikzstyle{hex} = [shape=regular polygon,regular polygon sides=6,minimum size=1cm, draw,inner sep=0,anchor=south]
        \tikzstyle{blue} = [fill=white!60!blue]
        \tikzstyle{red} = [fill=white!60!red]

        \node[hex, blue] (top) {};
        \node[hex, blue] (a) [below of = top] {$a$};
        \node[hex, blue] (d) [below of = a] {$d$};
        \node[hex, blue] (leftBlue) at (-.75, -.45) {};
        \node[hex] (b) [below of = leftBlue] {$b$};
        \node[hex, blue] (leftBlue2) at (-1.5, -.9) {};
        \node[hex, red] (leftRed1) [below of = leftBlue2] {};
        \node[hex, red] (leftRed2) [below of = b] {};
        \node[hex, red] (leftRed3) [below of = d] {};
        \node[hex, red] (rightRed1) at (.75, -.45) {};
        \node[hex, red] (c) [below of = rightRed1] {$c$};
        \node[hex, blue] (rightBlue2) [below of = c] {};
        \node[hex, red] (rightRed2) at (1.5, -.9) {};
        \node[hex, blue] (rightBlue2) [below of = rightRed2] {};

    \end{tikzpicture}
  \end{center}
\caption{An example small \ruleset{Hex} game, played from the 2 $\times$ 2 starting position to a win by Blue.}
\label{fig:hex}
\end{figure}

In this game, $B$ is the set of possible positions the game can end up in. $\lvert B \rvert = 3^5$, as each of the 5 active game nodes can be either unlabeled or be labeled $L$ or $R$. Every one of these positions is reachable. 
If the current board is $b_0$, than there are 5 values of $\alpha$ that result in a non-null value for $\rho_L(\alpha, b_0)$ and 5 values of $\beta$ which result in a non-null value for $\rho_R(\beta, b_0)$. These moves will in total map to 10 distinct boards. Any board where $s$ and $t$ are linked through a path of vertices with the same labeling, it will be a terminal board, and thus the opponent will be unable to play.
\end{itemize}

\subsection{Outcome Classes of Positions}
In CGT, game positions fall into one of four \emph{outcome classes}, based on which player has a winning strategy given who will go first:

\begin{itemize}
    \item \outcomeClass{N} ("Fuzzy"): In this class, the \emph{next}  (or first) player always has a winning strategy.  (It doesn't matter whether that player is Left or Right.)

    \item \outcomeClass{P} ("Zero"): For positions in this class, the \emph{previous} (second) player always has a winning strategy.

    \item \outcomeClass{L} ("Positive"): In this class, the \emph{Left} player has a winning strategy, independent of who goes first.

    \item \outcomeClass{R} ("Negative"): In this class, the \emph{Right} player has a winning strategy.
\end{itemize}

These outcome classes form a partition on all game positions.

For any game position, $b$,  $o(b)$ is the the outcome decision function  that returns the outcome class of $b$.  In other words, the codomain of $o$ is $\{\outcomeClass{N}, \outcomeClass{P}, \outcomeClass{L}, \outcomeClass{R}\}$ and $b \in o(b)$.

In an impartial ruleset,  $b$ must be in $\outcomeClass{P} \cup \outcomeClass{N}$,
for any position, $b$.
For impartial rulesets, this means that $b$ is in $\outcomeClass{P}$ if \emph{all} options of $b$ are in $\outcomeClass{N}$.
Alternatively, $b$ is in $\outcomeClass{N}$ if \emph{at least one} option is in $\outcomeClass{P}$.  Thus, we can prove that an (impartial) position is in $\outcomeClass{N}$ by showing one option in $\outcomeClass{P}$ instead of finding the outcome class of all of them.

Thus, the outcome class, in the case of impartial games, of a game position can be logically viewed as the outcome in an evaluation process of an AND-OR tree when we equal $\outcomeClass{N}$ as  ``\texttt{true} - the current player for having a winning move''
and $\outcomeClass{P}$ as  ``\texttt{false} - the current player has no winning move.''

\subsection{Game Trees}

A game instance $Z=(B,\Sigma,\rho_{L}, \rho_{R},b_0)$ defines a natural game tree, in
  which each node is defined by a sequence of ``feasible moves'' and has an associated game position from $B$.
For illustration, we will index these tree nodes by the sequence of moves leading up to them.
Thus, we denote the root of the game tree by $v_{\emptyset}$, and its
  associated position by $p_{\emptyset}$.
Clearly, $p_{\emptyset} = b_0$, the starting position of the game instance.
We can grow the game tree by a natural iterative procedure starting from $(v_{\emptyset},p_{\emptyset})$:
In the general iterative steps, suppose $v_{s}$ is a node in the game tree defined by a sequence $s$ of moves, with associated position $p_{s}$.
Suppose there are $T(s)$ feasible moves $\sigma_1,...,\sigma_{T(s)}$ for position
$p_{s}$.
Then, in this game tree, $v_s$ has $T(s)$
  {\em children},
respectively, $v_{\sigma_1\circ s},...,v_{\sigma_{T(s)}\circ s}$,
one for each feasible moves.
For $t\in [T(s)]$, the position $p_{\sigma_t\circ s}$ associated with node $v_{\sigma_t\circ s}$ is then $p_{\sigma_t\circ s} = \rho(p_{s},\sigma_t)$.
If $v_s$ is terminal, then it is a leaf of the game tree.

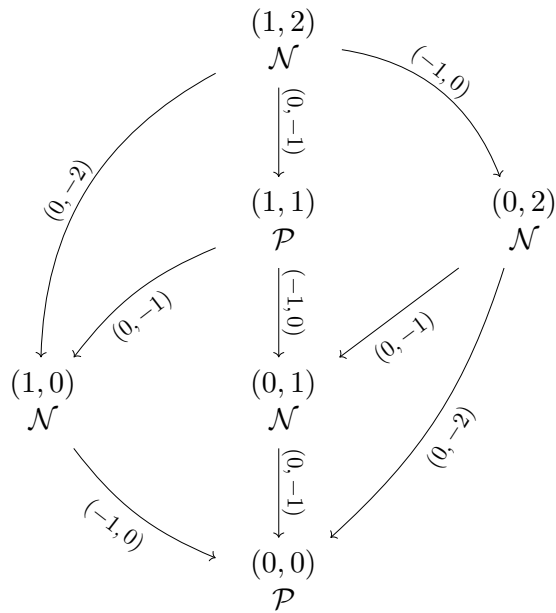
\begin{figure}[h]
  \begin{center}
  \scalebox{.8}{
  \begin{tikzpicture} [node distance = 1cm]
    \node (start) at (0, 6)  {\scalebox{\scaleAmt}{
        \begin{tabular}{c}
            $(1, 2)$ \\
            $\outcomeClass{N}$
        \end{tabular}}};

    \node (firstB) at (0, 3)  {\scalebox{\scaleAmt}{
        \begin{tabular}{c}
            $(1, 1)$ \\
            $\outcomeClass{P}$
        \end{tabular}}};

    \node (firstC) at (4, 3)  {\scalebox{\scaleAmt}{
        \begin{tabular}{c}
            $(0, 2)$ \\
            $\outcomeClass{N}$
        \end{tabular}}};

    \node (secondA) at (-4, 0) {\scalebox{\scaleAmt}{
        \begin{tabular}{c}
            $(1, 0)$ \\
            $\outcomeClass{N}$
        \end{tabular}}};

    \node (secondB) at (0, 0) {\scalebox{\scaleAmt}{
        \begin{tabular}{c}
            $(0, 1)$ \\
            $\outcomeClass{N}$
        \end{tabular}}};

    \node (end) at (0, -3) {\scalebox{\scaleAmt}{
        \begin{tabular}{c}
            $(0, 0)$ \\
            $\outcomeClass{P}$
        \end{tabular}}};

    \path[->]
        (start) edge [bend right] node [sloped, text width = 1.5cm, yshift = 7] {$(0, -2)$} (secondA)

        (start) edge [] node [sloped, text width = 1.5cm, yshift = 7] {$(0, -1)$} (firstB)

        (start) edge [bend left] node [sloped, text width = 1.5cm, yshift = 7] {$(-1, 0)$} (firstC)

        (firstB) edge [bend right=15] node [sloped, text width = 1.5cm, yshift = -11] {$(0, -1)$} (secondA)

        (firstB) edge [] node [sloped, text width = 1.5cm, yshift = 7] {$(-1, 0)$} (secondB)

        (firstC) edge [] node [sloped, text width = 1.5cm, yshift = -11] {$(0, -1)$} (secondB)

        (firstC) edge [bend left = 15] node [sloped, text width = 1.5cm, yshift = -11] {$(0, -2)$} (end)

        (secondA) edge [bend right = 15] node [sloped, text width = 1.5cm, yshift = -11] {$(-1, 0)$} (end)

        (secondB) edge [] node [sloped, text width = 1.5cm, yshift = 7] {$(0, -1)$} (end)
        ;

  \end{tikzpicture}}
  \end{center}
  \caption{Game tree for \ruleset{Nim} $(1,2)$, drawn as a DAG (Directed Acyclic Graph) to save space.  The outcome classes are drawn underneath each position.}
  \label{fig:nim12}
\end{figure}

See \cref{fig:nim12} for an example of a game tree showing that \ruleset{Nim} position $(1,2)$ is in $\outcomeClass{N}$.
Note that as it is defined above, a game tree  could have two different nodes with the same associated position, i.e., two difference sequences of moves lead to the same position. When illustrating game trees---as in Figure \ref{fig:nim12}---we often use a more compressed DAG representation for the game trees.

The following proposition --- which can be established by induction --- summarizes the above indexing convention for game trees.

\begin{proposition}
With the indexing convention above, if $v_s$ denotes a node in the game tree with associated game position $p_s\in B$, then,
$$p_s = \rho_\cp(b_0,s)$$
\end{proposition}

Notice that the {\em height} of a game tree represents the longest possible sequence of plays with the game.
The total number of nodes in the game tree is usually exponential in the height of the game tree.

A game instance $Z=(B,\Sigma,\rho_{L}, \rho_{R},b_0)$ in fact defines a family of game instances: For any $b\in B$, let
$Z[b] :=(B,\Sigma,\rho_{L}, \rho_{R},b)$
denote the game instance starting at position $b$ with the same ruleset.
The game tree for $Z$ can be viewed as a structure connecting all game trees for game instances in $\{Z[b]\ | \ b\in B \}$:
If $b$ is reachable from $b_0$ according to the ruleset, then the game tree for $Z[b]$ is a subtree in the game tree for $Z$.
Thus, in principle, to made intelligent decisions on $Z$, the players need to study not
just the initial game position $b_0$, but also at least some game positions reachable from $b_0$.
For this reason, we refer to $Z[b]$ (for $b$ reachable from $b_0$ according to the ruleset) a {\em partially-played} game instance of $Z$.


\subsection{Decision and Search Problems for Playing Games}

Making strategically optimal plays in combinatorial games requires one to solve some basic decision and search problems.
They include:

\begin{itemize}
\item {\sc Decision of Outcome}: Given a game position $b\in B$, return the outcome class of the game position. Particularly,
{\sc Decision of Outcome} addresses whether the game position is winnable by the current player: \texttt{true} if the current player has a winning move; otherwise  \texttt{false}.

\item {\sc Search of Winning moves}: Given a game board $b\in B$, return a feasible move $\sigma$ such that position $\cp$ $\rho_\cp(b,\sigma)$ has no winning move for its current player; if such a move does not exist, then return "not winnable."
\end{itemize}

Clearly, solving the second problem also solves the first, since the function should return false if and only if there doesn't exist a winning move. The following useful proposition is straightforward.

\begin{proposition}
If at a game position $b$ there are $n_{b}$ possible moves, $\sigma_1,...,\sigma_{n_b}$ then solving the {\sc Decision of Outcome} on
$n_{b}$ positions, $\rho_\cp(b,\sigma_1)$,..., $\rho_\cp(b,\sigma_{n_{b}})$ (some of which could be NULL), is sufficient to solve {\sc Search for Winning Moves} at position $b$.
\end{proposition}

Other interesting computational problems regarding combinatorial games have been considered.
Here, we just give one simple example:

\begin{itemize}
\item {\sc Decision of Realizability}:
Given a game board $b\in B$ in a game $Z=(B,\Sigma,\rho_L, \rho_R, b_0)$, and $\cp \in \{L, R\}$
deciding whether or not there exists a sequence of moves
$s \in \Sigma^*$ such that
$$ b= \rho_\cp(b_0,s)$$
\end{itemize}

{\sc Decision of Outcome} and {\sc Search for Winning Moves} are fundamental computational problems associated with combinatorial games.
Understanding their complexity in the quantum setting will be the main focus of our paper.


\subsection{Input Size -- Games' Descriptive Complexity }

To characterize the computational complexity of combinatorial games---in both classical and quantum settings---we need to first characterize their input size.
Although some games are {\em explicitly} represented\footnote{that is, all game positions and rules of interactions between moves and  positions are fully enumerated.} (e.g., \ruleset{Tic-Tac-Toe}),
many common combinatorial games (e.g., \ruleset{Nim}, \ruleset{Hex}, \ruleset{Go}) are {\em succinctly represented}.
A major reason for succinct representation is that the number of possible game positions are exponential in what
is usually considered as the natural size  (e.g., the succinct representation of a normal size \ruleset{Go} board is size $19\times 19$, while the entire game tree is doubly exponential to that) of the games.
In such cases, all three components, particularly, game positions and rules, can be "functionally"
represented in more succinct manners.
Another valid reason for succinct representation of combinatorial games comes from the fact that human are drawn to "complex games with elegant game boards and simple rules."

In complexity-theoretical studies of combinatorial games, the size of a game representation
---rather than explicit enumerations of game positions, moves, and transitional rules---is usually considered as the {\em input size} of a given game.
For example, while in many classical games, such as \ruleset{Node-Kalyes} and \ruleset{Geography},
the size of alphabet of possible moves is polynomial in the game's descriptive complexity, in others, such as \ruleset{Nim}, the number of moves could be exponential in the game's descriptive complexity.
In the latter cases, we say the alphabet of moves is {\em succinctly defined}.
For many succinctly-defined games, their {\em input size} is measured essentially by their {\em descriptive complexity}---usually in bits---of the start position.
For example, the input size of a \ruleset{Nim} position is the total number bits representing the starting piles.
Note that in \ruleset{Nim}, both the number of possible positions and the alphabet of possible moves for the position could be exponential in this descriptive complexity. Yet, remarkably, {\sc Decision for Outcome}
and {\sc Search for Winning Moves} can be solved in time polynomial in this succinctly-descriptive complexity.

Mathematically, we will, in this paper, use the following basic property of the {\em descriptive complexity} for succinctly-formulated combinatorial games:

\begin{property}[Descriptive Complexity]
\label{prop:DescriptiveComplexity}
Suppose $n_Z$ measures the {\em descriptive complexity of} a game instance $Z= (B,\Sigma,\rho_{L},\rho_{R},b_0)$.
Then:
\begin{enumerate}
\item Each position $b\in B$ and each move $\sigma \in \Sigma$ can be defined by a binary strings with
length polynomial in $n_Z$.
\item The move functions $\rho_L$ and $\rho_R$ are polynomial-time computable (in terms of $n_Z$).
\item Games can be further classified based on whether $|Z|$ is polynomial or super-polynomial (e.g. exponential) in $n_Z$:
\begin{itemize}
\item [(3.P)] In the case when $|Z|$ is polynomial in $n_Z$,
we say that players {\em explicitly} know their moves, because $Z$ can be explicitly enumerated (in polynomial time).
We will refer to this type of games, for example, \ruleset{Geography} and \ruleset{Node-Kayles}, as {\em polynomially-dense} games.
\item [(3.SP)] In the case when $|Z|$ is super-polynomial (e.g. exponential) in $n_Z$,  we say that players succinctly know their moves (think of \ruleset{Nim}), because $Z$ can not be explicitly enumerated with polynomial complexity.
We will refer to this type of games as {\em super-polynomailly-dense} games.

In this case, we take the following {\em generative framework} to express this succinct knowledge with two additional polynomial-time computable generative functions: $\mu_L: B\times |\Sigma|\rightarrow \Sigma$ and $\mu_R: B\times |\Sigma| \rightarrow \Sigma$, one for each players.
In essence, for any position $b\in B$ and a positive integer $t\in |\Sigma|$,
$\mu_L(b,t)$ ($\mu_R(b,t)$) denotes the $t^{th}$ potential move in the mind of the player Left (Right) at position  $b$.
\end{itemize}
\end{enumerate}
\end{property}

Both in theory and in practice,
  some natural parameters\footnote{such as, the side length of the \ruleset{Hex} board, the side length of the Sperner's triangle in \ruleset{Atropos}, $19 \times 19$ in classical \ruleset{Go}, the number of nodes in the graphs for \ruleset{Geography} or \ruleset{Node-Kayles}.}---usually polynomially related with the descriptive complexity of the game representations---are used to denote the games' input sizes.
Notice that in the traditional polynomial-time characterization of efficient computing, any polynomially-related input-size measure leads to equivalent characterization.

\subsection{Quantum Lift of Combinatorial Games}
\label{Sec:QuantumLift}
We now return the main subject of this paper, namely, quantum extensions of combinatorial games.
Throughout the paper, we will use the following notations.
In this definition, we fix a combinatorial game
$Z=(B,\Sigma,\rho_L, \rho_R,b_0)$.
In the discussions below, we also fix $h\in \{L,R\}$.

\begin{definition}[Quantum Moves and Superposition Width]
Let $w$ be a positive integer $w > 1$,
and $\sigma_1,...,\sigma_w\in \Sigma$ be $w$
distinct moves for game $Z$.
We will use $\langle \sigma_1\ |\ \sigma_2\ |\ ... \ |\ \sigma_w \rangle$
to denote the quantum move defined by
  the superposition of these classical moves.
We will refer $w$ as the {\em superposition width} of this quantum move.
We will use $\Sigma^{\QuantumLift(w)}$ to denote the set of all quantum moves with superposition width
at least $2$ and {\em at most} $w$.
\end{definition}

\begin{definition}[Quantum Positions (Quantum Game Boards)]
Let $s$ be a positive integer $s>1$,
and $b_1,...,b_s\in B$ be $s$
{\em distinct} game positions  in $Z$.
We will use $\langle b_1\ |\ b_2\ |\ ...\  | \ b_s\rangle$
to denote the quantum position (quantum game board) defined by the superposition of these classical game positions (boards).
In this case, we say $\langle b_1\ |\ b_2\ |\ ...\  | \ b_s\rangle$ is a $s$-wide quantum position, or $s$-wide superposition.
We will use $B^{\QuantumLift}$ to denote the set of all possible quantum positions.
\end{definition}

Below we will use the following operator to remove NULL and duplicates from a superposition.
\begin{definition}[Purifying]
Consider $b_1,...,b_s\in B \cup \{NULL\}$.
If all $b_i = NULL, \forall i\in [s]$,
 then $\filter(\langle b_1\ |\ ...\ | \ b_s\rangle ):= NULL$, otherwise,  $\filter(\langle b_1\ |\ ...\ | \ b_s\rangle)$ denotes the quantum position in $B^{\QuantumLift}$ corresponding to the superposition of the distinct non-NULL
game positions in $\{b_1,...,b_s\}$.
\end{definition}

We now introduce the notations for
  formalizing the transitional rules in quantum games,
  specifying the interactions between classical/quantum moves with
  classical/quantum positions.
We first introduce two operators:
\begin{definition}[Elementry Operators for Quantum Positions and Moves]
The first operator concerns the flavor-independent interaction between classical moves and quantum game boards.
\begin{itemize}
\item [$\otimes$]: For $\sigma \in \Sigma$ and
$\mathbb{B} = \left\langle b_1\ |\  b_2\  |\  ... \ |\ b_s\right\rangle \in B^{\QuantumLift}$,
then
\begin{eqnarray}
\sigma \otimes \mathbb{B} = \filter(\left\langle \rho_h(b_1,\sigma)\ |\ ...\ |\ \rho_h(b_s,\sigma)\right\rangle)
\end{eqnarray}
For $\mathbb{M}=\left\langle \sigma_1\ |\ ...\ |\  \sigma_{w'}\right\rangle \in \Sigma^{\QuantumLift(w)}$ ($2\leq w'\leq w$)
  and  $b \in B$, then
\begin{eqnarray}
\mathbb{M}\otimes b = \filter(\left\langle\rho_h(b,\sigma_1)\ |\ ...\ |\ \rho_h(b,\sigma_{w'})\right\rangle)
\end{eqnarray}
\end{itemize}

The second operator takes the basic superposition of two quantum boards.
\begin{itemize}
\item [$\oplus$]: Consider two quantum game boards
$\mathbb{B} = \left\langle b_1\ |\ b_2\ |\ ... \ |\ b_s\right\rangle \in B^{\QuantumLift}$ and
$\mathbb{B'} = \left\langle b'_1 \ |\  b'_2\  |\  ... \ |\ b'_{s'}\right\rangle \in B^{\QuantumLift}$,
\begin{eqnarray}
\mathbb{B}\oplus \mathbb{B'} = \filter(\left\langle b_1\ |\ b_2 \ |\  ... \ |\ b_s\ | \ b'_1\ |\ b'_2\ | \ ...\  |\ b'_{s'}\right\rangle).
\end{eqnarray}
\end{itemize}
\end{definition}

The following classifications will further help us to explain the nuance among differ quantum flavors.

\begin{definition}[Eligible Superpositions for Quantum Moves]
We say $\mathbb{M} = \left\langle \sigma_1\ |\  ...\ |\ \sigma_{w'} \right\rangle \in
\Sigma^{\QuantumLift(D,w)}$ ($2\leq w'\leq w$) is an
{\em eligible superposition} as a quantum move for a classical position $b\in B$ if for all $i\in [w']$,
$\rho_h(\sigma_i,b)  \neq NULL$,
i.e., meaning the superposition does not use infeasible moves, and hence the number of
realizations in $\mathbb{M}\otimes b$ is the same the width of $\mathbb{M}$.

We say
 $\mathbb{M} =
 \left\langle\sigma_1\ |\ ...\ |\ \sigma_{w'} \right\rangle \in
\Sigma^{\QuantumLift(D,w)}$ ($2\leq w'\leq w$) is an
{\em eligible superposition} as a quantum move for a quantum position
$\mathbb{B} = \left\langle b_1\ |\ ...\ |\ b_s\right\rangle \in B^{\QuantumLift(D,w)}$ ($s>1$) if for all $i\in [w']$,
$\sigma_i \otimes \mathbb{B} \neq NULL$.
In other words, each move $\sigma_i$ is feasible for some realization in $\mathbb{B}$.
\end{definition}

\begin{definition}[Safe and Respectful Classic Moves]
\label{Defi:SafeRespectful}
Fix a quantum position $\mathbb{B} = \left\langle b_1\ |\ ...\ |\ b_s\right\rangle \in B^{\QuantumLift(D,w)}$ ($s>1$).
We say a classical move $\sigma\in \Sigma$ is {\em safe}
for $\mathbb{B}$ if  $\forall j\in [s]$, $\rho_h(b_j,\sigma)\neq NULL$.
We say  $\sigma\in \Sigma$ is {\em respectful} to $\mathbb{B}$
if  $\forall j\in [s]$, $\rho_h(b_j,\sigma)\neq NULL$
whenever $b_j$ is not a terminal position in $Z$.
\end{definition}

\begin{proposition}
With respect to a quantum position $\mathbb{B} = \left\langle b_1\ |\ ...\ | \ b_s\right\rangle \in B^{\QuantumLift(D,w)}$ ($s>1$),
every safe classical move is also respectful.
\end{proposition}

We now introduce notations for formalizing
  extending combinatorial games with quantum moves.

\begin{definition}[Quantum Lift of Combinatorial Games]
For a quantum flavor $\phi\in \{A,B,C,C',D\}$, and
positive integer $w > 1$, we will use $$Z^{\QuantumLift(D,w)} = \{ B^{\QuantumLift(\phi,w)},\Sigma^{\QuantumLift(\phi,w)},\rho_L^{\QuantumLift(\phi,w)},
\rho_R^{\QuantumLift(\phi,w)}
b_0\}$$
to denote the game in quantum flavor $\phi$ with quantum moves of superposition width up to $w$.
In the quantum lift:
\begin{eqnarray}
B^{\QuantumLift(A,w)}  =
B^{\QuantumLift(B,w)} = B^{\QuantumLift(C,w)}=B^{\QuantumLift(C',w)}=B^{\QuantumLift(D,w)} = B \cup B^{\QuantumLift}\\
\Sigma^{\QuantumLift(A,w)} =  \Sigma^{\QuantumLift(w)}; \quad
\Sigma^{\QuantumLift(B,w)}  =  \Sigma^{\QuantumLift(C,w)}= \Sigma^{\QuantumLift(C',w)}= \Sigma^{\QuantumLift(D,w)}= \Sigma \cup \Sigma^{\QuantumLift(w)}
\end{eqnarray}
For the rule of the lifted quantum game, fix any
classical position $b\in B$, and move
$\sigma \in \Sigma$, as well as
any quantum position
$\mathbb{B} = \left\langle b_1\ |\ ...\ |\ b_s\right\rangle \in B^{\QuantumLift(D,w)}$ ($s>1$) and quantum moves $\mathbb{M} = \left\langle \sigma_1\ |\ ...\ |\ \Sigma_{w'} \right\rangle  \in
\Sigma^{\QuantumLift(D,w)}$ ($2\leq w'\leq w$).
All quantum flavors share the same interaction between quantum moves and classical/quantum positions.
For any feasible pair of position and move:
\begin{itemize}
\item  $\rho_h^{\QuantumLift(\phi,w)}(b,\mathbb{M}) = \filter( \mathbb{M}\otimes b)$.
\item $\rho_h^{\QuantumLift(\phi,w)}(\mathbb{B},\mathbb{M})
= \filter((\mathbb{B}\otimes\sigma_1)\oplus ... \oplus(\mathbb{B}\otimes\sigma_{w'})).$
\end{itemize}
Quantum flavors differ in their handling of classical moves.
The least restrictive one is quantum flavor $D$, whose game rule $\rho_h^{\QuantumLift(D,w)}$ is what most would expect:
\begin{itemize}
\item $\rho_h^{\QuantumLift(D,w)}(b,\sigma) = \rho_h(b,\sigma).$
\item $\rho_h^{\QuantumLift(D,w)}(\mathbb{B},\sigma)
= \filter(\sigma\otimes \mathbb{B})$,
\end{itemize}
The most restrictive one is flavor $A$,
   which disallows any classical move.

Other  quantum flavors, namely $B,C,C'$, capture various nuances.
With flavor $B$, classical moves can only be made if there is no eligible superposition for quantum move for the current game position.
With flavor $C$, only safe classical moves are allowed. With flavor $C'$, only respectful classical moves are allowed.

For $\phi\in \{B,C,C'\}$, whenever a classical move $\sigma\in \Sigma$ is allowed, the following game rule
applies:

\begin{itemize}
\item $\rho_h^{\QuantumLift(\phi,w)}(b,\sigma) = \rho_h(b,\sigma).$
\item $\rho_h^{\QuantumLift(\phi,w)}(\mathbb{B},\sigma)
= \filter(\sigma\otimes \mathbb{B})$,
\end{itemize}
\end{definition}

For impartial games, the following basic property of
quantum lift are useful for its outcome analysis.

\begin{proposition}[Quantum Lift of Impartial Games]\label{prop:Impartial}
For any impartial game $Z$,  quantum flavor
$\phi\in \{A,B,C,C',D\}$, and positive integer $w>1$,
the lifted quantum game $Z^{\QuantumLift(\phi,w)}$ remains a impartial game.
\end{proposition}

For partizan games, there are some subtleties in formalizing the outcome of the quantum extension.
We will provide more detailed discussion in Section \ref{sec:Reductions}.

We note that the concept of \textit{symmetry} can also apply to the quantum versions of combinatorial games, but we must be careful to alter the definition. Classically, two positions are symmetric if the game trees are symmetric.
However, this is insufficient for quantum games, since for quantum games the labeling of the moves can effect what games collapse, and thus who wins.
As such, we define \textit{quantum symmetry} when the game tree formed by two superpositions of quantum boards are symmetric. 

\section{Exploring \mbox{\ruleset{Quantum Undirected Geography}}}
\label{Sec:UndirectedGeography}

In this section, we will
   focus on \ruleset{Quantum Undirected Geography}.
Our two contrasting and complementing     algorithmic/complexity characterizations
  reveal the changing and majestic computational landscape of \ruleset{Quantum Undirected Geography}, with \cclass{PSPACE}-hard peaks over quantum positions and
  \cclass{Polynomial}-\cclass{Time}
  terrains across classical positions.
Furthermore, some \cclass{PSPACE}-hard positions are reachable from classical positions via polynomial-length quantum paths, making them realizable in \ruleset{Quantum Undirected Geography} with classical starts.
Quantum combinatorial games are wonderful succinctly-represented mathematical structures.
With this concrete quantum game as our example,
  we also illustrate a family of refined
  complexity measures for quantum combinatorial games.

\subsection{Quantumized Complexity of Combinatorial Games}

\label{Sec:Taxonomy}

Like in the traditional analysis of classical games, the most fundamental decision problem for
the quantum combinatorial games is also the decision of the outcome classes at a given (quantum) game position:

\begin{definition}[\sc{Decision of Outcome} of Quantum Games]
Consider a combinatorial game ruleset
${\cal R} =({\cal B},{\cal S},\rho_L,\rho_R)$.
Given a game position $\mathbb{B}\in {\cal B}^{\cal Q}\cup {\cal B}$ in its quantum extension, flavor $\phi\in \{A,B,C,C',D\}$, and superposition width $w$, determine the outcome class  of position $\mathbb{B}$ in the quantum extension of ${\cal R}$ with flavor $\phi$ and superposition with $w$.
\end{definition}

However, as potentially explosive extension of classical games,
quantum games have their own characters, introducing natural complexity measures of their solvability.
In particular, the ``degree of quantumness'' in $\mathbb{B}$ offers a natural parameter for classification of quantumized complexity.
We consider the following quantumness parameterization.
In the definition below, we say $\mathbb{B}$ is
{\em $s$-wide} if $\mathbb{B}$ contains $s$ realizations.
As such, when $s = 1$, $\mathbb{B}$ is a classical position and when $s>1$, $\mathbb{B}$ is a quantum position.

\begin{definition}[Bounded Forward Transform]
For a positive integer $T$ and a (starting) position $\mathbb{B}_0\in {\cal B}^{\cal Q}\cup {\cal B}$, we say $\mathbb{B}\in {\cal B}^{\cal Q}\cup {\cal B}$ is a $T$-{\em bounded forward transform} (or $T$-BFT for short) of $\mathbb{B}_0$ if
$\mathbb{B}$ can be obtained
by applying a sequence
$(m_1\circ m_2\circ\cdots\circ m_T)$ of quantum moves to position $\mathbb{B}_0$.
We denote it by:
$$\mathbb{B} := BFT(\mathbb{B}_0,m_1\circ m_2\circ\cdots\circ m_T).$$
\end{definition}

The next definition classifies the quantumized complexity of  combinatorial game rulesets according to $s$ and $T$.
\begin{definition}[Taxonomy of Quantumized Complexity]
The first group below is parameterized based on how wide  the querying position $\mathbb{B}$ is:
For a positive integer $s$, the {\em quantumized complexity of ${\cal R}$ with $s$-wide quantum start}
measures the worst-case complexity of algorithms for {\sc Decision of Outcome} when $\mathbb{B}$ contains $s$ realizations. Possibly interesting measures in this group include {\em Quantumized Complexity} with:

\begin{itemize}
\item {\bf Classical Start}: $\mathbb{B}\in {\cal B}$.
\item {\bf Constant-Wide Quantum Start}: $s$ is bounded by a constant.
\item {\bf Poly-Wide Quantum Start}: $s$ is bounded by a polynomial in the descriptive complexity of the underlying classical game instance.
\end{itemize}
The next group extends the group above by considering the case when
$\mathbb{B}$ can be specified as:
$\mathbb{B} = BFT(\mathbb{B}_0,m_1\circ m_2\circ\cdots\circ m_T)$:
For integers $s>0$ and $T\geq 0$, the {\em quantumized complexity of ${\cal R}$ with
$T$ moves after $s$-wide start}
measures the worst-case complexity of algorithms for {\sc Decision of Outcome} for $\mathbb{B}$
when given an $s$-wide $\mathbb{B}_0$ and
a sequence of $T$ quantum moves, $m_1\circ m_2\circ\cdots\circ m_T$ that defines $\mathbb{B}$.
Possibly interesting measures in this group include {\em Quantumized Complexity} with:
\begin{itemize}
\item {\bf Constant-Moves after Classical Start}: $\mathbb{B}_0\in {\cal B}$ and $T$ is a constant.
\item {\bf Poly-Moves After Classical Start}:
$\mathbb{B}_0\in {\cal B}$ and $T$ is bounded by a polynomial in the descriptive complexity of the underlying classical game instance.
\end{itemize}
We refer to such a (quantum) position $\mathbb{B}$ as a {\em reachable} quantum position from a classical start.
\begin{itemize}
\item {\bf Poly-Moves After Poly-Wide Quantum Start}: both $s$ and $T$ is bounded by a polynomial in the descriptive complexity of the underlying classical game instance.
\end{itemize}
\end{definition}

By the {\em quantumized complexity} of ruleset ${\cal R}$, we mean its quantumized complexity with poly-moves after poly-wide quantum start.
Note that any poly-moves after poly-wide quantum start position, $\mathbb{B} = BFT(\mathbb{B}_0,m_1\circ m_2\circ\cdots\circ m_T)$,  has a poly-sized {\em succinct representation},
although $\mathbb{B}$ may have exponentially number of realizations.

\begin{proposition}[Quantum Positions with Polynomial Description Size]
Suppose $\mathbb{B} = BFT(\mathbb{B}_0,m_1\circ m_2\circ\cdots\circ m_T)$, and $\mathbb{B}_0$ is $s$-wide.
Then, $\mathbb{B}$ has polynomial description size if both $s$ and $T$ are bounded by a polynomial in the descriptive complexity of the underlying classical game instance.
\end{proposition}

\subsection{Intractability of Quantum Geography Positions}

\label{Sec:PSPACECompleteQUG}

\begin{theorem}[Intractability of Poly-Wide Quantum Position]
\ruleset{Quantum Undirected Geography} is \cclass{PSPACE}-complete measured by quantumized complexity with poly-wide quantum start.
\end{theorem}

\begin{figure}[h]
  \begin{center}
  \begin{tikzpicture}[node distance = 2.5cm]
      \tikzstyle{vertex}=[circle,thick,draw = black, fill = white, minimum size=12mm]
      \tikzstyle{gray-vertex}=[circle,thick,draw = black, fill = lightgray, minimum size=6mm]

      \node[vertex] (a) at (0,0) {$a$};
      \node[vertex] (ab1) [right of = a] {$ab_1$};
      \node[vertex] (ab2) [right of = ab1] {$ab_2$};
      \node[vertex] (b) [right of = ab2] {$b$};
      \path[-]
         (a) edge [] node [] {} (ab1)
         (ab1) edge [] node [] {} (ab2)
         (ab2) edge [] node [] {} (b)
         ;
  \end{tikzpicture}

  \vspace{1cm}

  \begin{tikzpicture}[node distance = 2.5cm]
      \tikzstyle{vertex}=[circle,thick,draw = black, fill = white, minimum size=12mm]
      \tikzstyle{gray-vertex}=[circle,thick,draw = black, fill = lightgray, minimum size=6mm]

      \node[vertex] (a) at (0,0) {$a$};
      \node[vertex] (ab1) [right of = a] {$ab_1$};
      \node[vertex] (ab2) [right of = ab1] {$ab_2$};
      \node[vertex] (stop) [above right of = ab1] {\textsc{stop}};
      \node[vertex] (b) [right of = ab2] {$b$};
      \path[-]
         (a) edge [] node [] {} (ab1)
         (stop) edge [] node [] {} (ab2)
         (ab2) edge [] node [] {} (b)
         ;
  \end{tikzpicture}
  \end{center}
  \caption{Gadgets for reducing from a  \ruleset{Geography} edge $(a, b)$.  The top is the gadget as it will appear in the \textit{main-board} and all other boards except from $(a,b)$\textit{-board}.  The bottom is the edge as it appears in $(a,b)$\textit{-board}.  All other parts of the two boards will be the same.}
 \label{fig:QGeographyPolyWideHardGadget}
\end{figure}
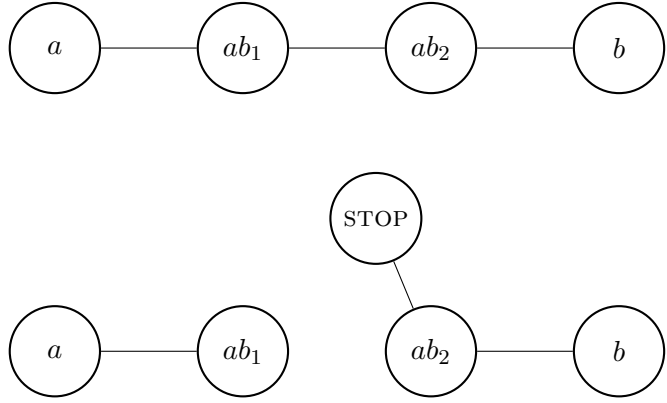

\begin{proof}
Because the game tree for \ruleset{Undirected Geography} has  height at most the number of nodes  in the geography graph,
by Theorem \ref{Theo:Upperbound},
\ruleset{Quantum Undirected Geography} with poly-wide quantum start is in \cclass{PSPACE}.

For hardness, we reduce from Quantum \ruleset{Geography} (which we prove to be PSPACE-hard for all flavors in
Theorem \ref{thm:geography} and Lemma \ref{lem:DAGGepgraphy}). This just a variant of Geography where we allow directed edges as well.

Consider a \ruleset{Geography} instance where
the underlying graph has $n$ nodes and $m$ edges.
We will create a superposition that consists $m + 1$ entangled realizations.
\begin{itemize}
    \item One realization, \textit{main-board}, has a very similar structure, but each arc $(a, b)$ is replaced by a three-part path with two new vertices, $ab_1$ and $ab_2$: $(a, ab_1)$, $(ab_1, ab_2)$, and $(ab_2, b)$.
    \item For each edge $(a,b)$, we include $(a,b)$\textit{-board}, which is exactly the same as \textit{main-board} except there is an extra vertex, \textsc{stop}, and the edge $(ab_1, ab_2)$ is replaced with $(\textsc{stop}, ab_2)$.
\end{itemize}
We show the two relevant edge gadgets in \cref{fig:QGeographyPolyWideHardGadget}.


The \textit{main-board} can never collapse out unless a player moves to \textsc{stop}, since all other edges are also in the other realizations. If a player does move to \textsc{stop}, then the game immediately ends, as there are no other edges leaving \textsc{stop} in that board. Thus, the game either ends when there are no moves left in the main game board, or \textsc{stop} is entered. Additionally, a player can move to \textsc{stop} if and only if their opponent reached the $ab_2$ vertex by moving from $b$, corresponding to going backwards on arc $(a, b)$ in the Quantum \ruleset{Directed Geography} game.  This cannot happen when moving from $ab_1$ since traversing the $(ab_1, ab_2)$ edge collapses out the realization with an edge to \textsc{stop}. Thus, traversing an edge backwards is always a losing move. Finally, when a game is lost in the corresponding Quantum \ruleset{Directed Geography} game, this Poly-Wide game will result in a loss within two moves, since no matter what move is chosen, either a move to \textsc{stop} or a move to $ab_2$ will win the game.

So, we can just follow a winning strategy for the instance of Quantum Directed Geography and win the game, and if we have a winning strategy for this game, then we can employ it in the instance of Quantum \ruleset{Directed Geography}
\end{proof}

Note that this only works for flavor D, since it relies upon a classical move to \textsc{stop}. However, we can simply replace this vertex with two vertices connected to $ab_2$, and then the same reduction works, by replacing moving to \textsc{stop} with making a quantum move to both of those vertices.

Note that the querying quantum position in our proof for the above theorem may not be reachable from a classical start.
The following theorem demonstrated a hardness result for determine the outcome class for quantum positions reachable from classical starts in \ruleset{Quantum Undirected Geography}.

\begin{theorem}[Intractability of Reachable Quantum Positions]
\ruleset{Quantum Undirected Geography} is \cclass{PSPACE}-complete measured by quantumized complexity with poly-moves after a classical start.
\end{theorem}

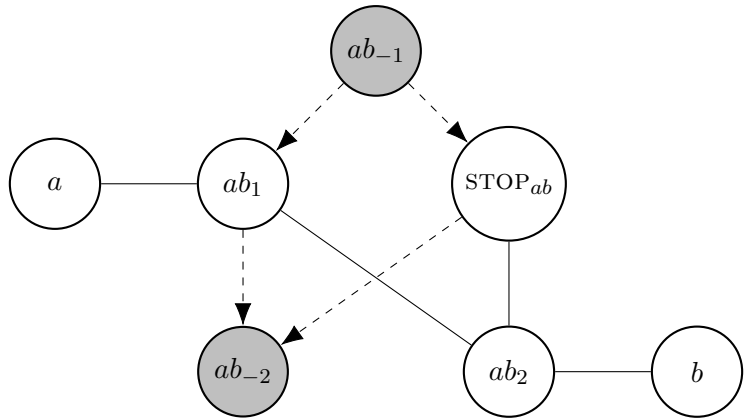
\begin{figure}[h]
  \begin{center}
  \begin{tikzpicture}[node distance = 2.5cm]
      \tikzstyle{vertex}=[circle,thick,draw = black, fill = white, minimum size=12mm]
      \tikzstyle{gray-vertex}=[circle,thick,draw = black, fill = lightgray, minimum size=6mm]

      \node[gray-vertex] (ab-1) at (0,0) {$ab_{-1}$};
      \node[vertex] (stopab) [below right of = ab-1] {\textsc{stop}$_{ab}$};
      \node[vertex] (ab1) [below left of = ab-1] {$ab_1$};
      \node[vertex] (a) [left of = ab1] {$a$};
      \node[vertex] (ab2) [below of = stopab] {$ab_2$};
      \node[vertex] (b) [right of = ab2] {$b$};
      \node[gray-vertex] (ab-2) [below of = ab1] {$ab_{-2}$};
      \path[dashed][-{Latex[length=3mm]}]
          (ab-1) edge [] node [] {} (ab1)
          (ab-1) edge [] node [] {} (stopab)
          (ab1) edge [] node [] {} (ab-2)
          (stopab) edge [] node [] {} (ab-2)
          ;
      \path[-]
         (a) edge [] node [] {} (ab1)
         (ab1) edge [] node [] {} (ab2)
         (stopab) edge [] node [] {} (ab2)
         (ab2) edge [] node [] {} (b)
         ;
  \end{tikzpicture}
  \end{center}
  \caption{Gadget for a Directed Geography edge $(a,b)$ in Undirected Geography.  Prior to the current position, there were moves $ab_{-1} \rightarrow \langle ab_1\ |\ \textsc{stop}_{ab} \rangle \rightarrow ab_{-2}$.}
  \label{fig:QGeographyRealizeableHardGadget}
\end{figure}

\begin{proof}
Here, we will be following the same idea as the original proof, only we will relax our assumption that their are only a polynomial number of starting realizations, but show that it can be reachable. We will refer to the realizations in our previous proof as \textit{core realizations}. We will call the remaining (exponential number of) realizations \textit{redundant realizations}.  For each redundant realization, $R$:

\begin{itemize}
\item There is a core realization such that if it collapses, so does $R$, and
\item At any given point, there is a core realization that contains all available moves in $R$.
\end{itemize}

By the first property, no redundant realization can ever be the only one remaining. By the second property, no core realization can be collapsed using moves only available due to the redundant realizations. Thus, the redundant realizations never provide any additional moves to either player.

To complete the proof, it remains to be shown that we can reach the point where all core boards exist, and the other boards are redundant.

For our new starting position, we will take all vertices and edges in our main board from the Poly-Wide reduction, and add vertices and all new edges. As shown in \cref{fig:QGeographyRealizeableHardGadget}, for each arc $(a, b)$ in QDG we can insert vertices $ab_1$, $ab_2$, and $\textsc{stop}_{ab}$, along with edges $(a, ab_1)$, $(ab_1, ab_2)$, $(\textsc{stop}_{ab}, ab_2)$, and $(ab_2, b)$.  In addition, we include two vertices, $ab_{-1}$ and $ab_{-2}$, which will have already been previously visited in all realizations by the time we reach the start.  These vertices are connected by edges $(ab_{-1}, ab_1)$, $(ab_{-1}, \textsc{stop}_{ab})$, $(ab_1, ab_{-2})$, and $(\textsc{stop}_{ab}, ab_{-2})$. For some other edge, $(c, d)$, we can connect the gadgets by setting $ab_{-2} = cd_{-1}$.

Now we prescribe a series of prior moves across all edges $\{(a_i, b_i)\ \mid\ i \in \{1,\ldots, m\}\}$: $(a_1b_1)_{-1} \rightarrow \langle ((a_1b_1)_1\ \mid \ \textsc{stop}_{a_1b_1} \rangle \rightarrow (a_1b_1)_{-2} = (a_2b_2)_{-1} \rightarrow \dots \rightarrow (a_{m-1}b_{m-1})_{-2} = (a_mb_m)_{-1} \rightarrow \langle (a_mb_m)_1 \mid\ \textsc{stop}_{a_mb_m} \rangle \rangle (a_mb_m)_{-2} \rightarrow x$  
were made, where $x$ is the starting vertex of QDG.

By construction, no realizations were ever collapsed in these prior moves. Additionally, there is a realization for each of the core realizations, by just following the branch that moved to \textsc{stop} for each edge gives us the main-board. And following the branch that moved to \textsc{stop} for each edge gadget other than some edge $e = (a, b)$ will give us the $(a, b)$-board.


To complete the proof, we only need to show that the rest of the realizations are redundant.  All realizations have only the edges in either the main-board or an $(a, b)$-board fulfilling the first property for redundancy. 
For the second property, we note that each of the other realizations must contain at least two \textsc{stop} vertices, so they collapse conditioned on the core realization on them. 



\end{proof}

\subsection{Solvability of Quantum Games With Classical Starts}

\label{Sec:PClassicalStartQUG}

\begin{theorem}[Quantum \ruleset{Undirected Geography} with Classical Starts in \cclass{P}]
For any classical \ruleset{Undirected Geography} position, the outcome class is the same as the outcome class of Quantum \ruleset{Undirected Geography} position on that same state (graph and vertex), in the quantum setting with flavor $D$ and superposition width $2$.
\end{theorem}\label{thm:quantumUGeographyClassicalStart}

In order to prove this, we will first show the algorithm that the winner in the classical game (\emph{hero}) will use to win under quantum play. After we describe the algorithm, we state and prove some useful lemmas, then we cover the proof of the theorem.

The hero will always make classical moves, but we need to show how they respond to quantum moves by their opponent, the \emph{villain}.  If the villain makes a quantum move, the hero will try to make a winning collapsing move (described further below).  If they cannot, they will keep track of the quantum superposition by contracting the two quantumly-chosen vertices into one combined vertex.

The hero keeps this record using an overlaid graph $G' = (V', E')$:
\begin{itemize}
    \item Initially, $V' = \{ \{v\}\ |\ v \in V\}$.
    We refer to $c(v)$ (the contraction with $v$) as the element  of $V'$ that contains $v$.
    \item $E'$ will be updated so that $(X, Y) \in E' \Leftrightarrow X, Y \in V'$ and $\forall x \in X, y \in Y: (x, y) \in E$.
    \item Whenever a player makes a classical move from $a \rightarrow b$ (meaning the current player makes a classical move after the previous player makes a classical move) the hero will remove $c(a)$ from $V'$ and all incident edges from $E'$.
    \item Whenever the villain makes a quantum move $a \rightarrow \langle v_1\ \mid\ v_2 \rangle$, the hero will again remove $c(a)$ from $V'$ and all incident edges from $E'$.
    \item Whenever the hero follows a quantum move with a classical move, $\langle v_1\ \mid\ v_2 \rangle \rightarrow b$, the hero will update $G'$ based on whether they make a collapsing move:
    \begin{itemize}
      \item If the hero collapses that quantum move, then they can remove the remaining $v_i$ as though it had been a classical move by the villain.
      \item If the hero does not collapse, but $c(v_1) = c(v_2)$, then all of those vertices have all been visited in all realizations.  The hero can remove $c(v_1)$ from $V'$ and all incident edges from $E'$.
      \item If the hero does not collapse and $c(v_1) \neq c(v_2)$, then the hero remove both $c(v_1)$ and $c(v_2)$ from $V'$ and replaces them with $c(v_1) \cup c(v_2)$.  Then the hero will reset $E'$ to match the definition given above.

    \end{itemize}
\end{itemize}

Next we describe how the hero chooses their move, using maximum matchings on $G'$.  (In our notation, we consider a matching, $M$, as both a set of pairs and a function.  So, $(a,b) \in M \Leftrightarrow M(a) = b \Leftrightarrow M(b) = a$.). In the classical version, the position is in \outcomeClass{N} iff the current vertex is in all maximum matchings on the graph \cite{DBLP:journals/tcs/FraenkelSU93}.  This means that after the classical winner's turn, the loser must start from a vertex not contained in some maximum matching.  Our hero will maintain a similar invariant on $G'$: there is a maximum matching on $G'$ such that the villain will be starting their turn on a vertex not in that matching.


Our algorithm is as follows:

\begin{itemize}
    \item If the villain makes a classical move to $v$, the hero considers any maximum matching, $M$ on $G'$, then moves to $x \in M(c(v))$ such that $(v,x) \in E$.  This leaves us with a maximum matching, $M \setminus \{(c(v), c(x))\}$ on the remaining graph that does not include $c(x)$, thus upholding the invariant.  The hero removes $c(v)$ and $c(x)$ from $V'$.
    \item If the villain makes a quantum move to $\langle\ a\ |\ b\ \rangle$, and $c(a) = c(b)$, then the hero acts as though the villain moved classically to only $a$, finding a maximum matching, $M$ on $G'$, then moving to $x \in M(c(a))$ such that $(a,x) \in E$.  Subtracting $(c(a), c(x))$ from $M$ results in a maximum matching without $c(x)$, upholding the invariant.  The hero removes $c(a) = c(b)$ and $c(x)$ from the partition $V'$ they are keeping track of.
    \item If the villain moves to $\langle\ a\ |\ b\ \rangle$, and $c(a) \neq c(b)$, then the hero has some additional work to do.  Notably, they have to find a matching, $M$, such that $\exists x \in M(c(a))$ where $(a, x) \in E$, but $(b, x) \notin E$; or $\exists x \in M(c(b))$ where $(b, x) \in E$, but $(a, x) \notin E$, if one exists.  There are three cases:
    \begin{itemize}
        \item If $\exists M$ with $x \in M(c(a))$ where $(a, x) \in E$, but $(b, x) \notin E$, then the hero moves to $x$.  Since $(b, x) \notin E$, the sets realizations where the villain moved to $b$ collapses out.  Subtracting $(c(a), c(x))$ from $M$ yields a maximum matching that does not include $c(x)$, upholding the invariant.  The hero removes $c(a)$ and $c(x)$ from $V'$.
        \item If $\exists M$ with $x \in M(c(b))$ where $(b, x) \in E$, but $(a, x) \notin E$, then the hero moves to $x$ as above, with the sets of realizations with a previous move to $a$ collapsing out.  Subtracting $(c(b), c(x))$ from $M$ yields a maximum matching that does not include $c(x)$, upholding the invariant.  The hero removes $c(b) \in E$ and $c(x)$ from $V'$.
        \item If no maximum matching exists with those requirements, then the hero can just use any maximum matching, $M$, to make their move.  The following must be true of $M$:
        \begin{itemize}
            \item $\forall x \in M(c(a)): (b,x) \in E$.
            \item $\forall y \in M(c(b)): (a,y) \in E$.
        \end{itemize}
        The hero can now move to any $x$ and update $V'$ by removing $c(x)$ and contracting $c(a)$ and $c(b)$ into one element $c(a) \cup c(b)$. In order to continue safely, the hero needs this new contracted vertex to be adjacent to $c(y)$, for any $y \in M(c(b))$. Thus, the hero needs both that $(c(a), c(y) \in E'$ and $(c(b), c(y)) \in E'$. The latter is already true because $M(c(b)) = c(y)$. We show the former using some extra lemmas. By \cref{lem:matching2matching}, $(a, x)$ and $(b, y)$ are in a maximum matching on $G$. Thus, there is another maximum matching with the swapped edges $(a, y)$ and $(b, x)$. Then, by \cref{lem:matching2contractedEdges}, $(c(a), c(y)) \in E'$.
    \end{itemize}
\end{itemize}

We use the following two lemmas to prove the efficacy of the hero's algorithm.


\begin{lemma}
  \label{lem:matchingSwapper}
  \label{lem:matching2matching}
    If $(c(a), c(b))$ is in a maximum matching of $G'$, then in any realization where $a$ and $b$ haven't been used, $(a,b)$ is part of a maximum matching on the unvisited graph.
\end{lemma}

\begin{proof}
  Let $M$ be the maximum matching on $G'$ containing $(c(a), c(b))$.  Also let $H$ be the remaining subgraph of $G$ that hasn't been visited in the given realization.  Then, for each element $X \in V'$, there is exactly one vertex in $H$ remaining.  Due to the definition of $E'$, that vertex must be adjacent to the vertex of $H$ inside the contraction $M(X)$.  Thus, each matched pair in $M$ corresponds to exactly one unique pair of neighbors in $H$, which creates a maximum matching on that graph.  Thus, $a$ and $b$ must be neighbors and $(a, b)$ is in the maximum matching on $H$.


\end{proof}

\begin{lemma}
  \label{lem:matching2contractedEdges}
  If $(a, b)$ is in a maximum matching on $G$ and $c(a) \neq c(b)$, then at any point in the game where $a$ and $b$ are still included in contractions, $(c(a), c(b)) \in E'$.
\end{lemma}

\begin{proof}
  We prove this by contradiction.  Assume $(c(a), c(b)) \notin E'$.  Thus, $\exists\ (a', b') \notin E$, where $a' \in c(a)$ and $b' \in c(b)$.
  Without losing generality, we assume that the most recent contraction breaks the statement of the lemma; all prior contraction-graphs $G'$ contained all edges from all maximum matchings in $G$. Assume that the villain's last quantum move was to $\langle x\ |\ y\rangle$ and the hero had to respond to a non-collapsing move at vertex $z$.  Thus:
  \begin{itemize}
    \item In the prior contraction, $c'$, $c'(x) \neq c'(y)$
    \item $c'(x) \cup c'(y) = c(a)$
    \item $\exists\ M$, a matching on $G'$ that the hero used to choose $z$.
    \item WLOG, $a' \in c'(x)$ and $b' \in M(c'(y))$
  \end{itemize}
\end{proof}

\begin{proof} (of Theorem \ref{thm:quantumUGeographyClassicalStart})




The invariant the hero maintains is: at the end of the hero's turn, after having moved to $x$, there is a maximum matching on $G'$ that does not contain the contracted vertex $c(x)$.  Thus, either
\begin{itemize}
    \item There are no more edges leaving $x$ and the villain loses immediately, or
    \item If there is a move and the villain moves from $x$ to $y$, then $y$ must be contained in one of the other matched pairs, meaning that the hero will be able to move to $y$'s match.  (Lemma \ref{lem:matching2matching} requires that $y$ is part of one of those matches, because if it wasn't, then $(x, y)$ would be part of a maximum matching on $G$ and that edge will be represented as an edge in a maximum matching in $G'$, which won't work with the invariant.)
\end{itemize}

When there are no moves left in $G'$, there are no moves left in $G$, since if the edge $(x, y)$ has $c(x) = c(y)$ then only one vertex is remaining in each realization, and if $c(x) \neq c(y)$, then $(c(x), c(y))$ must be in $G'$.

The invariant will be maintained because each turn the hero will start on a vertex in all maximum matchings of $G'$ and will traverse the edge from one of them.  Since the hero will always have a move to make on their turn, they will never lose the game.

\end{proof}

\section{Quantum Leap in \mbox{\ruleset{Nim}} Complexity}
\label{Sec:Nim}

In this section, we will prove that quantum moves profoundly impact the complexity of \ruleset{Nim}.
In the world of combinatorial games, \ruleset{Nim}
 has a very special---even magical---status.
Fundamental to the Sprague-Grundy theorem \cite{Sprague:1936,Grundy:1939}, closed-form expressions have been uncovered for \ruleset{Nim}'s outcome and winning strategy, providing a complete mathematical characterization of \ruleset{Nim} \cite{Bouton:1901}
as well as of all impartial games.
These expressions are also polynomial-time computable,
leading to efficient program
 to optimally play \ruleset{Nim}
 despite the fact that the depth of \ruleset{Nim}-game tree could be exponential in its descriptive size.\footnote{By contrast, most \cclass{PSPACE}-complete combinatorial games such as \ruleset{Hex}, \ruleset{Geography}, and \ruleset{Atropos} have game trees with height bounded linearly in their descriptive size.}
The striking contrast of \ruleset{Nim}'s polynomial-time solvability and many natural impartial games' {\cclass{PSPACE}}-complete intractability also illuminates the fundamental difference between mathematical characterization and computational characterization in combinatorial game theory.

By Proposition \ref{prop:Impartial},
\ruleset{Quantum Nim} remains an impartial game.
Thus, by the Sprague-Grundy theorem, \ruleset{Quantum Nim} can be mathematically characterized by classical \ruleset{Nim} (i.e.,  ``nimbers'').
However, as we shall prove below,  any such ``Sprague-Grundy-reduction'' from \ruleset{Quantum Nim} to the classical \ruleset{Nim} is unlikely computationally efficient.
In particular, we will show that its quantumized complexity with poly-wide quantum start  is above \cclass{NP}.
Formally, we prove the following theorem:\footnote{Unlike the theorems their proofs in other sections, this theorem and its proof below
only holds for quantum flavor $D$. 
}

\begin{theorem}[\ruleset{Quantum Nim} Leap Over \cclass{NP}]\label{theo:QNim}
The problem of determining outcome classes for
\ruleset{Quantum Nim} at a given poly-wide quantum position
is $\Sigma_2$-hard.\footnote{$\Sigma_2$ is a complexity class of decision problems on the second level of the {\em Polynomial-Time Hierarchy}---often denoted
by \cclass{PH}---which is a hierarchy of complexity classes for decision problems connecting \cclass{P} to \cclass{PSPACE}.
The base level of \cclass{PH} is \cclass{P}, and the first level of \cclass{PH} contains \cclass{NP} and \cclass{co-NP},
respectively, denoted by $\Sigma_1 = \cclass{NP} $ and $\Pi_1 = \cclass{co-NP}$.
Moving up the hierarchy,
for integer $k> 1$, $\Sigma_k$
 is the class of decision problems solvable
 by a \cclass{NP} Turing machine with an oracle access to a (complete)
   decision problem in $\Pi_{k-1}$.
Similarly, $\Pi_k$ is the
 is the class of decision problems solvable
 by a \cclass{co-NP} Turing machine with an oracle for a (complete) decision problem
   in $\Sigma_{k-1}$.
As $k$ increases to a polynomial magnitude, polynomial hierarchy $\Sigma_k$ and $\Pi_k$
   reach at \cclass{PSPACE}, by the fact that the general quantified Boolean formula
   (QBF) problem is complete for \cclass{PSPACE}.
In fact, one can precisely characterize polynomial-time hierarchy classes
  $\Sigma_k$ and $\Pi_k$ by quantified Boolean formula (QBF).
Generalizing the fact that \cclass{NP} and \cclass{co-NP} can be characterized,
  respectively, by logic decision of forms
$\exists\vec{x}\ F(\vec{x})$ and  $\forall \vec{x}\ F(\vec{x})$,
$\Sigma_k$ and $\Pi_k$ are precisely the languages characterized by logic decision of forms,
$\exists \vec{x_k}\forall \vec{x}_{k-1}\ ...\  F(\vec{x}_1,...,\vec{x}_k)$ and  $\forall \vec{x_k}\exists \vec{x}_{k-1}\ ...  \ F(\vec{x}_1,...,\vec{x}_k)$, respectively.
So, polynomial-time hierarchy captures the degree of {\em alternation structures} in complex decision problems.
There are other complexity classes, such as $\Delta_k$, associated with PH.
}
\end{theorem}




\subsection{Encoding \mbox{\rm \ruleset{Quantum Avoid True}} with \mbox{\rm \ruleset{Quantum Nim}}}

\label{Sec:QATToQNim}
Our proof of this theorem will go through the quantum lift of another combinatorial game, known in the literature as  \ruleset{Avoid True}.

As a key combinatorial game characterized in the seminal work of Schaefer \cite{DBLP:journals/jcss/Schaefer78}, \ruleset{Avoid True} is played on a (positive) conjunctive normal form (CNF).
In the initial game position, all ground-set variables are {\em free}, i.e., unassigned.
During the game, players take turns setting free variables to \texttt{true}.
A move
is feasible if setting the selected variable \texttt{true}
  will not make the CNF evaluate as a whole evaluate to \texttt{true}.
The next player loses if the position has no feasible move.

To establish Theorem \ref{theo:QNim},
  our first step is to efficiently encode
  each game position in \ruleset{Quantum Avoid True}
by a game position in  \ruleset{Quantum Nim}.
This encoding represents a significant step in using
superposition of \ruleset{Nim} positions---which individually is polynomial-time solvable---to characterize intractable combinatorial game positions.\footnote{See more discussions after the proof.}
The following theorem captures the details of the encoding.


\begin{theorem}[From \ruleset{Quantum Avoid True} to \ruleset{Quantum Nim}]
\label{QATQNIM}
For any superposition width $w$, any \ruleset{Avoid True} instance $Z$,
   and any position $\mathbb{B}$ in its quantum lift,
   $Z^{\QuantumLift(D,w)}$,
we can construct a poly-wide
\ruleset{Quantum Nim} position  $\mathbb{B'}$
such that
{\sc Decision of Outcome} for $\mathbb{B'}$ in \ruleset{Quantum Nim} (with flavor $D$)
yields the same result as {\sc Decision of Outcome} for $\mathbb{B}$ in $Z^{\QuantumLift(D,w)}$.
Moreover, this reduction can be computed in time polynomial in the
  descriptive size of $Z$ and $w$.
\end{theorem}
\begin{proof}
Below, we will focus on superposition width $w = 2$,
  as the proof can be naturally extended to other superposition width.

We start with our encoding of
  classical \ruleset{Avoid True} positions
  and then extend the encoding to
   all \ruleset{Quantum Avoid True} positions.
Recall that a classical \ruleset{Avoid True} position
 of instance $Z$ is defined by $Z$'s CNF $F$ and
 a subset $S$ of free variables.
The CNF is the AND of a set of OR-clauses
   (of only positive variables for \ruleset{Avoid True}).
Selecting any variable in a clause
  makes the evaluation of the clause \cclass{true},
  regardless of the situation with other variables.
Let $F_S$ be the {\em reduced} CNF,
  obtained from CNF $F$ by keeping
  only clauses with all free variables.
We will refer to clauses in $F_S$ as the
{\em active clauses} for position $(F,S)$.

To encode the \ruleset{Avoid True} position
  associated with $(F,S)$ with a
  superposition of \ruleset{Nim} positions,
  we use each active clause in $F_S$ to
  define a realization in the superposition
  by making each variable into a pile for
  that realization:
If a variable is in the clause or is not free,
  then set its corresponding pile to be of size 0
   in that realization;
  otherwise, set the corresponding pile to be of size 1.
For example, suppose the clause is
$X_1\vee X_2\vee X_3$,
  and free variables are
  $\{X_1,X_2,X_3,X_4,X_5\}$
  with the ground set variables
  $\{X_1,...,X_{10}\}$.
Then, the encoding defines the \ruleset{Nim}
  realization $[0, 0, 0, 1, 1, 0, 0, 0, 0, 0]$.

In general,  a quantum position
  $\mathbb{B} = \langle(F,S_1)\ |\ (F,S_2)\ |\ ... \ |\ (F,S_k)\rangle$
  in $Z^{\QuantumLift(D,2)}$ will be
  encoded by a wide superposition, $\mathbb{B'}$
  consisting of all \ruleset{Nim} realizations as
   defined above for each $F_{S_1}$,...,$F_{S_k}$, respectively.
Note that the encoding is polynomial-time constructible.

To prove the correctness of the reduction encoding,
  we  prove that the next player
  has a winning strategy at $\mathbb{B}$
  in \ruleset{Quantum Avoid True}
  if and only if the next player has a winning strategy at its encoding $\mathbb{B'}$
  in \ruleset{Quantum Nim}.
We will prove inductively
  that the following invariant always holds if we play
  our \ruleset{Quantum Nim} encoding 
  and its original \ruleset{Quantum Avoid True}
  position
in tandem:
For every active clause in each realization of
$Z^{\QuantumLift(D,2)}$
 there is exactly one realization in
 the \ruleset{Quantum Nim} encoding
 with the corresponding realization having piles of 0 for each variable in the clause and variables that is no longer free for the position.

As the basis case of the induction,
  the stated invariant is true at the start.
There is only one realization of
  \ruleset{Quantum Avoid True}.
In its \ruleset{Quantum Nim} encoding constructed above,
  each realization has piles of 0 for the variables
  in the corresponding clause, and all variables are free, so the rest of the piles are at 1.
Now, we will show that the invariant still
  holds after any move in \ruleset{Quantum Nim}
  by examining each effect of its
 ``coupled move'' on any possible clause.

\begin{itemize}
\item Any classical or quantum move targeting a pile of size 1 in the \ruleset{Quantum Nim} position
  does not collapse the \ruleset{Nim} realization,
consistent with the fact that selecting
  the corresponding variable results
   in the corresponding clause remaining active
in the \ruleset{Quantum Avoid True} instance.
After the move, the pile will now be at 0 in the
  \ruleset{Quantum Nim} realization, consistent with the fact that after selecting the corresponding
  variable in the realization
   of \ruleset{Quantum Avoid True},
  the variable is no long free.

\item Any classical or quantum move
   targeting a pile of 0 in the
  \ruleset{Quantum Nim} instance,
  that is not in the corresponding
  \ruleset{Quantum Avoid True} clause\footnote{in other words, the variable was taken previously in play}, causes this realization to collapse, just like it does in \ruleset{Quantum Avoid True} when the corresponding variable is selected.

\item If a classical or quantum move is made in the
 \ruleset{Quantum Nim} position on a pile of 0 that corresponds to a variable in the associated \ruleset{Quantum Avoid True} clause, then the clause is satisfied, and the corresponding realization collapses.
This matches the \ruleset{Quantum Avoid True} case where the clause is no longer active.
\end{itemize}
This invariant establishes the desired structural morphism between playing \ruleset{Quantum Avoid True} and playing its encoding \ruleset{Quantum Nim.}
\end{proof}

Note that this reduction doesn't hold for quantum flavors $C$ or $C'$, since it can make classical moves satisfying clauses at the start for \ruleset{Avoid True} in both, but they may collapse realizations in our instance of \ruleset{Quantum Nim}.


\subsection{Subtlety in Quantum Moves and Complexity Impact}

Our reduction from \ruleset{Quantum Avoid True} to  \ruleset{Quantum Nim} highlights some important
 aspects of quantum combinatorial games.
We now discuss them before
  proceeding to our next step aiming to
  characterize the complexity of
  \ruleset{Quantum Avoid True}.

\vspace{0.05in}
\noindent{\bf Expressiveness of \ruleset{Nim} Superposition}: In the classical setting,
 because \ruleset{Avoid True} is a
\cclass{PSPACE}-complete game and \ruleset{Nim} is a polynomial-time solvable  game, any reduction from  \ruleset{Avoid True}
  game positions to \ruleset{Nim} games must incur an {\bf\em exponential explosion}
    (unless $\cclass{PSPACE} = \cclass{P}$).
In contrast, our reduction efficiently characterizes each \ruleset{Quantum Avoid Truth} position by a superposition of a polynomial number of \ruleset{Nim}
  positions, each with a polynomial descriptive size.
Therefore, our polynomial-time reduction above from \ruleset{Quantum Avoid True} to \ruleset{Quantum Nim} sheds light sharply on the structural power of quantum
 superposition.
Our proof, in particular, illustrates
  the capacity of
  \ruleset{Quantum Nim} positions
  in succinct characterization of
  complex logical alternations intrinsic to \ruleset{Quantum Avoid True}.

\vspace{0.05in}
\noindent{\bf Game Position vs Game Alternation}:
Note that our reduction is not from \ruleset{Avoid True} to \ruleset{Quantum Nim}, but rather,
  and critically, is from \ruleset{Quantum Avoid True}
  to \ruleset{Quantum Nim}.
The reduction can encode every \ruleset{Avoid True} position
 with a poly-wide
 \ruleset{Quantum Nim} positions, and
if the \ruleset{Quantum Nim} can somehow ``forbids''---figuratively and mathematically---
 quantum moves from this position onwards,\footnote{More on this, see Section \ref{Sec:FinalRemarks}} then the ``classical continuation'' of
  the \ruleset{Quantum Nim} position
  precisely captures the alternation
  in \ruleset{Avoid True}.
However,  quantumness does matter in the resulting \ruleset{Quantum Nim}.
Hence, instead of preserving \ruleset{Avoid True}'s alternation structure,
  the reduction preserves
  \ruleset{Quantum Avoid True}'s game tree.

\vspace{0.05in}
\noindent{\bf A Seemingly Intuitive Conjecture}:
The need to understand the computational complexity of \ruleset{Quantum Avoid True} in order to determine the intractability of \ruleset{Quantum Nim} also brings us to a family of fundamental questions, centered around
a seemingly intuitive conjecture:
\begin{quote}
{\em Given more options to play strategically, combinatorial games is always as least as hard (computationally) in the quantum setting than in the classical setting?}
\end{quote}
or its more concrete analog:
\begin{quote}
{\em The quantum lift of any \cclass{PSPACE}-complete combinatorial
    game remains a \cclass{PSPACE}-hard game.}
\end{quote}

Had these conjectures being true, we would have proved---via Theorem \ref{QATQNIM}---that ``\ruleset{Quantum Nim} at a poly-wide superposition is a \cclass{PSPACE}-hard to evaluate.''

In Section \ref{Sec:QuantumCollapses},
   we will refute these
  intuitive conjectures by providing strong
  complexity-theoretical evidences
  that some combinatorial games
  might be strictly harder to play optimally than
  their quantum counterparts.
This counterintuitive complexity-theoretical result
  highlights the subtle impact of quantum moves to
  alternation structures of combinatorial games,
  as well as
  the mathematical/computational challenges
  in establishing hardness results for quantum games.

In this context, although \ruleset{Avoid True} is \cclass{PSPACE}-complete---and we will prove in Section \ref{Sec:PSpaceUpperBound} that
 \ruleset{Quantum Avoid True}
 is still \cclass{PSPACE}-solvable---whether or not \ruleset{Quantum Avoid True} is \cclass{PSPACE}-complete
remains an outstanding open question.

\subsection{Baseline of \mbox{\rm \ruleset{Quantum Nim}} in Polynomial-Time Hierarchy}

\label{Sec:QuantumSchaefer}

We now prove Theorem \ref{theo:QNim}, showing that
{\sc Decision of Outcome} of poly-wide \ruleset{Quantum Nim} positions
 is intractable, with complexity
 between \cclass{NP} and \cclass{PSPACE}.
  We will reduce to it---via \ruleset{Quantum Avoid True}---from the quantum lift of a game called \ruleset{Phantom-Move QSAT},
  a ``partition-free'' combinatorial game based on
 the quantified Boolean formula problem.
   \footnote{The nomenclature will be justified later on in Section 
\ref{Sec:QuantumCollapses}.}

\ruleset{Phantom-Move QSAT} is a partizan game played on a CNF, whose clauses may contain both positive and negative literals.
In this game, one player---the {\em True Player}---aims
  to satisfy the CNF formulae,
  while the other player---the {\em False Player}---wants the formulae to be unsatisfied.
The ground set variables are partitioned into two sets, True-Player's variables and False-Player's variables (either with the same size or one more in the first set) for the two players, respectively, to assign values to. Starting with the True player,
  the two players alternate turns setting
  one of their variables to \texttt{true} or \texttt{false}.
If the current player has no more free variable to assign but satisfied the player's win condition,
then the player are still may take
  an extra feasible move---the {\em phantom move}---at the end.

In other words, if the current player is the True Player\footnote{i.e., in the case when they both have the same number of variables.} \& the existing assignments already set the formulae \texttt{true} or the current player is the False Player\footnote{i.e., in the case when the True player have one more variable.} \& the existing assignments already set the formulae \texttt{false}, then the player ``earns'' the right of a phantom move; otherwise, without any feasible move at the position,
the current player loses the game.

Equivalently one can define the "phantom move" to be included with the final variable selection,
 such that the player that would assign the final variable may only do so if they will have reached their win condition upon playing that variable.
One can reduce from the other Phantom move variant by introducing a variable $v_{n+1}$ which only appears in a clause as $(v_{n+1} \vee \neg v_{n+1})$ given to the phantom move player.
So the hardness remains the same.
We will be reducing from this variant of the game.

Classically, this game, as proved by Schaefer \cite{DBLP:journals/jcss/Schaefer78},
 is PSPACE-complete.
We will prove in Section \ref{Sec:QuantumCollapses}
  that its quantum extension is
   $\Sigma_2$-complete or $\Pi_2$-complete, depending on which player is making the final variable assignment. When restricted to the cases where the final variable is assigned by the False Player, it is $\Sigma_2$-complete. When assigned by the True Player, it is $\Pi_2$-complete.
Below, we will use this fact to complete the proof of Theorem \ref{theo:QNim} by
  establishing the following theorem.
\begin{theorem}[Baseline Complexity of \ruleset{Quantum Avoid True}] \label{thm:avoidTrue}
    \ruleset{Quantum Avoid True} is
      $\Sigma_2$-hard in quantum flavor $D$ with any fixed quantum width $w \geq 2$
\end{theorem}
\begin{proof}
Consider a \ruleset{Phantom-Move QSAT} instance $Z$ with CNF.
In the proof, we will focus on the case when both players
  have the same number of variables, which is an even number. This is acceptable since the hardness results remain even when fixed to any arbitrary parity.
Recall that the player that goes first will be called the \textit{True player} and the player that goes second the \textit{False player}.
Below in both $Z$ and its quantum lift $Z^{\QuantumLift(D,w)}$, we will denote the True Player's variables
by $\{T_1,...,T_m\}$ and
  the False Player's variables by $\{F_1,...,F_m\}.$

In our proof, we will show that Schaefer's remarkable
 reduction from \ruleset{Phantom-Move QSAT} to
 \ruleset{Avoid True} in the classical setting can be quantum lifted into a \ruleset{Quantum Avoid True} instance to encode $Z^{\QuantumLift(D,w)}$:

Schaefer's reduction (in our notation) is the following:
\begin{itemize}
  \item For each True Player's variable $T_i$, we introduce two new variables $T_{i_1}$ and $T_{i_2}$, and create a positive clause $(T_{i_1} \vee T_{i_2})$. We will collectively call these clauses \textit{TV clauses}. $T_{i_1}$ will represent assigning a true literal, and $T_{i_2}$ will represent assigning a false literal.

  \item For each False-Player's variable $F_i$, we introduce three new variables $F_{i_1}$, $F_{i_2}$, and $F_{G_i}$, and create a clause  $(F_{i_1} \vee F_{i_2} \vee F_{iG})$.  We will collectively call these clauses \textit{FV clauses}. As before, $F_{i_1}$ will represent assigning a true literal, and $F_{i_2}$ will represent assigning a false literal. The $F_{iG}$ variable is simply an arbitrary variable for the purposes of ensuring a certain clause parity (which we will give the motivation for later). It will function as an alternate truth literal assignment, as it appears in every clause $F_{i1}$ does.

\item For each of the clauses in $Z$, we replace each instance of positive literal $T_i$ and $F_i$ with $T_{i_1}$ and $F_{i_1}$, respectively; we replace each instance of negated variable $\neg T_i$ and $\neg F_i$ with $T_{i_2}$ and $F_{i_2}$, respectively. We will collectively call these clauses \textit{QBF clauses}.

  \item For each instance of a variable $T_{i_1}$, $F_{i_1}$, $T_{i_2}$, and $F_{i_2}$ in the QBF clauses, we add, to same clause, variables $T_{i_1}'$, $F_{iG}$ $T_{i_2}'$, or $F_{i_2}'$, respectively. We will collectively call these new variables \textit{duplicate variables}.
  These serve both as a way to ensure the QBF clauses all have even parity, and for player strategy, as we will cover later.
\end{itemize}

Through our proof, we will let $X$ denote the \ruleset{Avoid True} instance obtained from Schaefer's reduction of \ruleset{Phantom-Move QSAT} instance $Z$. We will call the first player in $X$ Player True and the second player Player False. Schaefer proved that True can win (the impartial) $X$ if and only if True can win (the partizan) $Z$. Below, we will extend Schaefer's proof to show that their quantum lifts $Z^{\QuantumLift(D, w)}$ and $X^{\QuantumLift(D, w)}$ has the same winner (when optimally played).

Before going on with our proof, we first recall one of the key properties, formulated by Schaefer, and extend it to the quantum setting.

\begin{lemma}[Avoid-True Destiny]
For any classical position created by Schaefer's reduction, if all unsatisfied clauses have an even number of variables, then Player True will win, under arbitrary play by both players for the remainder of the game, and if all unsatisfied clauses have an odd number of variables, then Player False will win, under arbitrary play by both players for the remainder of the game.
\end{lemma}

\begin{proof}\footnote{We include a proof here both in order to provide readers more intuition about Schaefer's reduction and to add that it holds under arbitrary play (which wasn't explicitly stated by Schaefer).}\label{AvoidTrueArbitrary}
Let's first note that in this game, there are always an even number of total unsatisfied variables remaining on True's turn and an odd number of unsatisfied variables remaining on the False's turn.
This is because the total number of variables in the game is even at the start, and True goes first.
Thus, since exactly one variable is assigned each turn, the parity of remaining variables is maintained.

The proof then is simple. If there are even number of unassigned variables remaining, and all the unsatisfied clauses have an odd number of variables in them, then there must exist an unassigned variable not in at least one of the clauses.
Similarly, if there are an odd number of variables remaining, with all remaining clauses having an even number of variables, than there must exist an unassigned variable no in at least one of the clauses.
Since this property holds under induction for arbitrary play, one can't lose if there is a move remaining, the game length is finite, and as previously stated, the parity of number of variables remaining is equivalent to the player, this completes the proof.
\end{proof}

The above lemma shaped the ``parity'' design used in Schaefer's reduction.
It also defines a scenario satisfying the condition of ``quantumness doesn’t matter.''
Combining this with \cref{mainDoesn'tMatter}, we get:

\begin{corollary}[Quantum-QBF Destiny]\label{AvoidTrueTermination}
For any position reached in Schaefer's reduction,
if in all realizations, there are only TV clauses and/or QBF clauses unsatisfied, then False wins;  if in all realizations, there are only FV clauses unsatisfied, than True wins.
\end{corollary}

To prove the correctness of this reduction in the quantum setting, we simply need to prove that if True has a winning strategy in the instance of $Z^{Q(D, w)}$ that we are reducing from, then True has a winning strategy in $X^{Q(D, w)}$, and that if False has a winning strategy in $Z^{Q(D, w)}$, then False has a winning strategy in $X^{Q(D, w)}$.

In the proof below, we will crucially use the following fact that we will establish in 
Observation \ref{TrueStrategy}: True has a winning strategy in \ruleset{Quantum  Phantom-Move QSAT} if and only if True has a single assignment (of only classical moves) that can always win, regardless of what moves False makes.

So, if True is the winner in $Z^{Q(D, w)}$, we perscribe the following strategy for True in $X^{Q(D, w)}$:

\begin{enumerate}
\item Assign variables in the TV clauses according the winning strategy for $Z^{Q(D, w)}$
\item Assign all of the $F_{i2}'$ variables.
\end{enumerate}

Our proof is based on the following key observation is: If True is able to complete this strategy to completion, then all TV and QBF clauses must be satisfied, resulting in a True win by \cref{AvoidTrueTermination}.

To proceed, we just need to show that True's strategy is always a legal set of moves.
First, True must be able to satisfy all $m$ TV clauses, because False can't assign all $m$ FV clauses before True, thus whatever TV clause true assigns last can't be the last.
Then, because that variable assignment in $Z$ must have satisfied all clauses regardless of how the FV variables were assigned, that means in no realization is it legal to make a move that would satisfy all FV clauses, as all QBF clauses must be satisfied at the same time. Since all realizations still have the FV clauses unsatisfied, True can then assign all $F_{i2}'$ variables, since none of them appear in an FV clause.
So, that strategy is always achievable.

Now, for the second case.
Suppose that False wins the instance $Z^{Q(D,w)}$.
Then, False applies the following strategy:

\begin{enumerate}
\item For False's move $i \leq m - 1$, if $F_{i1}$ and $F_{i2}$ are not yet classically assigned, make a quantum move of $[F_{i_1} \mid F_{i_2}]$.
Otherwise, play arbitrarily on the FV clauses not yet assigned in all realizations.
\item For the final move, assign the variable in final FV to either \texttt{true} or \texttt{false}, whichever is a legal move. If both are, choose arbitrarily.
\end{enumerate}

Note that if False is able to carry out this strategy, then by
\cref{AvoidTrueTermination},
they will win the game, as they will have satisfied all FV clauses leaving only TV and/or QBF clauses remaining.

For the proof of the correctness of this strategy, first note that, as we will prove in 
Observation \ref{FalseStrategy}, if False has a winning strategy in  \ruleset{Quantum  Phantom-Move QSAT}, then repeatedly choosing, for an arbitrary variable, a quantum assignment of $\langle\texttt{true}\ \mid \ \texttt{false}\rangle$ is also a winning strategy.

So, if True has only played consistently, then False will be able to make a move on the final unsatisfied FV clause, as there exists a realization where not all QBF clauses are satisfied. If True ever makes a move with at least one realization that isn't on a variable in a TV clause, then when False moves to the final FV clause, there exists a realization where True hasn't assigned all of the TV clauses, so then False can make their final move arbitrarily on the clause.
\end{proof}

We can extend this result to each of the quantum flavors rather simply. To do so, note that if we modify the reduction so that every variable is duplicated 3 times (so that each instance variable $v_j$ appears also has $v_j^{\ast}$ and $v_j^{\ast \ast}$ appear in that clause), then the exact same proof holds for D, as each parity argument remains, and neither player's winning strategies invoke these at all until quantumness doesn't matter.

So, by \cref{StrongFlavor}, since these games have properly shaped game trees, the other flavors are also $\Sigma_2$-hard.

\subsection{Nim Encoding and Structural Witness of $\Sigma_2$-Hardness in Quantum Games}

\label{Sec:NimEncoding}

In our reduction from $\Sigma_2$-hard \ruleset{Quantum Avoid True},
our \ruleset{Quantum Nim} game
uses a poly-wide superposition of perhaps the simplest \ruleset{Nim} positions:
all \ruleset{Nim} positions that we used in our encoding are from the family of \ruleset{Boolean Nim}.
So, our $\Sigma_2$-hard intractability proof
of \ruleset{Quantum Nim} holds for \ruleset{Quantum Boolean Nim}
and can be extended broadly to
the quantum extension of
\ruleset{Subtraction} games, which are widely used to teach young children about mathematical thinking.

In this section, we present a more systematic theory to extend these hardness results.
Particularly, we can use the $\Sigma_2$-hard result of \ruleset{Quantum Nim} to get results for several other games in the quantum setting at poly-wide quantum positions. We show that if a game is able to classically properly embed a binary game with properties that we will describe, then its quantumized complexity at a poly-wide superposition is at least $\Sigma_2$-hard.

\begin{definition}%
[Robust Binary-Nim Encoding]
A ruleset ${\cal R}$ has a {\em robust binary-Nim encoding} if 
${\cal R}$ has a position $b$  with the following properties for any $n$.

\begin{itemize}
\item $b$ has a set $M$ of $\lvert M \rvert = n$ distinct feasible moves. 
\item
For each $\sigma\in M$,
selecting $\sigma$ moves the game from $b$ to
 a position $b_{\sigma} \in P$, such that position $b_{\sigma}$ has $(n-1)$ feasible moves given by $M \setminus \{\sigma\}$.
In addition,
$b_{\sigma}$ induces a robust binary-Nim encoding for $(n-1)$.
\item When $n=0$, the game will have no available moves.
\end{itemize}
\end{definition}

\begin{lemma}
If a game ${\cal R}$ has a robust binary-Nim encoding, then it is $\Sigma_2$-hard to determine
the winability  of poly-wide quantum positions
in its quantum setting.
\end{lemma}
\begin{proof}
We reduce from \ruleset{Quantum Boolean Nim}
 at a poly-wide superposition.
 Suppose there are $n$ piles and a superposition with $m$ realizations
 $\langle b_1 \ | \cdots \ | b_m\rangle$.

 Now we focus on the position $b$ in
 ${\cal R}$ that induces the robust binary-Nim encoding for $n$, with set of moves
 $M =\{\sigma_1, \ldots ,\sigma_n\}.$
 For each pile $i$, we associate it with move $\sigma_i$.
Each realization, $b\in \{b_1,\ldots ,b_m\}$, in \ruleset{Quantum Boolean Nim} defines
a set  $S\subseteq [n]$, corresponding to piles with value 1.
In the reduce quantum position for ${\cal R}$,
we make a realization with the position whose feasible moves are 
$\{\sigma_i | s\in S\}$.

This is a direct encoding, where we are just relabeling move $i$ to $\sigma_i$,
setting up  the
desired morphism between playing
\ruleset{Quantum Boolean Nim} and playing its encoding in the quantum setting of ${\cal R}$.
\end{proof}

We can easily embed \ruleset{Boolean Nim} in several games, even games with fixed outcome, like \ruleset{Brussel Sprouts}. For "interesting" games, such as \ruleset{Go} and \ruleset{Domineering}, this process is often simple, as one needs only to be able to make a board state a sum of nimbers with only one move in each of them.

\section{Quantumness Matters}\label{Sec:QuantumMatter}

When we 
allow quantum moves using one of the flavors, it may or may not effect the outcome class of the game.
We define terms related to whether outcome classes and strategies change under quantum moves.  Each term can relate to either a ruleset or a position; we define both.
\theoremstyle{definition}

\begin{definition}{(Quantumness Strongly Matters)}
Given  
a position $b$ in a classical combinatorial game ruleset ${\cal R}$,
if the outcome classes at $b$ are different in the classical and quantum setting,
(i.e., $o_Z(b)\neq o_{Z^{\QuantumLift}}(b)$ where $Z$ denotes the game instance with starting position $b$ and $Z^{\QuantumLift}$
be the quantum extension of $Z$),
then we say that \emph{quantumness strongly matters} for position $b$ in ${\cal R}$ (or
simply for  $Z^{\QuantumLift}$ and $Z$).

We say that \emph{quantumness strongly matters} for a ruleset, ${\cal R}$, if there exists a position in ${\cal R}$ where quantumness strongly matters.
\end{definition}

\begin{figure}[h]
\begin{center}
  \def\scaleAmt{1.3}
  \scalebox{.9}{
  \begin{tikzpicture} [node distance = 1cm]
    \node (start) at (0, .5)  {\scalebox{\scaleAmt}{
        \begin{tabular}{@{}c@{}}
            $(2,2)$ \\
            \multicolumn{1}{c}{$\outcomeClass{N}$}
        \end{tabular}}};
    \node (first) at (0, -2) {\scalebox{\scaleAmt}{
        \begin{tabular}{@{}l@{}c@{}}
            $\langle\ (1,2)$\\ $|\ (2,1)\ \rangle$ \\
            \multicolumn{1}{c}{$\outcomeClass{P}$}
        \end{tabular}}};
    \node (secondA) at (-5, -5) {\scalebox{\scaleAmt}{
        \begin{tabular}{@{}l@{}l@{}c@{}}
            $\langle\ (0,2)$\\ $ |\ (1, 0)$\\ $|\ (1,1)\ \rangle$ \\
            \multicolumn{1}{c}{$\outcomeClass{N}$}
        \end{tabular}}};
    \node (secondB) at (-2, -5) {\scalebox{\scaleAmt}{
        \begin{tabular}{@{}l@{}c@{}}
            $\langle\ (0,2)$\\ $|\ (1,1)\ \rangle$\\
            \multicolumn{1}{c}{$\outcomeClass{N}$}
        \end{tabular}}};
    \node (secondC) at (0, -5) {\scalebox{\scaleAmt}{
        \begin{tabular}{@{}c@{}}
            $(0, 1)$ \\ $\outcomeClass{N}$
        \end{tabular}}};
    \node (secondD) at (2, -5) {\scalebox{\scaleAmt}{
        \begin{tabular}{@{}l@{}l@{}c@{}}
            $\langle\ (0,2)$\\ $|\ (1,1)$\\ $|\ (0,1)\ \rangle$ \\
            \multicolumn{1}{c}{$\outcomeClass{N}$}
        \end{tabular}}};
    \node (secondE) at (6, -5) {\scalebox{\scaleAmt}{
        \begin{tabular}{@{}l@{}l@{}c@{}}
            $\langle\ (0,2)$\\ $|\ (1,1)$\\ $|\ (2,0)\ \rangle$ \\
            \multicolumn{1}{c}{$\outcomeClass{N}$}
        \end{tabular}}};
    \node (zero) at (0, -8) {\scalebox{\scaleAmt}{
        \begin{tabular}{@{}c@{}}
            $(0,0)$\\ $\outcomeClass{P}$
        \end{tabular} } };

    \path[->]
        (start) edge [] node [right, text width=1cm] {$\langle\ (-1, 0)$\\ $|\ (0, -1)\ \rangle$} (first)
        (first) edge [bend right] node [left, text width = 1cm, xshift = -28pt] {$\langle\ (-1, 0)$\\ $|\ (0, -2)\ \rangle$} (secondA)
        (first) edge [bend right] node [above, xshift = -15pt] {$(-1, 0)$} (secondB)
        (first) edge [] node [left] {$(-2, 0)$} (secondC)
        (first) edge [bend left] node [right, text width = 1cm, xshift = 6pt] {$\langle\ (-1, 0)$\\ $|\ (-2, 0)\ \rangle$} (secondD)
        (first) edge [bend left] node [right, xshift = 20pt, text width = 1cm] {$\langle\ (-1,0)$\\ $|\ (0, -1)\ \rangle$} (secondE)
        (secondA) edge [bend right] node [left, xshift = -5pt] {$(0,-2)$} (zero)
        (secondB) edge [bend right] node [left] {$(0,-2)$} (zero)
        (secondC) edge [] node [left] {$(0, -1)$} (zero)
        (secondD) edge [bend left] node [right, xshift = 5pt] {$(0, -2)$} (zero)
        (secondE) edge [bend left] node [right, xshift = 10pt] {$(-2, 0)$} (zero)
        ;

  \end{tikzpicture}}
 \end{center}
  \caption{Winning strategy for Next player in quantum \ruleset{Nim} $(2,2)$, showing that Quantumness Strongly Matters.  (There are four additional move options from $\braket{(1,2)\ |\ (2,1)}$ that are not shown because they are symmetric to moves given.)}
  \label{fig:nim22}
\end{figure}
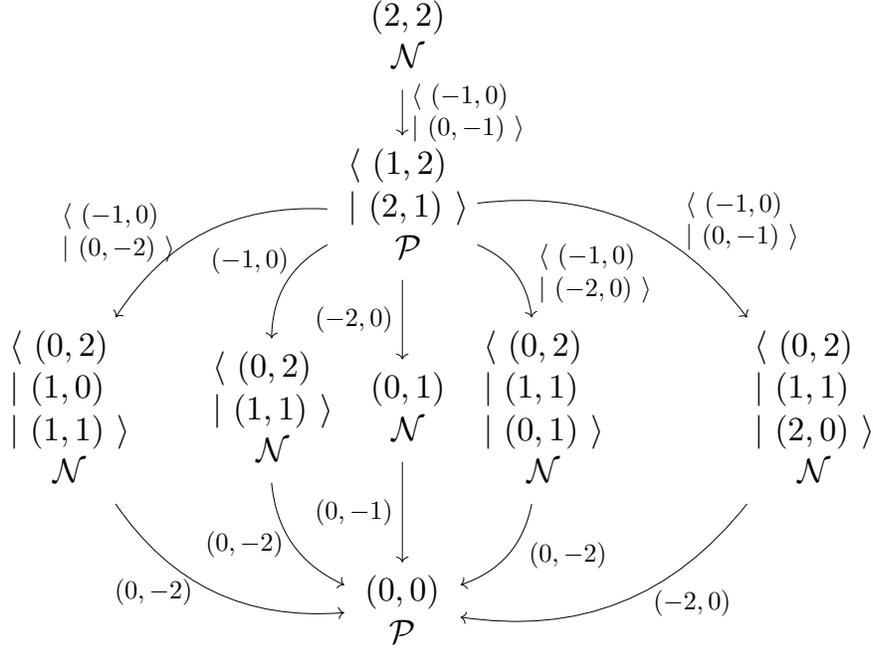

As an example of this, consider the \ruleset{Nim} position $(2,2)$ (2 sticks in each of two piles).  Under normal (non-quantum) play, this game is in $\outcomeClass{P}$, because the next player's moves can be mimicked by the previous player's until all piles are zero.  Under quantum play, the first player can win.  We show this winning strategy in \cref{fig:nim22}.



\begin{definition}{(Quantumness Weakly Matters)}
Given a position $b$ in a classical ruleset ${\cal R}$
we say that \emph{quantumness weakly matters} for position $b$
if:
\begin{itemize}
    \item The outcome classes at $b$ are the same in the classical and quantum setting
    (i.e., $o_Z(b) = o_{Z^{\QuantumLift}}(b)$
    where $Z$ denotes the game instance with starting position $b$ and $Z^{\QuantumLift}$ denotes the quantum extension of $Z$), and
    \item The winning strategy from $b$ will not suffice for the same player to win in  $Z^{\QuantumLift}$.
    In other words, from either
    $Z^{\QuantumLift}$ or some classical follower,
   $\bar{Z}^{\QuantumLift}$ of $Z^{\QuantumLift}$
    the winning player can only win by moving to an option
    $X$, that was not a winning move in the classical version.  It can be the case that
    either
    $X$ is a quantum superposition, or quantum strongly matters for  position $X$.

We say that \emph{quantumness weakly matters} for a ruleset, ${\cal R}$, if there exists a position of ${\cal R}$ where quantumness weakly matters.
\end{itemize}

In other words, we say that for a classical position, $b$, \emph{quantumness weakly matters} if the outcome class for that position 
is the same in classical and quantum settings,
but the winning strategy for the winning player changes in the quantum setting.
\end{definition}


An example of \emph{Quantum Weakly Matters} is \ruleset{Nim} position $(4,2)$.  This is a \outcomeClass{N}-position both classically and quantumly.  The classical winning first move is $(-2, 0)$ to $(2, 2)$.  With quantum moves, however, this is not a winning move, as $(2, 2)$ is an \outcomeClass{N}-position (see \cref{fig:nim22}).  Instead, the move of $(-1, 0)$ to $(3, 2)$ is the correct winning move.  We can show that  \outcomeClass{P}-position by showing that all options are \outcomeClass{N}-positions:

\begin{itemize}
    \item The positions resulting from classical moves are $(0, 2)$, $(1, 2)$, $(2, 2)$, $(3, 1)$, and $(3, 0)$.  $(2,2)$ has a winning move as shown in \cref{fig:nim22}.  The others all have winning moves to either $(0,0)$ or $(1,1)$, both of which are \outcomeClass{P}-positions.
    \item The positions using quantum moves are shown in \cref{fig:nim42} to be \outcomeClass{N}-positions.
\end{itemize}

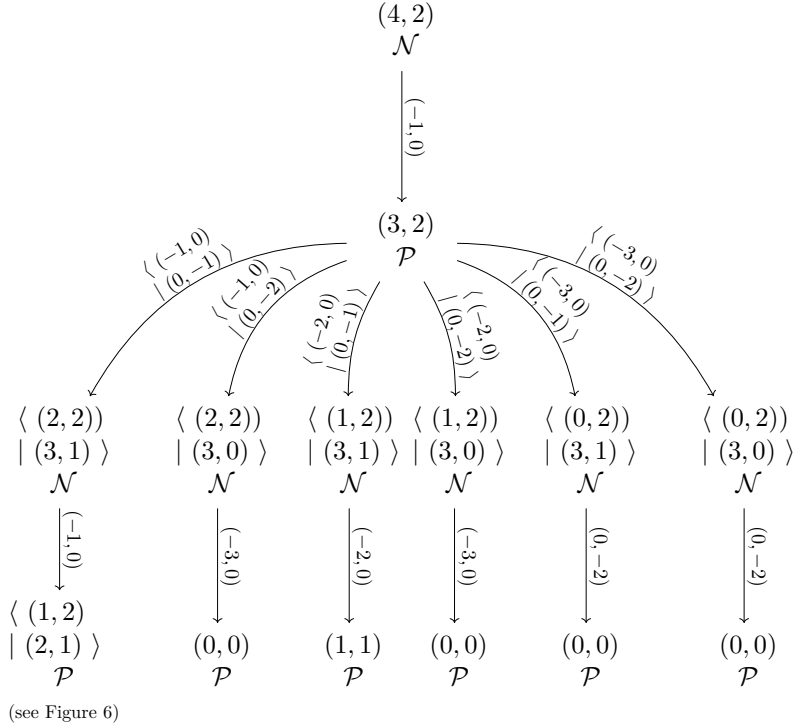
\begin{figure}[h]
\begin{center}
\scalebox{.7}{
  \begin{tikzpicture} [node distance = 1cm]
    \node (preStart) at (0, 8)  {\scalebox{\scaleAmt}{
        \begin{tabular}{c}
            $(4, 2)$ \\
            $\outcomeClass{N}$
        \end{tabular}}};

    \node (start) at (0, 4)  {\scalebox{\scaleAmt}{
        \begin{tabular}{c}
            $(3, 2)$ \\
            $\outcomeClass{P}$
        \end{tabular}}};

    \node (firstA) at (-6.5, 0) {\scalebox{\scaleAmt}{
        \begin{tabular}{l}
            $\langle\ (2,2))$ \\
            $|\ (3,1)\ \rangle$\\
            \multicolumn{1}{c}{$\outcomeClass{N}$}
        \end{tabular}}};

    \node (firstB) at (-3.5, 0) {\scalebox{\scaleAmt}{
        \begin{tabular}{l}
            $\langle\ (2,2))$ \\
            $|\ (3,0)\ \rangle$\\
            \multicolumn{1}{c}{$\outcomeClass{N}$}
        \end{tabular}}};

    \node (firstC) at (-1, 0) {\scalebox{\scaleAmt}{
        \begin{tabular}{l}
            $\langle\ (1,2))$ \\
            $|\ (3,1)\ \rangle$\\
            \multicolumn{1}{c}{$\outcomeClass{N}$}
        \end{tabular}}};

    \node (firstD) at (1, 0) {\scalebox{\scaleAmt}{
        \begin{tabular}{l}
            $\langle\ (1,2))$ \\
            $|\ (3,0)\ \rangle$\\
            \multicolumn{1}{c}{$\outcomeClass{N}$}
        \end{tabular}}};

    \node (firstE) at (3.5, 0) {\scalebox{\scaleAmt}{
        \begin{tabular}{l}
            $\langle\ (0,2))$ \\
            $|\ (3,1)\ \rangle$\\
            \multicolumn{1}{c}{$\outcomeClass{N}$}
        \end{tabular}}};

    \node (firstF) at (6.5, 0) {\scalebox{\scaleAmt}{
        \begin{tabular}{l}
            $\langle\ (0,2))$ \\
            $|\ (3,0)\ \rangle$\\
            \multicolumn{1}{c}{$\outcomeClass{N}$}
        \end{tabular}}};

    \node (secondA) at (-6.5, -4) {\scalebox{\scaleAmt}{
        \begin{tabular}{l}
            $\langle\ (1, 2)$ \\
            $|\ (2, 1)\ \rangle$ \\
            \multicolumn{1}{c}{$\outcomeClass{P}$} \\
            \multicolumn{1}{c}{\scalebox{.7}{(see \cref{fig:nim22})}}
        \end{tabular}}};

    \node (secondB) at (-3.5, -4) {\scalebox{\scaleAmt}{
        \begin{tabular}{l}
            $ (0, 0)$ \\
            \multicolumn{1}{c}{$\outcomeClass{P}$}
        \end{tabular}}};

    \node (secondC) at (-1, -4) {\scalebox{\scaleAmt}{
        \begin{tabular}{l}
            $ (1, 1)$ \\
            \multicolumn{1}{c}{$\outcomeClass{P}$}
        \end{tabular}}};

    \node (secondD) at (1, -4) {\scalebox{\scaleAmt}{
        \begin{tabular}{l}
            $ (0, 0)$ \\
            \multicolumn{1}{c}{$\outcomeClass{P}$}
        \end{tabular}}};

    \node (secondE) at (3.5, -4) {\scalebox{\scaleAmt}{
        \begin{tabular}{l}
            $ (0, 0)$ \\
            \multicolumn{1}{c}{$\outcomeClass{P}$}
        \end{tabular}}};

    \node (secondF) at (6.5, -4) {\scalebox{\scaleAmt}{
        \begin{tabular}{l}
            $ (0, 0)$ \\
            \multicolumn{1}{c}{$\outcomeClass{P}$}
        \end{tabular}}};

    \path[->]
        (preStart) edge [] node [sloped, text width = 1.5cm, yshift = 7] {$(-1, 0)$} (start)

        (start) edge [bend right] node [sloped, text width=1.5cm, yshift = 12] {$\langle\ (-1, 0)$\\ $|\ (0, -1)\ \rangle$} (firstA)
        (firstA) edge [] node [sloped, yshift = 7, text width = 1.5cm] {$(-1, 0)$} (secondA)

        (start) edge [bend right] node [sloped, text width=1.5cm, yshift = 12] {$\langle\ (-1, 0)$\\ $|\ (0, -2)\ \rangle$} (firstB)
        (firstB) edge [] node [sloped, yshift = 7, text width = 1.5cm] {$(-3, 0)$} (secondB)

        (start) edge [bend right=15] node [sloped, text width=1.5cm, yshift = 12] {$\langle\ (-2, 0)$\\ $|\ (0, -1)\ \rangle$} (firstC)
        (firstC) edge [] node [sloped, yshift = 7, text width = 1.5cm] {$(-2, 0)$} (secondC)

        (start) edge [bend left=15] node [sloped, text width = 1.5cm, yshift = 12] {$\langle\ (-2, 0)$\\ $|\ (0, -2)\ \rangle$} (firstD)
        (firstD) edge [] node [sloped, yshift = 7, text width = 1.5cm] {$(-3, 0)$} (secondD)

        (start) edge [bend left] node [sloped, text width = 1.5cm, yshift = 12] {$\langle\ (-3, 0)$\\ $|\ (0, -1)\ \rangle$} (firstE)
        (firstE) edge [] node [sloped, yshift = 7, text width = 1.5cm] {$(0, -2)$} (secondE)

        (start) edge [bend left] node [sloped, text width = 1.5cm, yshift = 12] {$\langle\ (-3, 0)$\\ $|\ (0, -2)\ \rangle$} (firstF)
        (firstF) edge [] node [sloped, yshift = 7, text width = 1.5cm] {$(0, -2)$} (secondF)
        ;

  \end{tikzpicture}}
\end{center}
  \caption{Game tree showing the quantum moves from \ruleset{Nim} position $(3, 2)$.}
  \label{fig:nim42}
\end{figure}

If neither quantumness strongly matters nor quantumness weakly matters, we say \emph{quantumness doesn't matter}.

\begin{definition}{(Quantumness Doesn't Matter)}
We say that for a classical position,
$b$ \emph{quantumness doesn't matter} if the outcome class for that position, 
is the same when played both classically and with quantum moves available, and the classical winning strategy still works for the winning player when responding to classical moves by the losing player.

For any ruleset, ${\cal R}$, we say that \emph{quantumness doesn't matter} if, for all positions, $b$, of ${\cal R}$, quantumness doesn't matter for $b$.
\end{definition}

The \ruleset{Nim} position $(1,2)$ is an example where quantumness doesn't matter.  It is an \outcomeClass{N}-position both classically (shown in \cref{fig:nim12}) and with quantum moves allowed (shown in \cref{fig:qnim12}).

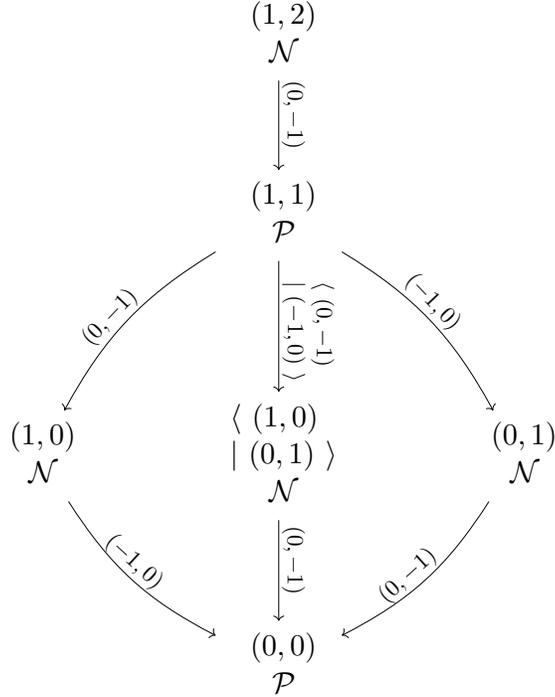
\begin{figure}[h]
  \begin{center}
  \scalebox{.8}{
  \begin{tikzpicture} [node distance = 1cm]
    \node (start) at (0, 7)  {\scalebox{\scaleAmt}{
        \begin{tabular}{c}
            $(1, 2)$ \\
            $\outcomeClass{N}$
        \end{tabular}}};

    \node (firstB) at (0, 4)  {\scalebox{\scaleAmt}{
        \begin{tabular}{c}
            $(1, 1)$ \\
            $\outcomeClass{P}$
        \end{tabular}}};

    \node (secondA) at (-4, 0) {\scalebox{\scaleAmt}{
        \begin{tabular}{c}
            $(1, 0)$ \\
            $\outcomeClass{N}$
        \end{tabular}}};

    \node (secondB) at (0, 0) {\scalebox{\scaleAmt}{
        \begin{tabular}{l}
            $\langle\ (1, 0)$ \\
            $|\ (0, 1)\ \rangle$ \\
            \multicolumn{1}{c}{$\outcomeClass{N}$}
        \end{tabular}}};

    \node (secondC) at (4, 0) {\scalebox{\scaleAmt}{
        \begin{tabular}{c}
            $(0, 1)$ \\
            $\outcomeClass{N}$
        \end{tabular}}};

    \node (end) at (0, -3.5) {\scalebox{\scaleAmt}{
        \begin{tabular}{c}
            $(0, 0)$ \\
            $\outcomeClass{P}$
        \end{tabular}}};

    \path[->]
        (start) edge [] node [sloped, text width = 1.5cm, yshift = 7] {$(0, -1)$} (firstB)

        (firstB) edge [bend right=15] node [sloped, text width = 1.5cm, yshift = 7] {$(0, -1)$} (secondA)

        (firstB) edge [] node [sloped, text width = 1.5cm, yshift = 14] {$\langle\ (0, -1)$\\ $|\ (-1, 0)\ \rangle$} (secondB)

        (firstB) edge [bend left=15] node [sloped, text width = 1.5cm, yshift = 7] {$(-1, 0)$} (secondC)

        (secondA) edge [bend right = 15] node [sloped, text width = 1.5cm, yshift = 7] {$(-1, 0)$} (end)

        (secondB) edge [] node [sloped, text width = 1.5cm, yshift = 7] {$(0, -1)$} (end)

        (secondC) edge [bend left=15] node [sloped, text width = 1.5cm, yshift = 7] {$(0, -1)$} (end)
        ;

  \end{tikzpicture}}
  \end{center}

  \caption{Partial game tree showing that \ruleset{Nim} position $(1,2)$ is in $\outcomeClass{N}$ when played with quantum moves.}
  \label{fig:qnim12}
\end{figure}

In Section \ref{Sec:PClassicalStartQUG}, we proved that \ruleset{Undirected Geography} is an example of an entire ruleset where quantumness doesn't matter with flavor $D$ and superposition width $2$.

\begin{theorem}
In the quantum setting with flavor $D$ and superposition width $2$,
    Quantumness doesn't matter for \ruleset{Undirected Geography}.
\end{theorem}

The next lemma shows while
quantumness may weakly matters but not strongly matters
for a position, as whole,
these two notions are equivalent for any ruleset.
\begin{lemma}
    For any ruleset, ${\cal R}$, if quantumness weakly matters, then quantumness strongly matters as well.
\end{lemma}

\begin{proof}
    Consider a ruleset, ${\cal R}$, where quantumness weakly matters.  Then, for some position, $b$, the classical strategy for the winning player won't suffice.  That means that either from $b$ or for some follower, $H$ of $b$, the classical winning options for the winning player is not a winning move in the quantum version of ${\cal R}$.  That means that, for any one of those options, $X$,
    the outcome classes differs in the two settings.
   Thus, quantumness strongly matters for ${\cal R}$.
\end{proof}

The inverse is not necessarily true, at least for some very restrictive cases.  If we consider \ruleset{Nim-2-2} as a ruleset with only two piles, each with zero, one, or two sticks, then quantum strongly matters for the position $(2,2)$, and it doesn't matter for the other positions $(1, 2)$, $(2, 1)$,  $(1, 1)$, $(0, 1)$, $(1, 0)$, and $(0, 0)$.  We leave this as an Open Problem to determine whether this is true for any ruleset not restricted by a fabricated ceiling like this.

Now, we will attempt to characterize quantumness doesn't matter in the various flavors.


\begin{lemma}
\label{mainDoesn'tMatter}
If all realizations of a game classically have only one possible winner, no matter what sequence of moves either player makes, then quantumness doesn't matter for any quantum flavor other than A.
\end{lemma}
\begin{proof}
In order for a game to classically have only a single possible winner, all paths to a leaf node in a game tree must be of the same parity. And under quantum flavors B, C', and D, for a game to end, all realizations must have no moves remaining. Otherwise, the rule sets would allow for some kind of move to be made. As such, all games end with all remaining realizations at leaf nodes. Since the number of moves to get here is the same in all realization, all paths to the game's end must have the same parity as the classical game's terminal nodes. Then, that same player wins no matter what moves either player makes. Therefore, quantumness doesn't matter.

Under ruleset C, at least one realization must be at a leaf node in order for a game to end. This is because if there are at least 2 possible moves remaining, they can take a quantum move, and if there is only 1 move remaining and there is no realization that is a leaf node, that move must be valid in all realizations. And as with the above, if the number of moves is equal to the parity of a leaf node, then it must be the same number of moves as a classical leaf node, and thus the parity argument still holds true.
\end{proof}

Its worth noting that this explicitly does not hold for ruleset A, since one can easily construct a game where there is only one move at some moment, ending the game instantly.

In rulesets C, C', and D we can get a more general characterization:

\begin{lemma}
If the winning player's strategy in the classic game involved making the same set of moves regardless of the set of moves the opponent has made in the game, then quantumness doesn't matter
\end{lemma}

\begin{proof}
This is easy to see with trivial induction. We assume that all realizations are winning realizations if both players played a classical game on them. Then, after the classical loser makes their move, the winner can make their classical next move, which must be a winning one in all realizations. Since we maintain this invariant that we are winning in all realizations, and the game can never terminate unless a player is losing classically in a realization, this classical strategy must win.
\end{proof}

Finally, in D, we can get a slightly more general characterization:

\begin{lemma}
If the winning player's strategy in the classic game involved making the same set of moves regardless of the \textit{last} move the opponent has made in the game, then quantumness doesn't matter in ruleset D.
\end{lemma}

\begin{proof}
The proof is the same as the previous proof, except we invoke making the response to the previous move instead.
\end{proof}

From this lemma, we have the corollary:

\begin{corollary}
For any game that can be represented as a polynomial sized DAG, with move labels being unique identifiers to the nodes they move to, quantumness doesn't matter in flavor D.
\end{corollary}

This puts \ruleset{Quantum DAG Geography} in
\cclass{P} for flavor $D$.

\begin{corollary}
In the quantum setting with flavor $D$ and any superposition width $w$,
quantumness doesn't matter to \ruleset{Quantum DAG Geography}, and thus, remains polynomial-time solvable.
\end{corollary}

So far, we haven't been able to extend   polynomial-time solvability of \ruleset{Quantum DAG Geography} to other flavors.

Another question that one may consider is when the quantum flavors do or do not matter. Then, we extend hardness results from D to any other flavor by using the following properties:

\begin{lemma}\label{WeakFlavor}
For a set of entangled realizations $G$, if for every position $p \in G$ and every classical move $x$, there exists another classical move $y$ such that the game tree is perfectly symmetrical (which is to say, the game trees save the same nodes, withe the same children, and the same labels), we can reduce from D to any other flavor.
\end{lemma}
\begin{proof}
This holds trivially, as if D chooses a quantum move at any given point, then it is fine, and if $D$ chooses classical move $x$, then it can also choose a quantum move $[p \mid y]$ instead, and those realizations are exact copies, so we may continue play under just one, and as this property is inductive, we are done.
\end{proof}

While this lemma may appear fairly weak, it is still usable for a future hardness result \cref{NodeKaylesAllFlavors}. However, we would like a stronger general statement. Thus, we improve this.

We define a \textit{reduced game tree} as a game tree where we remove all positions where arbitrary play can is sufficient to win the game. If it ends on the "wrong" winner, then we add a pair of leaves below it. This doesn't change the outcome class.

\begin{lemma}\label{StrongFlavor}
For a set of entangled realizations $G$, if the following property holds recursively in D, following at a subset of winning strategies $W$ of sets of winning strategies $X$ (for the reduced game tree) that are sufficient to cover any winning strategy, then we can reduce to flavors B, C, and C' from D if the following property holds:

For any winning move $x$ that, after played, leaves a set of winning strategies $W'$ that coves each possible winning strategy, there is another classical move $y$ such that each strategy in $W'$ is still a winning strategy
\end{lemma}
\begin{proof}
Just as before, this is mostly trivial. We will construct a winning strategy for usingly only quantum moves. Whenever D selects a quantum move, we don't change anything. At any junction where flavor D chooses move $x$ classically, we can instead choose $[x, \mid y]$ instead, where $y$ is the paired winning move. Then, by definition, both game trees have $W'$ has sets of possible winning strategies, which then inductively becomes the set $W$ for the next move.

Since this strategy only uses quantum moves, the other variants can all use this, regardless of tree structure.

Once we hit the end of the set of moves, then quantumness doesn't matter, so by \cref{mainDoesn'tMatter}, it is a win in in all realizations other than A.
\end{proof}

Note that we only can't get A in general. If one can show that the quantumness doesn't matter portion of the game can be won by A, like we did for Avoid True, then it can be extended as well.


\section{Quantum Transformation of \cclass{PSPACE}-Complete Games}\label{sec:Reductions}

So far, we've seen that when we introduce quantumness to classically tractable combinatorial games, sometimes they remain tractable (e.g., \ruleset{Undirected Geography} with classical starts) and sometimes they become intractable (e.g., \ruleset{Quantum Nim} and \ruleset{Quantum Undirected Geography} with poly-wide quantum starts).
The next logical direction to explore then is the complexity landscapes of quantum
transformation of
 classically intractable games.
Is their intractability 
always preserved?
 Is their complexity ever lost?
 In this section, we hope to
 shed some light on
 these questions with various concrete games.


\subsection{Quantum Preservation and Quantum Collapse}

\label{Sec:QuantumCollapses}

It is well known that the Quantified Boolean formula problem (QBF)---of determining whether a quantified Boolean formulae is true or false---can be viewed as a classical combinatorial game.
Technically, combinatorial games must also fulfill the ``normal play'' requirement,  meaning that a player must lose {\em if and only if} when they are unable to make a move.
There are several equivalent ways to do this (to transform logic QBF into combinatorial games), and are all classically polynomial-time reducible to each other.
Complexity-theoretically, QBF is the canonical complete problem for \cclass{PSPACE}, and thus all these QBF variants are \cclass{PSPACE} combinatorial games;
they form the bedrock of \cclass{PSPACE} reductions.

However, because quantum games
  are more subtly dependent
  on the explicit move definitions
  in the ruleset,
  these known reductions from
   the ``classical world''
  don't necessarily continue to hold.
In this subsection, we discuss
  various ``end-of-QBF'' game variants.
We will give intractability results
  for all ``natural'' \ruleset {Quantum QBF} variants,
 some of which is critical to our intractability analysis of \ruleset{Quantum Nim} in the previous section.
We will also highlight the subtle differences among these \ruleset {Quantum QBF} variants.

In the literature, the Quantified Boolean formula problem is also known as the Quantified Satisfiability Problem (QSAT).
So, QBF and QSAT are used interchangeably.
However, in this paper,  we will---for clarity of presentation---use the following ruleset naming convention for \ruleset{QBF} and \ruleset{QSET},
in order to denote two different ruleset families of combinatorial games rooted in the Quantified Boolean formula problem.

\begin{itemize}
\item {\bf The \ruleset{QBF}  Family}: We will use \ruleset{QBF} to denote the family of
the combinatorial games that {\em textually} implements the Quantified Boolean formula problem: An instance of \ruleset{QBF} is given by a CNF $f$, whose clauses may contain both positive and negative literals, over two ordered lists of Boolean variables.
We have two cases depending on whether or not one of the list has one more variables than the other one:
\begin{itemize}
\item {\bf Case Even}:
$(T_1,...,T_n)$ and $(F_1,...,F_n)$
\item {\bf Case Odd}
$(T_1,...,T_n, T_{n+1})$ and $(F_1,...,F_n)$
\end{itemize}
So, the quantified boolean formulae is, respectively:
\begin{eqnarray*}
\mbox{\bf Case Even:} & \exists T_1\forall F_1\cdots \exists T_n\forall F_n & f(T_1,...,T_n,F_1,...,F_n) \\
\mbox{\bf Case Odd:} &\exists T_1\forall F_1\cdots \exists T_n\forall F_n\exists T_{n-1} & f(T_1,...,T_n, T_{n+1}, F_1,...,F_n)
\end{eqnarray*}
In this game, one player---{\em Player True}---aims
  to satisfy the CNF formulae,
  while the other player---{\em Player False}---wants the formulae to be unsatisfied.
Starting with True, the two players alternate turns setting their next variables
to \texttt{True} or \texttt{False}.
In the other words, in their respective $i^{th}$ turn, Player True sets variable $T_i$ and Player False sets variable $F_i$.
The variants in this family, as we shall define later, differ in the details of the termination condition in transforming logic QBF into a combinatorial game.

\item {\bf The \ruleset{QSAT} Family}: We will use \ruleset{QSET} to denote the family of
{\em partition free assignment} \ruleset{QBF} games, which
relaxes on the variable that they can set at each turn.
Like in \ruleset{QBF}, in the game instance of \ruleset{QSET} defined by a CNF formulae $f$,
two players---Player True and Player False---take turns to set their own variables;
Player True can only set variables in $\{\cdots T_i\cdots\}$ and Player False can only set variables in $\{\cdots F_i \cdots \}$.
Again, we have two cases---Case Even and Case Odd---depending on whether Player True has one more variables than Player False.
For the game, Player True aims to satisfy the CNF formulae,  while Player False wants the formulae to be unsatisfied.
Unlike in \ruleset{QBF}, however, the players can ``freely" set any of their unassigned variables to
\texttt{True} or \texttt{False}, as opposed to setting variables according to the prescribed order.
The variants in this family also differ in the
details of termination condition in transforming logic QBF into combinatorial games.
\end{itemize}



We now summarize the variants details the termination of QBF/QSAT games.

\begin{itemize}
\item {\bf Termination Rules}:
In the first three variants, players continues until to play the game until all variables are assigned.
Then the following rules governs how games terminate.

\begin{itemize}
\item In \ruleset{Phantom-Move QBF}/\ruleset{Phantom-Move QSAT}, after
all variables have been assigned,
the final player have a feasible move---called the phantom move---to make
if and only if they have a winning assignment.
\end{itemize}
This variant of the game is implicitly invoked in several classic reductions,
  including Schaefer’s Avoid True reduction \cite{DBLP:journals/jcss/Schaefer78} as seen in the previous section.
In the next two variants,
after setting all of the variables,
False selects the clause they believe True failed to satisfy.
\begin{itemize}
\item In \ruleset{ClauseSelector QBF}/\ruleset{ClauseSelector QSAT},
this move is only available if one exists, and if False plays here they win.
\item In \ruleset{LiteralSelector QBF}/\ruleset{LiteralSelector QSAT},
False can always choose any clause, then True attempts select the variable they did satisfy.
If they can, False has no moves so True wins.
If they can't, then True has no moves and loses.
\end{itemize}
Both are common interpretation used in reductions.
For example, \ruleset{ClauseSelector QBF} is used in the original \ruleset{Node-Kayles} reduction \cite{DBLP:journals/jcss/Schaefer78}, and \ruleset{LiteralSelector QBF} is used in the \ruleset{Geography} reductions \cite{LichtensteinSipser:1980}.

\item {\bf The Mercy Rules}: 

The last three variants define when the game terminates early after a player has achieved a victory condition (and to stop the losing player to continue playing in a completely hopeless position).
The victory condition for Player True is having a true literal in every clause, and for Player False is having a clause with all false literals in it.
The only difference between the following three is how the mechanic for ending the game early works.
\begin{itemize}
\item 
In \ruleset{TKO QBF}/\ruleset{TKO QBF},
the game ends automatically as soon as a victory condition is reached.
\item
In \ruleset{Haymaker QBF}/\ruleset{Haymaker QSAT},
a player may opt to, instead of a normal variable selection, make a 
``game ending blow"which is only available after the player's win condition has been met.
After the move is made, the opponent can no longer move.

\item In \ruleset{KO QBF}/\ruleset{KO QSAT},
while a player makes a move, if they have a achieved a winning state, they can also choose to end the game,
forcing the other player to be unable to make a move.
\end{itemize}
\end{itemize}

For complexity-theoretical analysis,
 we first focus on the more restrictive
\ruleset{QBF} family, in which players must set their variable according to the prescribed order.

\begin{theorem}[Quantum Collapses of \ruleset{QBF}: Classical Starts]
\label{thm:PhantomMove}
For quantum flavors $C'$ and $D$, and superposition width $2$,
determining the outcome class of \ruleset{Quantum Phantom-Move QBF} with a classical start is $\Sigma_2$-complete in Case Even, and  is $\Pi_2$-complete in Case Odd.
\end{theorem}

\begin{proof}
Notice that the player who makes the final move (the "phantom move") will always have a collapse so that the variables are assigned in a way that they will win (if it is possible).
And there is no point at which quantum moves can collapse before that point.
As such, the optimal move for this collapsing player will be to make exclusively quantum moves, and the optimal move for other player will be to never make a quantum move.

So, one player makes a variable assignment of half of the variables, and then only wins if every other assignment of the other variables are winning positions for them. Otherwise, the phantom move player wins.
If the phantom move player is False, then this is exactly a $\Sigma _2$ SAT problem.
Otherwise, it is exactly a $\Pi _2$ SAT problem.

For formal completeness, we note the trivial two way reduction.
In Case Even,
we can create a $\Sigma_2$ SAT problem, giving True's variables as the "exists" variables, and False's variables as the "for all" variables.
In Case Odd,
we can use a $\Pi_2$ SAT solver by similarly giving False's variables as the "for all" variables, and True's variables as the "exists" variables.
As previously described, both of these will output which player wins correctly.

For the other direction, we can do the same with both, only in the other direction. In other words, make all "exist" variables into True's variables and all "for all" variables into False's variables. Then, we have a corresponding \ruleset{Quantum Phantom Move QBF} instance.
\end{proof}

Note that in flavor $A$, \ruleset{Quantum Phantom-Move QBF} is trivially in linear time, since no player can ever take the phantom move (because it is a classical move and classical moves are not allowed in flavor $A$), so the game is just based on parity.
Under flavor $B$, it is pretty obviously NP-complete for the true player going last (Case ODD) and co-NP-complete if the false player goes last (Case Even).
Since all players always must take the only quantum move available to them, and then whoever collapses will win so long as there is a single assignment which evaluates to true or false (depending on who makes that move).
This is just a SAT problem when True goes last, and otherwise is \cclass{co-NP} complete \ruleset{Tautology} when False goes last.

We can also extend the result to the quantumized complexity of \ruleset{Quantum Phantom-Move QBF} with poly-wide quantum starts.

\begin{corollary}[Quantum Collapses of \ruleset{QBF}: Poly-Wide Quantum Starts] \label{cor:Poly-WidePhantomMove}
For quantum flavors $C'$ and $D$, and superposition width $2$,
determining the out come class of \ruleset{Quantum Phantom-Move QBF} with a poly-wide quantum start is $\sum _2$-complete in Case Even , and is $\Pi _2$-complete in Case Odd.
\end{corollary}
\begin{proof}
Given a poly-wide quantum position $\mathbb{B}$ in which the $i^{th}$ realization is defined by CNF $f_i$.
The players make decision as if they play for \ruleset{Quantum Phantom-Move QBF} with the classical start for logic function $\bigvee f_i$.
\end{proof}

\begin{theorem}
For quantum flavors $D$ and $C'$, superposition width $2$,
determining the outcome class for
\ruleset{Quantum Literal Selector QBF} with a classical start
is \cclass{co-NP} complete.
    \label{thm:selector}
\end{theorem}

\begin{proof}
Since Player True wins so long as the variable in the selected clause has a quantum variable assignment, and can't be collapsed before then, they will always make a quantum move, and Player False will never make a quantum variable assignment.
So we can, WLOG, remove every clause that has a True's variable in it, which will always be satisfied.
Now, False  wins if and only if there exists an assignment to make a clause unsatisfied, and thus make the SAT evaluate to false. This is is exactly \textsc{Tautology}, the canonical co-NP-complete problem. Reducing between the two is trivial, since it involves doing nothing one way, and creating dummy clauses and variables for Player True when reducing from \textsc{Tautology} to \ruleset{Quantum Literal Selector QBF}.
\end{proof}

\begin{theorem}
\label{thm:clauseselector}
For quantum flavors $D$, $C'$, and $C$, and superposition width $2$,
determining the outcome class for
    \ruleset{Quantum Clause Selector QBF} with a classical start is \cclass{NP} complete.
\end{theorem}

\begin{proof}
This closely follows the proof of \ruleset{Quantum Literal Selector QBF}.
We assume WLOG that that no clauses have a false variable appear both negated and unnegated (since those cases are trivially a win for True). Note that false just needs a single realization to have a clause with all false literals. So they will always make a quantum move, since any classical play they could make could also be reached by playing quantum and choosing the same clause in the end. And since any clause with quantum literals can assumed to be false, we can WLOG remove all of false's variables, and True can win iff there exists a satisfying assignment for the rest of the variables. This is exactly a CNF problem.
\end{proof}

For rulesets $A$ and $B$, both selector variants are trivially in linear time and constant time, respectively, since both players are forced to make quantum moves until the end, upon which any two clauses that isn't a tautology works for \ruleset{Quantum Clause Selector QBF}, and for \ruleset{Quantum Literal Selector QBF}, any two literals will suffice, so it is always a True win.

\begin{theorem}
\label{thm:kyles}
For flavor $D$ and superposition width $2$,
   \ruleset{Quamtum TKO QBF} with a classical start is a \cclass{PSPACE}-complete game.
\end{theorem}

\begin{proof}
By our complexity characterization of Section
\ref{Sec:PSpaceUpperBound},
\ruleset{Quantum TKO QBF} is a \cclass{PSPACE}-solvable game.

We will show hardness of the quantum game by reducing from its classical TKO variant.
Consider any given \ruleset{TKO QBF} instance $Z'$, in which for clarity of presentation we assume that True's variables are labeled $T_1, T_9, T_{17}, \dots, T_{8n + 1}$ with False's variables labeled $F_4, F_{12}, F_{20}, \dots, F_{8n - 4}$, where the order of play is:
$$T_1, F_4, T_9, F_{12}, \dots, F_{8n - 4}, T_{8n + 1}.$$

Then, we reduce $Z'$ to another  \ruleset{TKO QBF} instance $Z$, and analyze the complexity for determining the outcome class for its quantum lift $Z^{\QuantumLift(D,2)}$.
In $Z$, we will create variables such that the order of play is:
$$T_1, F_2, T_3, F_4, \dots, F_{8n}, T_{8n + 1}, F_{8n + 2}.$$
For each classic False's variable $F_i$, we create clauses
$\dots \wedge (F_i \vee T_{i+1}) \wedge (F_{i+2} \vee \neg F_{i + 2}) \wedge \dots$.  For each classic true variable $T_i$, we create clauses
$\dots \wedge (T_i \vee \neg F_{i+1} \vee T_{i+2}) \wedge (\neg T_i \vee F_{i+1} \neg \vee T_{i + 2}) \wedge (F_{i + 1} \vee T_{i + 4} \vee F_{8n + 2}) \wedge (\neg F_{i+1} \vee \neg T_{i + 4} \vee F_{8n + 2}) \wedge (F_{i + 3} \vee \neg F_{i + 3}) \wedge \dots$.

First, note that if Player False makes a quantum move on any of the classic False's variables $F_i$, then Player True  can set $T_{i + 1}$ to false to collapse False's move to true. As such, if the QBF evaluates classically to true, then False can't make a quantum move, and thus True wins.

Now we will show that if Player False is winning the classic QBF, then False will win this quantum QBF. Notice that if True makes any classical move on classical variable $T_i$, they will not be hurt by the additional clauses, as no matter what false does, they can satisfy those clauses by an appropriate response.

If Player True makes a quantum move on any of the classic true variables $T_i$,
then Player
False has a winning strategy, by setting $F_{i + 1}$ to false

If $T_{i+2}$ is set to quantum, then $F_{i+1}$ can never collapse until the end, even if $T_i$ collapses. As such, no matter what $T_{i+4}$ is set to, there exists a possibility that no all clauses are satisfied, and by setting $F_{n+1}$ to false, all the possibilities where Player True had satisfied all clauses collapse, and thus one of the extra clauses must be all false, giving Player False the win.

If $T_{i+2}$ is set to true, then there are 3 possibilities. The first $T_i$ never collapses, which follows the same outcome as above, giving Player False the win.
The second is that it collapses to \texttt{false}. In this case, $F_{i+1}$ never collapses, and thus the old argument still holds. The last is that it collapses to \texttt{true}. This causes $F_{i+1}$ to collapse to true. If $T_{i + 4}$ was set to \texttt{true} as well, then Player False will win as there is one clause with only falses and a future False's variable in it. If it was set to quantum, then the old argument applies where eventually Player False will win when setting $F_{8n+2}$. If it was set to \texttt{false}, then Player False  can just play as if it was set to \texttt{true}, since if it collapses to \texttt{false}, Player False wins.

If $T_{i+2}$ is set to \texttt{false}, then the same general arguments hold. If $T_i$ never collapses, then we follow the old argument to win. If it collapses to \texttt{true}, then $F_{i+1}$ never collapses, and the old argument holds. If it collapses to \texttt{false}, then $F_{i+1}$ collapses to \texttt{false}. If $T_{i+4}$ is set to \texttt{false}, then there is an unsatisfiable clause for Player True, 
causing the player to lose. If it was set to quantum, than the old argument still holds. If it was set to \texttt{false}, then Player False  can just play as if $T_i$ was set to \texttt{false}, since if it collapses to \texttt{true}, Player False wins.

The reduction is polynomial space, since we are just quadrupling the number of variables and adding a constant number of clauses for each variable.
\end{proof}

Notice that with quantum flavor $A$ and $B$, the winner is determined by parity if a player doesn't lose in all realizations. The False Player always has a move if and only if \textsc{Tautology} problem on the CNF evaluates to false, and the True Player always has a move if and only if the CNF evaluates to true. Thus, if the CNF ends with a True Player assignment, the game is \cclass{NP}-complete, and if the CNF ends with a False Player assignment, it is \cclass{Co-NP}-complete.

\begin{theorem}
With quantum flavor $D$ and superposition width $2$, \ruleset{Quantum KO QBF}
is a \cclass{PSPACE}-complete game.
    \label{thm:brb}
\end{theorem}

\begin{proof}
Again, by our complexity characterization of Section
\ref{Sec:PSpaceUpperBound},
\ruleset{Quantum KO QBF} is a \cclass{PSPACE}-solvable game.

For hardness of the quantum game, we will reduce from the classical
\ruleset{KO QBF}.
Consider any given \ruleset{KO QBF} instance $Z'$.
For clarity of presentation we assume that True's variables are labeled  $T_1, T_9, T_{17}, \dots, T_{8n - 7}$ and False's variables are labeled $F_5, F_{13}, F_{21}, \dots, F_{8n - 3}$, where the order of play on the variables is:
$$T_1, F_5, T_9, F_{13}, \dots, T_{8n - 7}, F_{8n - 3}.$$

Then, we reduce $Z'$ to another  \ruleset{KO QBF} instance $Z$, and analyze the complexity for determining the outcome class for its quantum lift $Z^{\QuantumLift(D,2)}$.
In $Z$, we will create variables such that the order of play is:
$$T_1, F_2, T_3, F_4, \dots, F_{8n - 1}, T_{8n}, F_{8n+1}.$$
First, in all of the existing clauses, we add  an additional $\vee T_{8n} \vee F_{8n + 1}$. We then create three  new ``gadget'' clauses for each True's variable $T_i$:
$$\dots \wedge (T_i \vee T_{i+2}) \wedge (\neg T_i \vee \neg T_{i+2}) \wedge (F_{i+1} \vee \neg F_{i+1}) \wedge \dots$$
Then, for each False's variable $F_i$, we add the following three ``gadget'' clauses:
$$\dots \wedge (F_i \vee F_{i+2} \vee \neg T_{8n}) \wedge (\neg F_i \vee \neg F_{i+2} \vee \neg T_{8n}) \wedge (T_{i + 1} \vee \neg T_{i+1}) \wedge \dots.$$
The gadget clauses are designed to satisfy that the logic value of the formulae for $Z$ preserves
the logic value of the formulae for $Z'$.

First, suppose
$Z'$ has value \texttt{false}, so $Z$ also has value \texttt{false}.
Note that if Payer True makes a quantum move for any of the classic variables $T_i$, then no matter what they assign $T_{i+2}$ to, there exists a realization where either $(T_i \vee T_{i+2})$ or $(\neg T_i \vee \neg T_{i+2}$ evaluates to \texttt{false}, so Player False can
win with a KO move.
Player True will not be able to deliver
a KO move to win 
because there will exist clauses that have not yet been satisfied.

False will simply play their classic strategy for classic variable $F_i$ and echo that assignment for $F_{i+2}$.
If True plays purely classically until $T_{8n}$, then the classic variables of one of the classic clauses will evaluate to \texttt{false} when they select $T_{8n}$.
So, they will need to set it to \texttt{true} to avoid having that clause evaluate to \texttt{false} the next turn.
However, since Player False is echoing their choices into $F_{i+1}$, at least one clause of the form $(F_i \vee F_{i+2} \vee \neg T_{8n})$ currently has $F_i \vee F_{i+2} = 0$, and $T_{8n}$ variable will be set to 1, which will allow Player False to make the KO move on their next turn.
If Player True makes a quantum move for $T_{8n}$, then False can still deliver the KO move
because there exists a realization with all \texttt{false}.

Now suppose that $Z'$ has value \texttt{true},
so $Z$ also has value \texttt{true}.
True's strategy is to play their classical strategy on the old variables, and satisfy the new clauses with only True's variables in them by choosing appropriately.
If there is any point where Player False  makes a quantum move, True plays as if False selected \texttt{true}
for the variable.
Notice that Player False is unable to KO
until Player True sets $T_{8n}$, since that is in every clause that could possibly evaluate to \texttt{false}.
And then when Player True sets $T_{8n}$, they can set it to \texttt{false} and press the button, satisfying all of the new clauses, and there exists a realization where all of the classical clauses are \texttt{true} (when Player False chose \texttt{true}
for the quantum move).

The reduction is in polynomial space, as all we do is create three more variables for each existing variable, and three new clauses for each existing variable, so the reduction is linear in the QBF input size.
\end{proof}

\begin{theorem}
For quantum flavor $D$ and superposition width $2$, \ruleset{Quantum Haymaker QBF}
is \cclass{PSPACE}-complete game.
    \label{thm:mxc}
\end{theorem}

\begin{proof}
Like all of other variants of \ruleset{Quantum QBF}, \ruleset{Quantum Haymaker QBF} is still a polynomial-space solvable game.

For hardness analysis, the reduction is exactly the same, and the proof follows the same guidelines as for \ruleset{Quantum KO QBF}. Notice that if Player True makes a quantum move on $T_i$, no matter what they select for $T_{i+2}$,
Player False can 
make the game ending blow on their next turn.
If False makes a quantum move, False can't
make the game stopping move
until after $F_{8n + 1}$ has been set.
But if Player True player was the classical winner, setting $T_{8n}$ to \texttt{false}, and was playing optimally to one of False's moves, True can make the game-ending blow after that to win the quantum game.
\end{proof}

We now turn our attention to
\ruleset{Partition-Free QBF} and their quantumized complexity.
In contrast to standard \ruleset{QBF} games,
 the partition-free version relaxes the order requirement on setting the variables.
As we discussed before, we will refer to this family of QBF-based games as the \ruleset{QSAT}
family.
Classically, all variants of \ruleset{QSAT} are \cclass{PSPACE}-complete games.
Below, we consider the complexity
\ruleset{QSAT} games in the quantum setting.
Most importantly, for our proof of quantum leap in ruleset{Nim}'s complexity , we prove the following theorem.

\begin{theorem}[Quantum Collapses of Partition-Free QBF]
For quantum flavors $D$ and $C'$ and superposition width $2$,
\ruleset{Phantom Move QSAT} with a classical start
is a $\Sigma _2$-complete game in Case Even
and is a $\Pi _2$-complete game in Case Odd.
    \label{thm:partitionFreePhantom}
\end{theorem}

\begin{proof}
We will call the player that makes the phantom move the PM Player and the other player as the NM Player, for short.

Then, the theorem follows from the following observations:

\begin{observation}\label{FalseStrategy}
If the PM Player has a winning strategy, then arbitrarily assigning each variable to $\langle \texttt{true}\ \mid\ \texttt{false}\rangle$ is also a winning strategy.
\end{observation}
\begin{proof}
Suppose for the sake of contradiction that the PM player has a winning strategy $S_1$ but the only quantum strategy $S_2$ isn't a winning strategy.
Then, that means that NM Player has a sequence of moves such that in all realizations,
the PM player hasn't fulfilled their winning condition, as otherwise they could move into the phantom move.
But then, that exact sequence of moves is a winning sequence of moves against $S_1$, and any possible deviations of it, $S_1$ isn't a winning strategy.
\end{proof}

\begin{observation}\label{TrueStrategy}
If the NM payer has a winning strategy, then they have a winning strategy of selecting classic moves independent of what the
PM player does.
\end{observation}
\begin{proof}
If the NM Player has a winning strategy against PM player, then they must also have one against the PM Player strategy of selecting just quantum assignments as mentioned in \cref{FalseStrategy}.
But, regardless of what this strategy is, the fact that it wins means that it must win against all realizations of this strategy, since otherwise the PM Player could collapse to that realization on the phantom move, then win. And clearly the NM Player can do this with only classic moves, as any realization of a quantum strategy must also not have any winning realizations for the NM Player either.
\end{proof}

From here, we can form a trivial reduction to and from $\Sigma_2$-SAT and $\Pi_2$-SAT for instances where False gets the phantom move and True gets the phantom move, respectively.
\end{proof}

The proof can natually be extended to handle any superposition width.
It can be extended to poly-wide quantuam starts.

\begin{corollary}[Quantum Collapses of Partition-Free \ruleset{QBF}: Poly-Wide Quantum Starts] \label{thm:Poly-WidePhantomMove}
For quantum flavors $C'$ and $D$, and superposition width $2$,
determining the out come class of \ruleset{Quantum Phantom-Move QSAT} with a poly-wide quantum start is $\sum _2$-complete in Case Even , and is $\Pi _2$-complete in Case Odd.
\end{corollary}

\begin{proof}
Given a poly-wide quantum position $\mathbb{B}$ in which the $i^{th}$ realization is defined by CNF $f_i$.
The players make decision as if they play for \ruleset{Quantum Phantom-Move QSAT} with the classical start for logic function $\bigvee f_i$.

However, the proof is a little bit more complicated than that of \ruleset{Quantum Phantom-Move QBF} to deal with the case that players can make quantum moves involving more than one variables.
\end{proof}

\begin{theorem}
\ruleset{Quantum Clause Selector QSAT}
and
\ruleset{Quantum Literal Selector QSAT}
are respectively \cclass{NP}-complete and
\cclass{co-NP}-complete, under flavors $D$ and $C'$.
\end{theorem}

\subsection{The Depth of Quantum Collapse}
\label{Sec:Depth}

We now have seen instances of classically \cclass{PSPACE}-complete games becoming complete in the first or second level of the polynomial hierarchy. This leads us to a natural question of whether there exists games that
collapse to any level of the polynomial hierarchy. \footnote{Theorem \ref{Theo:PSPACE-Intermediate} is also motivated by celebrated computational-complexity characterizations of Lander's  \cclass{NP}-intermediate problems \cite{Lander} and Ko's intricate, meticulously separation of levels of
polynomial-time hierarchy \cite{KoPHSeparation}.}
For all levels other than 0, we have an answer:

\begin{theorem}[\cclass{PSPACE}-Intermediate Quantum Games]
\label{Theo:PSPACE-Intermediate}
For any complexity class $\Sigma_k$, $\Pi_k$, or both $\Sigma_k$ and $\Pi_k$, there exists a classically \cclass{PSPACE}-complete game, that in $Z^{Q(D, w)}$, is complete in that class in ruleset D.
\end{theorem}

\begin{proof}
We create games called \ruleset{QuantumLevel($\ell$, $C$)}, where $\ell$ and $C$ are replaced with different values.
The value of $\ell$ is the level of hierarchy that we want the quantum complexity to be, and $C$ is one of $\Sigma$, $\Pi$, and $\Sigma$ and $\Pi$.
Then, this game will have classic complexity \cclass{PSPACE}-complete and quantumized complexity under $Z^{Q(D, w)}$
 of $C_{\ell}$.
We will first be focusing on only the cases where $C$ is $\Sigma$ or $\Pi$.

The game of \ruleset{QuantumLevel($\ell$, $C$)} is played over a CNF with variables labeled $T_1, T_2, T_3, \dots T_n$ and $F_1, F_2, \dots F_n$, where $\floor{\frac{\ell}{2}}$ divides $n$. If $\ell$ is odd, there are an additional $x = \frac{n}{\floor{\frac{\ell}{2}}}$ variables. If $C = \Sigma$ then they are labeled $T_{n+1}, T_{n+2}, \dots T_{n + x}$. If $C = \Pi$, then they are labeled $F_{n+1}, F_{n+2}, \dots F_{n + x}$. Each of the $T$ variables and each of the $F$ variables are further divided into groups of size $x$. We will thus call $T_1, T_2, \dots T_{x}$ as $G_{T1}$, $T_{1 _+ x}, T_{2 + x}, \dots T_{2x}$ as $G_{T2}$, and so on, and the same for $G_{Fi}$.

Classically, the game is played as follows:

\begin{enumerate}
\item Both players alternate turn assigning the variables in $G_{i}$, in order of their index.
\item Upon all variables in the group being assigned, we find the \textit{starting player} of the group. This is True if $C = \Sigma$, and is otherwise False.
\item The starting player echoes the assignment of the first variable in the group that they assigned.
The
other player then echoes it as well. Play like this continues until the same player echoes the last assignment in the group.
\item For the last assignment, the same player echoes it, then the other player echoes it, then the same player echoes it again, after which the other player then begins to echo their variables, and the process repeats for their variables.
\item Repeat step 1, now for the variables in $G_{i + 1}$
\item if there is a group in the end for one player with no variables for the other player, then the player with the variables plays and the other player just echoes their assignment.
\item Once all variables are assigned, the player's who turn it is may move if and only if they have achieved their win condition in all clauses.
\end{enumerate}

It is not difficult to see classically that these games are PSPACE-complete. They are simply normal
\ruleset{QBF} problems with additional moves where players will always be able to move, making no changes at all to the assignment. Thus a trivial two way reduction between the game exists.

For a quantum game, we can reduce to and from the SAT canonical for the desired complexity class. Note for this CNF, that the optimal strategy for the first phase is to just take a quantum move of both variable assignments, as one can assign them later, and there is no benefit for assigning them earlier. Then, if one chooses a quantum move during phase 1 or 2, the opponent gets to make an assignment for you, so it is always in the players favor to make a classical assignment here. But then, when done for all variables, this means that this is exactly a problem on the $\ell$ level of the hierarchy, and thus there is a trivial reduction between the $\Sigma_{\ell}$ and $\Pi_{\ell}$ SAT problems.

Finally, we briefly note that for \ruleset{QuantumLevel($\ell, \Sigma$ and $\Pi$)}, we simply make the ruleset contain all games in \ruleset{QuantumLevel($\ell, \Sigma$)} and \ruleset{QuantumLevel($\ell, \Pi$)}. Since the games in the first $\Sigma_{\ell}$-complete canonically and for the second are $\Pi_{\ell}$-complete canonically, this game is now $\Sigma_{\ell}$ and $\Pi_{\ell}$ complete, completing the proof.
\end{proof}

Note that we can very slightly modify this game to get the completeness results for any flavor. Just simply make any label of True or False have a pair. Whenever one has a truth assignment as an option, they can express this with two different labels. Clearly no hardness results changes in D. By \cref{WeakFlavor}, the hardness results hold in the rest of the variants as well.


\subsection{Natural \cclass{PSPACE}-Complete Games}
\label{Sec:PSPACEComplete}
\ruleset{Quantum Geography}: \ruleset{Geography} is a game played on a di-graph, where players take turns moving a token from one vertex to an adjacent one, and it is illegal to move to a vertex already traversed. Like the rest of the games discussed, because it is normal play, the game ends when a player has no remaining adjacent vertices to traverse to.
Early, we considered \ruleset{Undirected Geography}, a special case of \ruleset{Geography}, in which the graph is undirected.
We have extended its classical polynomial-time solvability to \ruleset{Quantum Undirected Geography} with a classical start, and have
also showed that its quantumized complexity with a poly-wide start is \cclass{PSPACE}-complete.


Classically,
\ruleset{Geography} is a \cclass{PSPACE}-complete
combinatorial game.
We now prove that \ruleset{Quantum Geography}
remains \cclass{PSPACE}-complete.

\begin{theorem}
 For flavor $D$, \ruleset{Quantum Geography} remains \cclass{PSPACE}-complete (independent of superposition width).
    \label{thm:geography}
\end{theorem}
\begin{proof}
We have a very simple reduction from classic 
\ruleset{Geography}.
We replace each edge in the \ruleset{Geography} graph as shown in \cref{fig:geographyproof} with a path through two new vertices.

Now, if a player ever makes a quantum move from a classical move, e.g. from $A$ to the super position of $AB_1$ and $AC_1$, then the opponent can immediately collapse to either $AB_2$ or $AC_2$, effectively choosing which of $B$ and $C$ will be moved to.
Thus, making a quantum move only gives the next player the power to choose your move and will never give a classically-losing player a winning quantum strategy.

In the next section, we will prove a result implying that \ruleset{Quantum Geography} remains \cclass{PSPACE}-solvable.
\end{proof}

\begin{figure}[h]
\begin{center}
\begin{tikzpicture}[node distance = 2cm]
    \tikzstyle{vertex}=[circle,thick,draw = black,minimum size=6mm]
    \node[vertex] (A)  {$A$};
    \node[vertex] at (3, 0) (B) {$B$};
    \node[vertex] at (4, 0) (A2) {$A$};
    \node[vertex] (AB1) [right of = A2] {$AB_1$};
    \node[vertex] (AB2) [right of = AB1] {$AB_2$};
    \node[vertex] (B2) [right of = AB2] {$B$};

    \path[-{Latex[length=3mm]}]
        (A) edge node [below, yshift=-.75cm] {\ruleset{Geography} Edge} (B)
        (A2) edge node {} (AB1)
        (AB1) edge node [below, yshift = -.75cm]{Resulting \ruleset{Quantum Geography} Gadget} (AB2)
        (AB2) edge node {} (B2);
\end{tikzpicture}
\end{center}
\caption{\cclass{PSPACE} Reduction to \ruleset{Quantum Geography} is a simple transformation on the edges.}
\label{fig:geographyproof}
\end{figure}
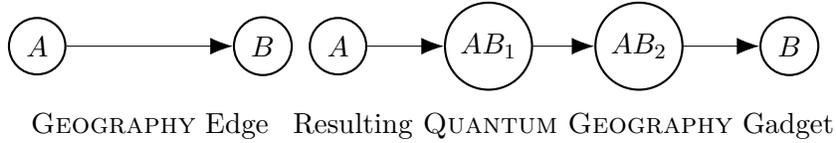

Unfortunately, we can't simply apply one of our lemmas to extend this hardness result to the other flavors. So, instead, we craft a more careful reduction, inspired by this one, that can work with only quantum moves selected.

\begin{lemma}
\label{lem:DAGGepgraphy}
\ruleset{Quantum Directed Geography} is PSPACE-complete in all flavors.
\end{lemma}
\begin{proof}
PSPACE membership is obvious because the game is polynomially short.

For hardness, we reduce from classical \ruleset{Geography}: For each directed edge $(A, B)$, introduce 5 vertices: $AB_{a_1}$, $AB_{a_2}$, $AB_{b_1}$, $AB_{b_2}$, $AB_{c_1}$, and $AB_{c_2}$. Let there be directed edges $(A, AB_{a_1})$, $(A, AB_{a_2})$, $(AB_{a_1}, AB_{a_2})$, $(AB_{b_1}, AB_{b_2})$, $(AB_{a_2}, B)$, $(AB_{a_2}$, $AB_{c_1})$, $(AB_{b_2}$, $AB_{c_1})$ and $(AB_{c_1}, AB_{c_2})$. Notice that now, one may always take these two edges for any directed edge from $A$. Then, the other player must either move classically or quantumly to $AB_{a_2}$ and/or $AB_{b_1}$. If $B$ doesn't exist, the player either loses, or moves to $AB_{c_1}$, after which, the other player would move to $AB_{c_2}$ and win. If $B$ does exist, the player can make a quantum move to $B$ and $AB_{c_2}$, after which the player may collapse with their own quantum move.
\end{proof}

\noindent\ruleset{Quantum Node Kayles}: \ruleset{Node Kayles} is a game where players alternate turns placing tokens on vertices of a given graph.
A player is only able to place tokens on vertex if that vertex does not already contain a token and is not adjacent to any vertex with a token.
As such, a player is unable to move when the tokens form an {\em maximal}
independent set.
This game is classically \cclass{PSPACE}-complete, and we prove that its quantum
extension is as well.

\begin{theorem}
For flavor $D$,
    \ruleset{Quantum Node Kayles} remains a \cclass{PSPACE}-complete game (independent of superpositon width).
    \label{thm:nodeKayles}
\end{theorem}


\begin{figure}[h]
\begin{center}
\begin{tikzpicture}[node distance = 1.5cm]
    \tikzstyle{vertex}=[circle,thick,draw = black, fill = white, minimum size=6mm]
    \tikzstyle{group} = [rectangle, rounded corners, line width=.5mm, draw = black, fill = gray!30]

    \node[vertex] (x1) at (0,0) {$x_1$};
    \node[vertex] (x1b) [right of = x1] {$\overline{x_1}$};

    \node[vertex] (y14) at (7, 0) {$y_{1, 4}$};
    \node[vertex] (y13) [right of = y14] {$y_{1, 3}$};
    \node[vertex] (y12) [right of = y13] {$y_{1, 2}$};

    \node[vertex] (x2T) at (0,-3) {$x_{2T}$};
    \node[vertex] (x2F) [right of = x2T] {$x_{2F}$};
    \node[vertex] (x2Tb) [right of = x2F] {$\overline{x_{2T}}$};
    \node[vertex] (x2Fb) [right of = x2Tb] {$\overline{x_{2F}}$};

    \node[vertex] (y24) at (7, -3) {$y_{2, 4}$};
    \node[vertex] (y23) [right of = y24] {$y_{2, 3}$};

    \node[vertex] (x3T) at (0,-6) {$x_{3T}$};
    \node[vertex] (x3F) [right of = x3T] {$x_{3F}$};
    \node[vertex] (x3Tb) [right of = x3F] {$\overline{x_{3T}}$};
    \node[vertex] (x3Fb) [right of = x3Tb] {$\overline{x_{3F}}$};


    \node[vertex] (c1) at (0,-9) {$C_1$};
    \node[vertex] (c2) [right of = c1] {$C_2$};
    \node[vertex] (c3) [right of = c2] {$C_3$};

    \begin{scope}[on background layer]
        \node[fit={(x1) (x1b) (y12) (y13) (y14)}, group] (Level1x) {};
        \node[fit={(x2T) (x2F) (x2Tb) (x2Fb) (y23) (y24)}, group] (Level2x) {};
        \node[fit = {(x3Fb)}] (Level3xConnector) {};
        \node[fit={(x3T) (x3F) (x3Tb) (x3Fb)}, group] (Level3x) {};
        \node[fit={(c1) (c2) (c3)}, group] (Clauses) {};
    \end{scope}

    \node (Label1) at (-3,  0) {Level 1};
    \node (Label2) at (-3, -3) {Level 2};
    \node (Label3) at (-3, -6) {Level 3};
    \node (LabelC) at (-3, -9) {Level 4};

    \node (ClauseConnectLeft)  at (7, -8) {};
    \node (ClauseConnectRight) at (7.25, -7.88) {};
    \node (Level3ConnectAbove) at (9.43, -4.2) {};
    \node (Level3ConnectBelow) at (9.30, -4.5) {};

    \draw[-] (c2) to [out=-45,in=-135] (6,-9) to [in=-50] (x2Fb);

    \path[-]
        (x1) edge [] node [] {} (x2T)
        (x1) edge [] node [] {} (x2Tb)
        (x1b) edge [] node [] {} (x2F)
        (x1b) edge [] node [] {} (x2Fb)
        (x2T) edge [] node [] {} (x3T)
        (x2T) edge [] node [] {} (x3Tb)
        (x2F) edge [] node [] {} (x3T)
        (x2F) edge [] node [] {} (x3Tb)
        (x2Tb) edge [] node [] {} (x3F)
        (x2Tb) edge [] node [] {} (x3Fb)
        (x2Fb) edge [] node [] {} (x3F)
        (x2Fb) edge [] node [] {} (x3Fb)
        (c1) edge [bend left=45] node [] {} (x1)
        (c1) edge [bend left=45] node [] {} (x2T)
        (c1) edge [bend left=60] node [] {} (x2F)
        (c1) edge [] node [] {} (x3Tb)
        (c1) edge [] node [] {} (x3Fb)
        (c2) edge [] node [] {} (x2Tb)
        (c2) edge [] node [] {} (x3T)
        (c2) edge [] node [] {} (x3F)
        (c3) edge [] node [] {} (x3T)
        (c3) edge [] node [] {} (x3F)
        (c3) edge [] node [] {} (x3Tb)
        (c3) edge [] node [] {} (x3Fb)
        (y14) edge [] node [] {} (Level2x)
        (y14) edge [bend left=40] node [] {} (Level3ConnectBelow)
        (y13) edge [] node [] {} (Level2x)
        (y13) edge [bend left=50] node [] {} (ClauseConnectLeft)
        (y12) edge [bend left=10] node [] {} (Level3ConnectBelow)
        (y12) edge [bend left=39] node [] {} (ClauseConnectLeft)
        (y24) edge [] node [] {} (Level3x)
        (y23) edge [bend left=40] node [] {} (ClauseConnectLeft)
        ;
    \path[-]
        (Clauses) [line width=0.5mm] edge [bend right=10] node [] {}(ClauseConnectRight)
        (Level3xConnector) [line width=0.4mm] edge [bend right=41] node [] {} (Level3ConnectAbove)
        ;
\end{tikzpicture}
\end{center}
\caption{Example of the \cclass{PSPACE} reduction from \ruleset{QBF} to \ruleset{Quantum Node Kayles}.  QBF: $\exists x_1: \forall x_2: \exists x_3: (x_1 \vee x_2 \vee \overline{x_3}) \wedge (\overline{x_2} \vee x_3) \wedge (x_3 \vee \overline{x_3})$}
\label{fig:kayles}
\end{figure}
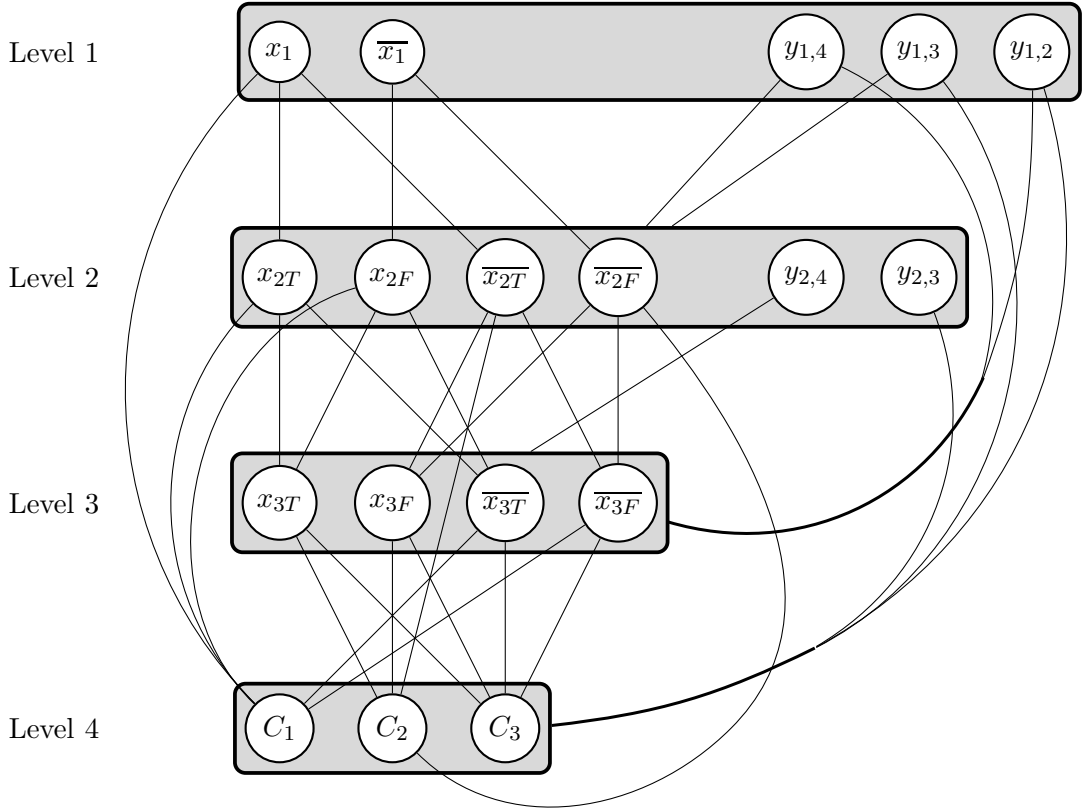

Our proof of this theorem follows a similar plan to the original proof of the \cclass{PSPACE}-hardness of \ruleset{Node Kayles}\cite{DBLP:journals/jcss/Schaefer78}.

\begin{proof}
We reduce from instances of the classic \ruleset{QBF} problem where true assigns both the first and last variable.
This assumption is without loss of generality, since if a instance ended with \texttt{false}, we can just give Player True a variable $x_{n+1}$ and add an extra clause $(x_{n+1} \vee \bar{x}_{n+1})$.

To create the reduction seen in \cref{fig:kayles}, we do the following:

\begin{enumerate}
  \item For $x_1$, create two vertices $x_1$ and $\bar{x}_1$.
  \item For all other variables $x_i$, create four vertices $x_{iT}$, $x_{iF}$, $\bar{x}_{iT}$, and $\bar{x}_{iF}$. We refer to the group of these as $X_i$
  \item For each original variable $x_i$, create vertices $y_{ij}$ for all $i < j \leq n + 1$
  \item Connect all vertices in each $X_{i}$ to each other vertex in $X_i$ and to each $y_{ij}$. We will call the clique this forms "level $i$".
  \item Create a vertex for each clause $C_i$ and add edges between all clause vertices.  We consider all clause-vertices to be at level $n+1$.
  \item Connect each $x_i$ to $x_{(i + 1)T}$ $\bar{x}_{(i + 1)T}$. Connect each $\bar{x}_i$ to $x_{(i + 1)F}$ and $\bar{x}_{(i + 1)F}$
  \item Connect each $x_i$ to every clause $C_j$ where $x_i$ appears unnegated, and connect each $\bar{x}_i$ to each clause $C_j$ where $x_i$ appears negated.
  \item Connect each $y_{ij}$ to every other vertex in the graph except the vertices on level $j$.
\end{enumerate}

This game of quantum node Kayles is winnable by the first player if and only if the original \ruleset{QBF} is a true instance.

First, we will show that this is true classically.
We refer to the first player as ``True''  and the second player as ``False''.
We prescribe the following order of play: on turn $i$, the current player will play classically on a vertex representing variable $x_i$.
When all of these variables are selected, False will then attempt to play on a remaining clause vertex. If they can, then False will win, as each level of vertices has had a vertex played in it, and each level is a clique.
If False can't play on a clause vertex, then they are all covered, so no vertices remain in the graph and True wins. Clearly if both players play on the variable vertices optimally, then False can chose a clause vertex at the end if and only if they could in the \ruleset{QBF}.

It remains to be shown that if this prescribed play is violated by the classically losing player, they will lose.
Assume the order is observed until level $i$, then it is on turn $i+1$ it is violated.
For each of these violating plays, the classically-winning player has a winning response:

\begin{itemize}
\item If they classically played a vertex on level $j$, then the other player simply plays on $y_{(i+ 1)j}$ and wins.
\item If they classically played on $y_{(i + 1)j}$, then the other player simply needs to play on any vertex on level $j$ to win.
\item If the player made a quantum move that includes a vertex in a later level $j$, the other can play $y_{(i + 1)j}$, which will collapse all realizations where they didn't play on level $j$. In those realizations, there are no moves left for the other player.
\item For any $y_{(i+ 1)j}$ chosen, the other player can continue playing normally, since those will collapse immediately upon making the next move.
\item Since play has been classical up to turn $i$, the previous $x_i$ selection leaves only two options open for for their quantum move. This quantum move will collapse immediately after the other player makes their next move.
We will say they select in a way such that it collapses to \texttt{true}, and make the move after the correct response for a
logic assignment by the previous player.
\end{itemize}
\end{proof}

\begin{corollary}\label{NodeKaylesAllFlavors}
\ruleset{Quantum Node Kayles} is \cclass{PSPACE}-complete in all flavors
\end{corollary}
\begin{proof}
Its in \cclass{PSPACE} because the game is polynomially short.

We can preserve the complexity by creating a new vertex $v$ for each vertex $u$ where $v$ shares all of the same edges as $u$, along with an edge $(u, v)$, as a selection of any $u$ and $v$ achieves exactly the same result, and is thus equivalent. Then, by the symmetry property this game has, and by \cref{WeakFlavor}, it is \cclass{PSPACE}-hard in all flavors.
\end{proof}

\noindent\ruleset{Quantum Bigraph Node Kayles}:
\ruleset{Bigraph Node Kayles} is a variant of \ruleset{Node Kayles} in which nodes are partitioned into red nodes and blue nodes, where the blue player can only play on blue vertices and the red player can only play on red vertices.
Classically, \ruleset{Bigraph Node Kayles} is \cclass{PSPACE}-complete.
Below we prove that it remains \cclass{PSPACE}-complete in the quantum setting.

\begin{theorem}
For all flavors,
   \ruleset{Quantum BiGraph Node Kayles} is \cclass{PSPACE}-complete (independent of
   superposition width).
    \label{thm:bigraphNK}
\end{theorem}

\begin{proof}
We use the proof from regular \ruleset{Quantum Node Kayles} and color the vertices based on who should play on the vertex if the player is following the rules. Since the legal moves are still possible, players have less cheating options, and if they do cheat, the legal response is still available, the reduction still holds.
See Figure \ref{fig:bigraphNodeKayles} for illustration.
\end{proof}

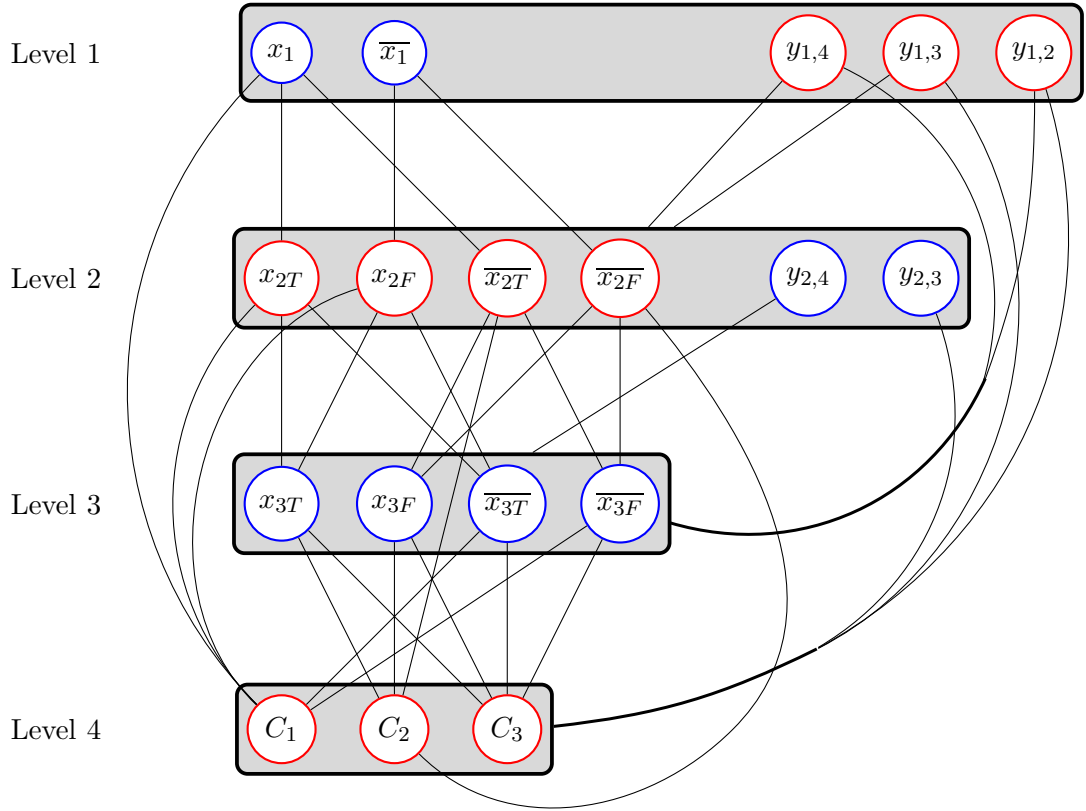
\begin{figure}[h]
\begin{center}
\begin{tikzpicture}[node distance = 1.5cm]
    \tikzstyle{vertex}=[circle,thick,draw = black, fill = white, minimum size=6mm]
    \tikzstyle{blue-vertex}=[circle,thick,draw = blue, fill = white, minimum size=6mm]
    \tikzstyle{red-vertex}=[circle,thick,draw = red, fill = white, minimum size=6mm]
    \tikzstyle{group} = [rectangle, rounded corners, line width=.5mm, draw = black, fill = gray!30]
    \tikzstyle{blue-group} = [rectangle, rounded corners, line width=.5mm, draw = black, fill = gray!30]
    \tikzstyle{red-group} = [rectangle, rounded corners, line width=.5mm, draw = black, fill = gray!30]

    \node[blue-vertex] (x1) at (0,0) {$x_1$};
    \node[blue-vertex] (x1b) [right of = x1] {$\overline{x_1}$};

    \node[red-vertex] (y14) at (7, 0) {$y_{1, 4}$};
    \node[red-vertex] (y13) [right of = y14] {$y_{1, 3}$};
    \node[red-vertex] (y12) [right of = y13] {$y_{1, 2}$};

    \node[red-vertex] (x2T) at (0,-3) {$x_{2T}$};
    \node[red-vertex] (x2F) [right of = x2T] {$x_{2F}$};
    \node[red-vertex] (x2Tb) [right of = x2F] {$\overline{x_{2T}}$};
    \node[red-vertex] (x2Fb) [right of = x2Tb] {$\overline{x_{2F}}$};

    \node[blue-vertex] (y24) at (7, -3) {$y_{2, 4}$};
    \node[blue-vertex] (y23) [right of = y24] {$y_{2, 3}$};

    \node[blue-vertex] (x3T) at (0,-6) {$x_{3T}$};
    \node[blue-vertex] (x3F) [right of = x3T] {$x_{3F}$};
    \node[blue-vertex] (x3Tb) [right of = x3F] {$\overline{x_{3T}}$};
    \node[blue-vertex] (x3Fb) [right of = x3Tb] {$\overline{x_{3F}}$};


    \node[red-vertex] (c1) at (0,-9) {$C_1$};
    \node[red-vertex] (c2) [right of = c1] {$C_2$};
    \node[red-vertex] (c3) [right of = c2] {$C_3$};

    \begin{scope}[on background layer]
        \node[fit={(x1) (x1b) (y12) (y13) (y14)}, blue-group] (Level1x) {};
        \node[fit={(x2T) (x2F) (x2Tb) (x2Fb) (y23) (y24)}, red-group] (Level2x) {};
        \node[fit = {(x3Fb)}] (Level3xConnector) {};
        \node[fit={(x3T) (x3F) (x3Tb) (x3Fb)}, blue-group] (Level3x) {};
        \node[fit={(c1) (c2) (c3)}, red-group] (Clauses) {};
    \end{scope}

    \node (Label1) at (-3,  0) {Level 1};
    \node (Label2) at (-3, -3) {Level 2};
    \node (Label3) at (-3, -6) {Level 3};
    \node (LabelC) at (-3, -9) {Level 4};

    \node (ClauseConnectLeft)  at (7, -8) {};
    \node (ClauseConnectRight) at (7.25, -7.88) {};
    \node (Level3ConnectAbove) at (9.43, -4.2) {};
    \node (Level3ConnectBelow) at (9.30, -4.5) {};

    \draw[-] (c2) to [out=-45,in=-135] (6,-9) to [in=-50] (x2Fb);

    \path[-]
        (x1) edge [] node [] {} (x2T)
        (x1) edge [] node [] {} (x2Tb)
        (x1b) edge [] node [] {} (x2F)
        (x1b) edge [] node [] {} (x2Fb)
        (x2T) edge [] node [] {} (x3T)
        (x2T) edge [] node [] {} (x3Tb)
        (x2F) edge [] node [] {} (x3T)
        (x2F) edge [] node [] {} (x3Tb)
        (x2Tb) edge [] node [] {} (x3F)
        (x2Tb) edge [] node [] {} (x3Fb)
        (x2Fb) edge [] node [] {} (x3F)
        (x2Fb) edge [] node [] {} (x3Fb)
        (c1) edge [bend left=45] node [] {} (x1)
        (c1) edge [bend left=45] node [] {} (x2T)
        (c1) edge [bend left=60] node [] {} (x2F)
        (c1) edge [] node [] {} (x3Tb)
        (c1) edge [] node [] {} (x3Fb)
        (c2) edge [] node [] {} (x2Tb)
        (c2) edge [] node [] {} (x3T)
        (c2) edge [] node [] {} (x3F)
        (c3) edge [] node [] {} (x3T)
        (c3) edge [] node [] {} (x3F)
        (c3) edge [] node [] {} (x3Tb)
        (c3) edge [] node [] {} (x3Fb)
        (y14) edge [] node [] {} (Level2x)
        (y14) edge [bend left=40] node [] {} (Level3ConnectBelow)
        (y13) edge [] node [] {} (Level2x)
        (y13) edge [bend left=50] node [] {} (ClauseConnectLeft)
        (y12) edge [bend left=10] node [] {} (Level3ConnectBelow)
        (y12) edge [bend left=39] node [] {} (ClauseConnectLeft)
        (y24) edge [] node [] {} (Level3x)
        (y23) edge [bend left=40] node [] {} (ClauseConnectLeft)
        ;
    \path[-]
        (Clauses) [line width=0.5mm] edge [bend right=10] node [] {}(ClauseConnectRight)
        (Level3xConnector) [line width=0.4mm] edge [bend right=41] node [] {} (Level3ConnectAbove)
        ;
\end{tikzpicture}
\end{center}
\caption{Example of the \cclass{PSPACE} reduction from QBF to \ruleset{Q-BiGraphNodeKayles}, using the same formula as in \cref{fig:kayles}.  Here Blue is going first as the True player.}
\label{fig:bigraphNodeKayles}
\end{figure}

\noindent \ruleset{Quantum Snort}:
\ruleset{Snort} is a game where one player is a blue player, and the other is a red player. Players alternate placing tokens of their color onto vertices of a given graph.
Players can't place tokens on vertices adjacent to vertices with a token of the opponent's color.
\ruleset{Snort} is classically \cclass{PSPACE}-complete.
We prove \ruleset{Quantum Snort} remains
\cclass{PSPACE}-complete.

\begin{theorem}
For flavors $D$ and $C'$,
  \ruleset{Quantum Snort} is \cclass{PSPACE}-complete (independent of superposition width).
    \label{thm:snort}
\end{theorem}

\begin{proof}
We reduce from \ruleset{Bigraph Node Kayles}.
We simply create an extra vertex connected to all red vertices and place a red token on it, and then create a vertex with a blue token on it connected to all blue vertices,
\end{proof}

\begin{figure}[h]
\begin{center}
\begin{tikzpicture}[node distance = 1.5cm]
    \tikzstyle{vertex}=[circle,thick,draw = black, fill = white, minimum size=6mm]
    \tikzstyle{blue-vertex}=[circle,thick,draw = blue, fill = white, minimum size=6mm]
    \tikzstyle{red-vertex}=[circle,thick,draw = red, fill = white, minimum size=6mm]
    \tikzstyle{group} = [rectangle, rounded corners, line width=.5mm, draw = black, fill = gray!30]
    \tikzstyle{blue-group} = [rectangle, rounded corners, line width=.5mm, draw = black, fill = gray!30]
    \tikzstyle{red-group} = [rectangle, rounded corners, line width=.5mm, draw = black, fill = gray!30]

    \node[blue-vertex] (a) at (0,0) {$a$};
    \node[red-vertex] (b) [right of = a] {$b$};
    \node[red-vertex] (c) [below left of = a] {$c$};
    \node[blue-vertex] (e) [below right of = a] {$e$};
    \node[red-vertex] (d) [below of = a] {$d$};
    \node[blue-vertex] (h) [below right of = d] {$h$};
    \node[red-vertex] (f) [above right of = h] {$f$};
    \node[blue-vertex] (g) [left of = h] {$g$};

    \path[-]
        (a) edge [] node [] {} (b)
        (a) edge [] node [] {} (c)
        (a) edge [] node [] {} (e)
        (b) edge [] node [] {} (e)
        (c) edge [] node [] {} (d)
        (c) edge [] node [] {} (g)
        (d) edge [] node [] {} (e)
        (d) edge [] node [] {} (g)
        (d) edge [] node [] {} (h)
        (e) edge [] node [] {} (f)
        (f) edge [] node [] {} (h)
        ;

    \node[vertex] (a2) at (6,0) {$a$};
    \node[vertex] (b2) [right of = a2] {$b$};
    \node[vertex] (c2) [below of = b2] {$c$};
    \node[vertex] (d2) [below of = c2] {$d$};
    \node[vertex] (e2) [below of = a2] {$e$};
    \node[vertex] (f2) [below of = d2] {$f$};
    \node[vertex] (g2) [below of = e2] {$g$};
    \node[vertex] (h2) [below of = g2] {$h$};
    \node[vertex, fill = blue] (blue) at (5,-2) {};
    \node[vertex, fill = red] (red) at (8.5,-2) {};

    \path[-]
        (a2) edge [] node [] {} (b2)
        (a2) edge [] node [] {} (c2)
        (a2) edge [] node [] {} (e2)
        (b2) edge [] node [] {} (e2)
        (c2) edge [] node [] {} (d2)
        (c2) edge [] node [] {} (g2)
        (d2) edge [] node [] {} (e2)
        (d2) edge [] node [] {} (g2)
        (d2) edge [] node [] {} (h2)
        (e2) edge [] node [] {} (f2)
        (f2) edge [] node [] {} (h2)
        (blue) edge [] node [] {} (a2)
        (blue) edge [] node [] {} (e2)
        (blue) edge [] node [] {} (g2)
        (blue) edge [] node [] {} (h2)
        (red) edge [] node [] {} (b2)
        (red) edge [] node [] {} (c2)
        (red) edge [] node [] {} (d2)
        (red) edge [] node [] {} (f2)
        ;
\end{tikzpicture}
\end{center}
\caption{Graph for \ruleset{Q-BiGraphNodeKayles}, followed by the result of the reduction to \ruleset{Q-Snort} on that graph.}
\label{fig:bigraphNodeKayles-2}
\end{figure}
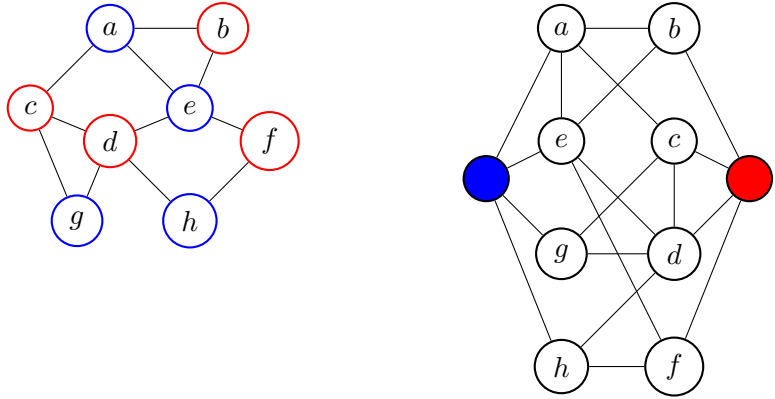


\section{A Complexity Barrier of Quantum Leap}\label{Sec:PSpaceUpperBound}

While the introduction of quantum moves
  potentially increases the complexity of
  combinatorial games,
  we show in this section that essentially all \cclass{PSPACE}-solvable combinatorial games
  remain in \cclass{PSPACE}
   with strengthened quantum power.
We prove that for a large family of combinatorial games---including many of our favorite \cclass{PSPACE}-complete ones such as \ruleset{QSAT}, \ruleset{Hex}, \ruleset{Geography}, \ruleset{Atropos}, \ruleset{Avoid true}, \ruleset{Node Kayles}, {\em  etc}, {\em etc}---quantum extensions (of moves and superpositions) does
 not elevate the complexity of the game to beyond \cclass{PSPACE}.
Particularly, this general upper bound applies to
  any \cclass{PSPACE}-solvable combinatorial games
  satisfying the following properties:

\begin{definition}[Polynomially-Short Games]
A combinatorial game is called a {\em polynomially-short game} if for any of its instances,
the height of its game tree is polynomial in the instance's descriptive size.
\end{definition}

For example, \ruleset{Geography} is a polynomially-short
game.
Each \ruleset{Geography} instance is defined
 by a graph $G=(V,E)$, whose number of nodes,
  $|V|$,
  characterizes the instance's descriptive size
  (because the number of (directed) edges $|E|$ can be at most $|V|^2$).
Note that \ruleset{Geography}
 is also a polynomially-dense game because it has only $|V|$ possible moves.
Because no node of $G$ can be selected twice in the game, the height of the game tree
  of \ruleset{Geograph} on $G$ is at most $|V|$.
In contrast, \ruleset{Nim}---although polynomial-time solvable---is not a polynomially-short game, because for almost all \ruleset{Nim} instances,
the height of their game trees are exponential in their description sizes.
Note also that \ruleset{Nim} is also a super-polynomially dense game.

We can use standard technique, with an algorithm  generating and evaluating the game trees, to show that all polynomially-short games are \cclass{PSPACE}-solvable.

In this section, we prove the following theorem, which establishes a sharp complexity-theoretical ceiling for quantum leap in the complexity of these well-studied \cclass{PSPACE}-complete combinatorial games.

\begin{theorem}[Quantumized Polynomially-Short Games is in \cclass{PSPACE}]
\label{Theo:Upperbound}
For any polynomially-short game $Z$, quantum flavor $\phi$, and superposition width $w$ polynomially in the descriptive complexity of $Z$,  $Z^{\QuantumLift(\phi,w)}$
is \cclass{PSPACE}-solvable.
\end{theorem}

\subsection{Quantum Game Trees Are No Taller}

Before proving this theorem,
  let's start with some high-level
   intuition of the proof.
To establish the \cclass{PSPACE}-solvability---like in the classical case---we need to evaluate the game tree of the quantumized games in polynomial space.
A key observation, in this context, is that quantum extension does not increase the height game tree: 

\begin{proposition}[Height of Game Trees]\label{prop:QuantumHeight}
For any game $Z$, flavor $\phi$, and
  superposition width $w$,
the game tree for $Z^{\QuantumLift(\phi,w)}$ can't be taller than the game tree for $Z$.
\end{proposition}
\begin{proof}
Every position associated with the game tree for $Z^{\QuantumLift(\phi,w)}$
  is a superposition of realizations, each of which is defined
  by the interaction between the initial game position (which is classical) and
  a sequence of classical moves.
Furthermore, the length of the sequence is equal to the length
  of the sequence of
  quantum moves in $Z^{\QuantumLift(\phi,w)}$
   that defines the superposition.
Therefore, each realization must appear as a position in the game tree
  for $Z$ at the same level.
Thus, the game tree for $Z^{\QuantumLift(\phi,w)}$ can not be taller
  than the game tree for $Z$.
In fact, for flavor $D$ and $C'$, $C'$, and $B$,
  quantum lift preserves the height of the game tree.
For flavor $A$, quantum lift may make the game tree shorter.
\end{proof}

Proposition \ref{prop:QuantumHeight}
 is a crucial observation to our algorithmic proof.
It implies that we {\em can still} apply the DFS-based framework---just like in the classical case---to evaluate the game trees,
 {\em provided} that we can effectively evaluate the
  outcome at each node in the game trees.
Thus, one of the main technical barriers that we need to overcome is to evaluate, in polynomial space, these superpositions, which potentially could have exponential number of realizations.


\subsection{{\rm\mbox{\ruleset{Quantum Node Kayles}}} in \cclass{PSPACE}}

For clarity of presentation,
  we first focus our proof of Theorem \ref{Theo:Upperbound}
  on the quantum extension of a concrete \cclass{PSPACE}-complete game, namely
  \ruleset{Node Kayles}, in flavor $D$, and with
  superposition width $2$.
This proof is in fact quite general, directly extendable other quantum flavors, to cover all {\em polynomially-dense} \cclass{PSPACE}-complete games including \ruleset{QSAT}, \ruleset{Hex}, \ruleset{Geography}, \ruleset{Atropos}, \ruleset{Avoid true}, {\em  etc}, {\em etc}.
We then extend our algorithmic construction to complete the proof for Theorem \ref{Theo:Upperbound}.

\begin{theorem}[Solvability of \ruleset{Quantum Node Kayles}] \label{Theo:QNK}
\ruleset{Quantum Node Kayles} with flavor $D$ and superposition width $2$ is polynomial-space solvable.
\end{theorem}
\begin{proof}
Suppose $Z$ is defined by a graph $G = (V,E)$ with $V=[n]$ and $|E|= m$.
Consider the game tree of $Z^{\QuantumLift(D,2)}$
  consisting either classic moves (i.e., selecting a node from $V$)
 or quantum moves (i.e., selecting a superposition of two nodes from $V$).
Like its classical counterpart---\ruleset{Node Kayles}---this game tree has height at most $n$ (Proposition \ref{prop:QuantumHeight}).
Furthermore,  each internal node in the game tree has at most $O(n^2)$ children
because there are
at most $n$ possible classical moves and $n(n-1)/2$ possible quantum moves.
However, unlike in \ruleset{Node Kayles}, some nodes in the game tree---either internal or at the leaf level---could be a superposition of
 possibly exponential number of realizations,
 because the number of realizations is usually exponential in
 the number of quantum moves leading to that superposition.
Thus, we cannot simply apply the traditional \cclass{PSPACE}
  postorder-traversal based (DFS) procedure to evaluate this game tree.

In our polynomial-space algorithm
  for \ruleset{Quantum Node Kayles},
 we will use an "inner" tree-traversal for analyzing the superposition inside the traditional postorder-traversal based DFS evaluation of the game tree.
We first prove the following:
 \begin{itemize}
 \item {\bf Identifying Leaves}: We can, in polynomial space, determine if a node in the game tree for \ruleset{Quantum Node Kayles} is a leaf node or an internal node.
\item {\bf Outcomes at Leaves}: We can, in polynomial space, evaluate the outcome at each leaf in the game tree for  \ruleset{Quantum Node Kayles}, even if its superposition has exponential number of realizations.
 \item {\bf Outcomes at Internal Nodes}: We can, in polynomial space,  evaluate the outcome class of each internal node in the game tree, given the outcome classes of its children, (in the context of recursive evaluation of the game tree for \ruleset{Quantum Node Kayles}, as in the standard \cclass{PSPACE} argument), even if its superposition has exponential number of realizations.
\end{itemize}

Let's first focus on the task of identifying leaves in the game tree for $Z^{\QuantumLift(D,2)}$.
Our approach to outcome evaluations of super-polynomial superposition at both leaves and internal nodes is similar.
We will address them afterwards.

The key observation is the following:
Although some superposition associated with the game tree for $Z^{\QuantumLift(D,2)}$ can have exponential number of realizations,  each superposition's realizations can be organized by another tree of
$O(n)$ depth, which we will refer to as a {\em QP-tree} (short for ``quantum position tree''),
  in order to differentiate them from the game tree itself.

Like its corresponding superposition in the game tree for $Z^{\QuantumLift(D,2)}$,
 each QP-tree is formulated by the sequence of ``classical'' and ``quantum'' moves
 defining the superposition.
Each move in the sequence grows the QP-tree by one level.
At level $0$ is the root, corresponding to the starting position of $Z^{\QuantumLift(D,2)}$.
For $t=1$ onward, suppose $m_t$ is the $t^{th}$ move in the sequence.
Then, $m_t$ defines level $t$ positions from positions at level $(t-1)$ in the QP-tree.
Consider a QP-tree node with position $b$ at level
$(t-1)$, there are two cases depending on
whether $m_t$ is a classical or quantum move.
\begin{enumerate}
\item If $m_t = \sigma_t$ is classical and feasible for $b$, then $m_t$ creates a single child in the QP-tree, with position $\rho(b,\sigma_t)$.
\item If $m_t = [\sigma_{t,1}|\sigma_{t,2}]$
is a quantum move, then
(a) if both $\sigma_{t,1}$ and $\sigma_{t,2}$ are feasible for $b$, then $m_t$ creates
two children in the QP-tree,
 with positions $\rho(b,\sigma_{t,1})$ and $\rho(b,\sigma_{t,2})$, respectively;
(b) if one of them, say $\sigma_{t,2}$ is not feasible for $b$, then $m_t$ creates
 one child in the QP-tree, with position $\rho(b,\sigma_{t,1})$;
 (c) if both $\sigma_{t,1}$ and $\sigma_{t,2}$ are infeasible, then no children is created.
\end{enumerate}
Based on this construction, the height of
  the QP-tree is bounded by
  the length---let's denote it by $T$ for now---of the sequence of moves  defining it.
The superposition is NULL if there is the QP-tree is empty at level $T$; otherwise, the
superposition consists of all distinct positions at level $T$ of the QP-tree.\footnote{The QP-trees for other quantum flavors are slightly more subtle, and we will discuss them later in the section.}


Let's now return to the game tree for $Z^{\QuantumLift(D,2)}$,
in which, a node $v$ is a leaf if its associated superposition has no feasible classical or quantum move.
For quantum flavor $D$, if any realization in $v$'s associated superposition has a feasible move (either classical or quantum), then $v$ is an internal
node in the game tree.
Thus, to determine if  $v$'s superposition contains any
  realization with a feasible move, we apply an postorder-traversal based DFS procedure to evaluate its QP-tree.
Suppose there are $T$ moves in the sequence defining $v$.
Then, we use the DSF procedure to test if any of its position at level $T$ has a feasible move.\footnote{In fact, for flavor $D$, it is sufficient to simply test whether any  level-$T$ position has a feasible classical move. The reason is the following:  In flavor D, for any position, the existence of a feasible quantum move implies the existence of a feasible classical move.}
Upon encountering the first feasible move at level $T$, we can conclude that $v$ is not a leaf in the game tree.
If the procedure terminates without detecting any feasible move out of level $T$, then we conclude that $v$ is a leaf in the game tree.
The DFS procedure for analyzing QT-trees can be implemented in space polynomial in $T$.

Because \ruleset{Node Kayles} is impartial, by Proposition \ref{prop:Impartial},
\ruleset{Quantum Node Kayles} is also impartial.
Thus, each superposition in the game tree can only
  be in one of the two outcome classes:
  \outcomeClass{N} or \outcomeClass{P}.
Because the superposition associated with every leaf of the game tree is in \outcomeClass{P}, in polynomial space, we can both identify whether or not a node in the game tree for \ruleset{Quantum Node Kayles} is a leaf, and determine the outcome class if the node is a leaf.

Because this theorem only considers quantum moves with superposition width 2, each node in the game tree has only polynomial of possible feasible moves.
Thus, we can in polynomial space determine the outcome class of each superposition in the game tree, and particularly the outcome class of the starting position (associated with the root of the game tree) by running the post-order-traversal DFS procedure with our polynomial-space QP-tree evaluation algorithm. Therefore, \ruleset{Quantum Node Kayles} with flavor $D$ and superposition width $2$ is polynomial-space solvable.
\end{proof}

\subsection{Quantumized Complexity Polynomially-Short Games}

In this subsection, we extend the proof of Theorem \ref{Theo:QNK} from \ruleset{Quantum Node Kayles} (with superposition width $2$ in flavor $D$) to quantum lift of all polynomially-short games (with polynomial supervision width in any  quantum flavor).

\begin{proof}(of Theorem \ref{Theo:Upperbound})
We will address the following basic issues to extend our \cclass{PSPACE}-proof to general polynomially-short games:
\begin{itemize}

\item [I] {\bf Going beyond constant superposition width}: Our proof for Theorem \ref{Theo:QNK} can be directed extended from 2 to any constant supervision width,
because the only property that it uses is the fact that every node in the game tree has polynomial number of children.
However, for quantum moves with polynomial superposition, some nodes in the game tree may have exponentially number of children.
We need to show that they still can be processed in polynomial space.

\item [II] {\bf Going beyond polynomially-dense games}: Like many popular combinatorial games, \ruleset{Node Kayles} is polynomially dense, i.e., the alphabet of moves in \ruleset{Node Kayles}---i.e., the selection of a node in the instance-defining graph---is of polynomial size in the descriptive complexity.
However, in general, combinatorial games could be super-polynomially dense.

\item [III] {\bf Going beyond quantum flavor $D$}: Other quantum flavors require more subtle decision on the feasibility of classical moves that may impact the decision on whether or not a node in the game tree is a leaf.

\item [IV] {\bf Going beyond impartial games}: \ruleset{Node Kayles} is an impartial game, which simplifies the logic of outcome classes.
In general, combinatorial games could be partisan.
\end{itemize}

Of the four, the main technical matters are the  exponential explosion and the complex structure introduced by the first two issues.
So we will focus on them first.

Let $\cclass{EXP}(\cclass{POLY})$ be the family of functions that can be expressed by the natural exponential of a polynomial function.
For example,
  $e^{3n^3+2n^2+1}$ is such a function.

We will use the following basic facts to show
  that the number of nodes in the game trees
  for quantum polynomially-short games
  is not hopelessly large.

\begin{proposition}\label{prop:ExpPoly}
$\forall$ polynomial $h$ and $f\in \cclass{EXP}(\cclass{POLY})$,
$f^h \in \cclass{EXP}(\cclass{POLY})$.
\end{proposition}

\begin{proposition}\label{prop:ExpTExp}
$\forall$  $f, g\in \cclass{EXP}(\cclass{POLY})$,
$f\times g \in \cclass{EXP}(\cclass{POLY})$.
\end{proposition}

\begin{proposition}[Size of Game Trees]
For any polynomially-short game instance $Z$,
  the number of positions in its game tree is upper bounded by a function in $\cclass{EXP}(\cclass{POLY})$ in terms of its descriptive complexity.
\end{proposition}

Note that the number of nodes in the game tree  for $Z$ is upper bounded by $O(fan\mbox{\rm -}out^{height})$,
 where $fan\mbox{\rm -}out$ and $height$, respectively, are the game tree's fan-out and height.
The proposition then follows from Proposition \ref{prop:ExpPoly} because that the game tree  has height bounded by a polynomial and  fan-out bounded by a function in \cclass{EXP}(\cclass{POLY}) in its descriptive complexity.



\begin{proposition}
For any polynomially-short game instance $Z$, for any quantum flavor $\phi\in \{A,B,C,C',D\}$, and polynomial
superposition width $w$,
  the number of realizations in each superposition associating with the game tree
  for $Z^{\QuantumLift(\phi,w)}$
  is upper bounded by a function in $\cclass{EXP}(\cclass{POLY})$ in terms of its descriptive complexity. Moreover, realizations in the superposition can be organized according to a tree structure with fan-out and height, polynomial in $Z$'s descriptive complexity.
\end{proposition}

To see this proposition, we first note that by
Proposition \ref{prop:QuantumHeight} that the height of
game tree for $Z^{\QuantumLift(\phi,w)}$ is upper bounded by a polynomial $h$ in terms of $Z$'s' descriptive complexity.
Thus, each superposition is defined by a sequence of at most $h$ quantum/classical moves, and hence, the number of realizations in each superposition is upper bounded by $w^h$.
The tree structure, of height $h$ as well,  for organizing the realizations in the superposition can be constructed in the similar manner as the one in the proof for Theorem \ref{Theo:QNK}.

\begin{proposition}\label{prop:QTreeSize}
For any polynomially-short game instance $Z$,
for any quantum flavor $\phi\in \{A,B,C,C',D\}$,
and polynomial superposition width $w$,
  the number of positions in the game tree
  for $Z^{\QuantumLift(\phi,w)}$
  is upper bounded by a function in $\cclass{EXP}(\cclass{POLY})$ in terms of its descriptive complexity.
\end{proposition}

To see this proposition, we first note that by
Proposition \ref{prop:QuantumHeight} that the height of
game tree for $Z^{\QuantumLift(\phi,w)}$ is upper bounded by a polynomial $h$ in terms of $Z$'s' descriptive complexity.
The fan-out of this game tree
is upper bounded by the number of possible moves to the power of $w$, which can be upper bounded  by a function in $\cclass{EXP}(\cclass{POLY})$.
Therefore, following Proposition \ref{prop:ExpPoly},   the number of nodes in the game tree for $Z^{\QuantumLift(\phi,w)}$ is also upper bounded by a function in $\cclass{EXP}(\cclass{POLY})$ in terms of
$Z$'s descriptive complexity.

Proposition \ref{prop:QTreeSize}
  immediately implies that
$Z^{\QuantumLift(\phi,w)}$ is solvable in \cclass{EXPTIME} in terms of game's descriptive complexity.
We now extend Theorem \ref{Theo:QNK} to show that $Z$'s game tree can be evaluated in polynomial space.
In contrast to the game tree in Theorem \ref{Theo:QNK},
the game tree for $Z^{\QuantumLift(\phi,w)}$ potentially have exponential fan-out.

Fortunately, we can enumerate all quantum moves of superposition width $w$ in space (although not in time) polynomial in both $w$ and the descriptive complexity of $Z$.
Recall for succinctly-defined game, by Property \ref{prop:DescriptiveComplexity} (3.SP),
we can use the polynomial-time generative function
$\mu_L: B\times |\Sigma|\rightarrow \Sigma$ and $\mu_R: B\times |\Sigma| \rightarrow \Sigma$ to enumerate all classical moves at a given position in polynomial space.

We first consider the simpler case when $\mu_L$ an $\mu_R$ are independent of the position, (i.e. $\forall b, b'\in B$,
$\mu_L(b,:) = \mu_L(b',:)$ and $\mu_R(b,:) = \mu_R(b',:)$).
In this case, we can simply enumerate all potential quantum moves with superposition width by running a nested loop of depth $w$.
This nested loop  needs only polynomial space (i.e. $w$ times the space needed to enumerate all classical moves.)

For the general case when $\mu_L$ an $\mu_R$ are potentially {\em position sensitive} (i.e., players have local view of the moves depending on the position), we need a slightly more sophisticated enumeration method at  quantum positions with superpolynomial number of realizations.
Suppose at a node in the game tree for $Z^{\QuantumLift(\phi,w)}$ with quantum position
$\mathbb{B}$ defined by a sequence
   $m_T \circ ...\circ m_1$ classical/quantum moves,
where $T$ is polynomial in the descriptive complexity of $Z$ and each $m_t$ ($t\in [T]$) is a classical move
(denoted by $m_t = \sigma_t$) or a quantum move of superposition width up to $w$ (denoted by $m_t = \mathbb{M}_t$).
We can use the following procedure to enumerate all classical moves compatible with realizations in $\mathbb{B} = \rho^{\QuantumLift(\phi,w)}(b_0,m_T \circ ...\circ m_1)$:
\begin{itemize}
\item ({\em Outer Loop}): Apply the polynomial-space post-order-traversal DFS procedure to the QP-tree for $\mathbb{B}$ to enumerate realizations in $\mathbb{B}$;
\item ({\em Inner Loop}): For each realization, we apply the generative function $\mu_L$ an $\mu_R$ to enumerate all its compatible classical moves.
\end{itemize}
Then, in order to enumerate all potential quantum
 moves of superposition width $w$ at position $\mathbb{B}$, we can run a nested loop of depth $w$ with the above procedure.
This QP-tree based method has one more advantage.
For each potential quantum move generated, using $\rho_L$ and $\rho_R$, we can immediately certify whether or not it is feasible for the quantum position $\mathbb{B}$ (i.e., it is a superposition of classical moves, each of which is feasible for some realization in $\mathbb{B}$.)
Thus, we can also determine in polynomial space whether or not $\mathbb{B}$ has no feasible quantum move available to support decisions in quantum flavors $A$ and  $B$.
In addition, for each potential classical move generated,
we can use a post-order-traversal DFS procedure to the QP-tree for $\mathbb{B}$ to determine if it is a safe of respectful classical move for $\mathbb{B}$ (as defined in Definition \ref{Defi:SafeRespectful}).
Thus, we can support the decisions in quantum flavors $C$ and $C'$ in polynomial space.

Thus, we can address issues I-III---by going beyond constant superposition width, polynomially-dense games and , quantum flavor $D$---and extend the polynomial-space solution of Theorem \ref{Theo:QNK} to general polynomially-short games.

The outcome classes for partizan games are more refined than for impartial games.
To go beyond impartial games in our evaluation of the
game tree for $Z^{\QuantumLift(\phi,w)}$, at each node, we sequentially run two threads, one with $L$ as the current player and one with $R$ as the current player, to provide the complete outcome classification for the position.
\end{proof}






\section{Extensions, Conjectures, and Open Questions}

\label{Sec:FinalRemarks}

This has been a fun research project.
We have greatly enjoyed exploring the field of quantum combinatorial games.
The beauty of quantum games---succinctly in representation, rich in structures, explosive in complexity, dazzling for visualization, and sophisticated for strategical reasoning---
has draw us to play concrete games full of subtleties
and to characterize abstract properties pertinent to complexity consequence.

In this research, we have investigated several basic questions in quantum combinatorial game theory:
Mathematically, we have studied how quantum moves impact s the outcome of combinatorial games.
Computationally, we have shown that quantum moves can  cause significant leaps as well as collapses to games' complexity.
Structurally, we have shed new light on the class of polynomially-short games by showing while quantum moves can make game positions (potentially exponentially) wider, they could not make game trees deeper.
Thus, all combinatorial games in this class---including many classical \cclass{PSPACE}-complete games---remain polynomial-space solvable in the quantum setting.

We think that quantum combinatorial game theory
is an wonderful field for continuing research.
It encourages imaginative thinking and has a broad  computational and mathematical connection.
During our research, we have identified far more questions than we could answer.
We have often been humbled by our inability to characterize seemingly simple quantum games
such as \ruleset{Quantum Trilean Nim} with a classical start in which each pile initially has at most two stones (i.e., three values: 0, 1, 2).
In this conclusion section, we will share some open questions and conjectures
uncovered by this research.
We will also discuss some potential
extensions to quantum combinatorial games
in order to make them more attractive for human play and manageable for practical implementation.


\subsection{The Complexity of \mbox{\rm \ruleset{Quantum Nim}}}

As one of the very first quantum games that we have intensively studied, the complexity of \ruleset{Quantum Nim} remains elusive for precise characterization.
While \ruleset{Nim} is polynomial time solvable, it is not a polynomially-short game.
Thus, the  \cclass{PSPACE}-barrier, as established in Theorem \ref{Theo:Upperbound} for quantumized polynomially-short games, does not directly apply to \ruleset{Quantum Nim}, even with a classical start.
On the other hand, while we are able to show that there is a quantum leap of \ruleset{Nim}'s complexity over \cclass{NP},
the dependency of quantum leap on the ``degree of quantumness'' in game's starting position
as well  s
the magnitude of the quantum leap, measured either by  levels of polynomial-time hierarchy or
by its degree above \cclass{PSPACE}, remains unknown.
To put it simply, the complexity of \ruleset{Quantum Nim} is widely open.

We start with a concrete conjecture:

\begin{conjecture}[Quantum Avoid True]
Determining the outcome class of \ruleset{Quantum Avoid True} with a classical start is  \cclass{PSPACE}-complete. Consequently, determining the outcome class of \ruleset{Quantum Nim} with a poly-wide quantum start is \cclass{PSPACE}-hard.
\end{conjecture}

The following two open questions concern about the worst-case complexity upper bound for
\ruleset{Quantum Nim}.

\begin{openquestion}[\ruleset{Quantum Nim}: Algorithmic]
Is \ruleset{Quantum Nim} with:
$$
\left\{\begin{array}{l}
\mbox{\rm a classical start}\\
\mbox{\rm a poly-wide quantum start}\\
\mbox{poly-moves after poly-wide start}
\end{array}
\right.
$$
\cclass{PSPACE}-solvable? If it is not, then are they all \cclass{EXPTIME}-solvable?
\end{openquestion}

By Theorem \ref{Theo:Upperbound}, we know these games are \cclass{EXPSPACE}-solvable.

\begin{openquestion}[\ruleset{Quantum Nim}: Complexity]
Is \ruleset{Quantum Nim} (with poly-moves after a poly-wide start) an \cclass{EXPTIME}-hard game? Can it be \cclass{EXPSPACE}-complete?
\end{openquestion}

By Theorem \ref{theo:QNim}, we proved that \ruleset{Quantum Boolean Nim} with a poly-wide quantum start is $\Sigma_2$-hard.
Also, by a parity argument, we know \ruleset{Quantum Boolean Nim} with a classical start is polynomial-time solvable.
Because \ruleset{Trilean Nim}---in fact, \ruleset{Unary Nim} in which the number of stones in each pile is specified by a unary representation---is  polynomially-short,
by Theorem \ref{Theo:Upperbound}, \ruleset{Quantum Trilean Nim} in particular and \ruleset{Quantum Unary Nim} in general are \cclass{PSPACE}-solvable games.

\begin{openquestion}[\ruleset{Quantum Trilean Nim} and \ruleset{Quantum Unary Nim}]
Is \ruleset{Quantum Trilean Nim} with a classical start polynomial-time solvable?
Is \ruleset{Quantum Unary Nim} with a classical start polynomial-time solvable?
Are various \ruleset{Quantum Substraction} games with a classical start polynomial-time solvable?
Or are the \cclass{PSPACE}-complete?
\end{openquestion}

\subsection{Quantumized Complexity of \mbox{\rm \ruleset{Undirected Geography}}}

\ruleset{Quantum Undirected Geography} is the concrete game we spent most of our time on during this research.
As computer scientists, we are drawn to its graph-theoretical representation and deep connection with the beautiful matching theory.
While we have presented a sharper characterization on its quantumized complexity than \ruleset{Nim}, several basic questions remain open regarding its dependency on the degree-of-quantumness in games' starting positions.

\begin{openquestion}[Reachable Poly-Wide Quantum Geography Positions]
Is \ruleset{Quantum Undirected Geography} with
a ``classical-position reachable''  poly-wide quantum start polynomial-time solvable or
\cclass{PSPACE}-complete?
\end{openquestion}

There are several variants of this open question, including:
\begin{itemize}
\item Is \ruleset{Quantum Undirected Geography} with
log-moves after a classical start polynomial-time solvable or
\cclass{PSPACE}-complete?
\item Is \ruleset{Quantum Undirected Geography} with constant-moves after a classical start polynomial-time solvable or
\cclass{PSPACE}-complete?
\end{itemize}
Answering these questions will provide a more detailed picture for \ruleset{Quantum Undirected Geography}.
So far, we have shown that \ruleset{Quantum Undirected Geography} with a classical start is polynomial-time solvable, with a poly-wide quantum start is \cclass{PSPACE}-complete, and with poly-moves after a classical start \cclass{PSPACE}-complete.

\subsection{Practically-Inspired Frameworks for Quantum Games}

\label{Sec:Practice}

One of the basic challenges in practical implementation of quantum combinatorial games  is that quantum games---as currently defined---may have superposition with superpolynomial (i.e., exponential) number of realizations.
Mathematically, such superposition can still have polynomial-sized description (i.e., by poly-moves after a classical start or poly-moves after a poly-wide quantum start).
But to be a desirable real-world game, more visually-pleasant game boards, with complete and direct
information,  are preferred.

In this subsection, we discuss some more
for incorporating quantum elements
in classical combinatorial games and
 mathematical/algorithmic questions
that these frameworks inspire.

\subsubsection{QCGs with Bounded Quantum Dimensions}

A natural approach to curb quantum explosion is to bound the number of realizations that a ``quantum board'' can have.

\begin{definition}[Bounded-Dimensional Quantum Games]
Suppose $Z^{\QuantumLift(\phi,w)}$ is a
quantum game extended from a classical game  $Z$, with quantum flavor $\phi \in \{A,B,C,C',D\}$, and superposition width $w$.

For any integer $s\geq 0$,
let $Z^{Dim[s]\QuantumLift(\phi,w)}$ denote the variation of $Z^{\QuantumLift(\phi,w)}$, in which a quantum move becomes infeasible if it results in  more than $s$ (none NULL) realizations.
\end{definition}

In other words, $Z^{Dim[s]\QuantumLift(\phi,w)}$ starts at the classical position as defined by  $Z$, and is played according to the quantum rule of $Z^{\QuantumLift(\phi,w)}$ with one exception:
Quantum moves are only allowed if they don't introduce too wide quantum board.
Thus,  all game positions are at most $s$-wide
during the play of $Z^{Dim[s]\QuantumLift(\phi,w)}$.

Many of our open questions extended to quantum game with bounded dimensions.
For example:
\begin{openquestion}
Can constant-dimensional \ruleset{Quantum Nim} be \cclass{NP}-hard to play optimally?
\end{openquestion}


\subsubsection{Demi-QCGs}
Demi-quantum combinatorial games (Demi-QCGs)
  lives in the classical world
  with an initial quantum endowment:
It starts with a quantum board consisting of a superposition of multiple realizations.
During the game, only classical moves are allowed to play on the quantum board.
Therefore, as these games progress,
  the number of realizations in
   the quantum board can only shrink.
 Formally:

 \begin{definition}[Demi-Quantum Games]
For any combinatorial game instance
$Z = (B,\Sigma,\rho_L,\rho_R,b_0)$,
each of its quantum board
$\mathbb{B} = [b_1|...|b_s]$ (i.e., superposition of distinct positions $b_1,..., b_s\in B$)
defines a {\em demi-quantum} game,
$Z^{demi\QuantumLift}[\mathbb{B}]= (B^{\QuantumLift},\Sigma, \rho^{demi\QuantumLift}_L,\rho^{demi\QuantumLift}_R,\mathbb{B})$, with starting quantum board $\mathbb{B}$, and game rules:
$\rho^{demi\QuantumLift}_h(\mathbb{B}',\sigma)
:= \filter(\sigma\otimes\mathbb{B}'), \forall \sigma \in \Sigma, \mathbb{B}'\in B^{\QuantumLift}$.
\end{definition}

We say the demi-quantum game $Z^{demi\QuantumLift}[\mathbb{B}]$
  is $s$ {\em wide}, where $s$ denote the number of realizations in $\mathbb{B}$.
If $s$ is a constant, then
  we call $Z^{demi\QuantumLift}[\mathbb{B}]$      a {\em constant-wide} demi-quantum game.
If $s$ is a polynomial in the descriptive complexity of $Z$, then  we call $Z^{demi\QuantumLift}[\mathbb{B}]$ a {\em poly-wide} demi-quantum game.

For example, in Section \ref{Sec:QATToQNim}, 
  we have built an instance
   of \ruleset{Demi-Quantum Nim} from \ruleset{Avoid True}:
The initial quantum board consists
  of a superposition of \ruleset{Nim} piles reduced from an instance of \cclass{PSPACE}-complete \ruleset{Avoid True}.
We have shown in the proof of Theorem~\ref{QATQNIM} that if from then on, only classical moves are allowed---as in demi-quantum games---then there is a structure-preserving relation between the
resulting \ruleset{Demi-Quantum Nim} and the original \ruleset{Avoid True}.
Therefore, we have already established:

\begin{theorem}
Poly-wide \ruleset{Demi-Quantum Nim} is \cclass{PSPACE}-complete.
\end{theorem}

\begin{conjecture}
Constant-wide \ruleset{Demi-Quantum Nim} is \cclass{NP}-hard to play optimally.
\end{conjecture}



If $\mathbb{B}$ is reachable from a classical position, then we call
$Z^{demi\QuantumLift}[\mathbb{B}]$
 a {\em classically reachable} demi-quantum game.

\begin{conjecture}
Reachable poly-wide \ruleset{demi-Quantum Nim} remains \cclass{PSPACE}-hard to play optimally.
\end{conjecture}

\begin{openquestion}
Is reachable 2-wide \ruleset{demi-Quantum Nim} polynomial-time solvable?
Can reachable constant-width \ruleset{demi-Quantum Nim} be \cclass{NP}-hard to play optimally?
\end{openquestion}

\subsubsection{Budgeted-QCGs}

In budgeted-quantum combinatorial games (budgeted-QCGs), players compete in a quantum combinatorial game with a ``quantum budget."

\begin{definition}[Games with Quantum Budgets]
Suppose $Z^{\QuantumLift(\phi,w)}$ is a
quantum game extended from a game instance $Z$, with quantum flavor $\phi \in \{A,B,C,C',D\}$, and superposition width $w$.
For any integer $q\geq 0$,
let $Z^{symBudgeted[q]\QuantumLift(\phi,w)}$ denote the budgeted quantum game of  $Z^{\QuantumLift(\phi,w)}$ in which both players can make at most $q$ quantum moves in total.
For any integers $q_L, q_R\geq 0$,
let $Z^{asymBudgeted[q_L,q_R]\QuantumLift(\phi,w)}$ denote the budgeted quantum game of $Z^{\QuantumLift(\phi,w)}$ in which
player $L$ and player $R$, respectively, can make at most $q_L$ and $q_R$ quantum
moves.
\end{definition}

Note that
$Z^{symBudgeted[q]\QuantumLift(\phi,w)}$
 and $Z^{asymBudgeted[q_L,q_R]\QuantumLift(\phi,w)}$ has a basic difference:
If $Z$ is an impartial game, then the former remains impartial, while, even with $q_L = q_R$, the latter is no longer impartial.
The number of realizations in any quantum board encountered in playing $Z^{symBudgeted[q]\QuantumLift(\phi,w)}$
($Z^{asymBudgeted[q_L,q_R]\QuantumLift(\phi,w)}$)
is at most $w^q$ ($w^{q_L+q_R}$).

Unlike demi-QCGs, budgeted-QCGs allows players to strategically
define their desirable quantum board.
After both players ran out of quantum moves, they play the remaining games as a demi-quantum game.
But one thing is for sure:
The resulting
demi-quantum game is reachable from a classical starting position.

\begin{conjecture}[\ruleset{Nim} with A Single Quantum Move]
For any instance $Z$ of \cclass{Nim}, for any flavor $\phi\in\{A,B,C,C',D\}$,
$Z^{symBudgeted[1]\QuantumLift(\phi,2)}$ is polynomial-time solvable.
The same is conjectured for $Z^{asymBudgeted[1,0]\QuantumLift(\phi,2)}$
and for
$Z^{asymBudgeted[0,1]\QuantumLift(\phi,2)}$.
\end{conjecture}

\begin{openquestion}[Fixed-Parameter Complexity]
Can \ruleset{Budgeted-Quantum Nim} with parameters  $q$ ($q_L,q_R$) and $w$ be solvable in time $N^{O(poly(w^q))}$, where $N$ denotes the descriptive complexity of the underlying \ruleset{Nim} instance, from which the budgeted quantum game is extended?
\end{openquestion}

\subsection{The Landscape of Quantumized Complexity}

Shaped by the celebrated Schaefer's dichotomy theorem in complexity theory \cite{SchaeferDichotomy},
We entered the field of quantum combinatorial games
expecting a more {\em dichotomy outlook}---between \cclass{P} and \cclass{PSPACE} and between \cclass{NP} and \cclass{PSPACE}---over the quantumized complexity landscape.
We have since been pleasantly surprised by our findings along the way:
The evidences that we have collected---illustrated in Theorem \ref{Theo:PSPACE-Intermediate} in particular---have highlighted regions of quantum games with intermediate intractable complexity between \cclass{NP} and \cclass{PSPACE}.
We believe that we only have a brief glimpse of the quantumized complexity landscape so far.

Among the many basic complexity-theoretical questions about quantum games, we would like to have sharper mathematical characterization on:

\begin{itemize}
\item the power of large superposition width of quantum moves;

\item the dependency of  game's structure/complexity on the degree-of-quantumness in starting positions and reachability from classical positions;

\item the connection between quantumized complexity and the height of game trees, and

\item in particular, the impact of degree-of-succinctness of game's representation---which is often the source of superpolynomial game-tree height---to both classical and quantumized complexity.
\end{itemize}

At the current stage, we only have a limited understanding of quantum moves with polynomial superposition width.
We know by  Theorem \ref{Theo:Upperbound} that increasing superposition width from 2 to polynomial does not cause quantumized complexity to leap over \cclass{PSPACE} for polynomially-short games.
However, our current proof of polynomial-time solvability of \ruleset{Quantum Undirected Geography} with a classical start
is tailored to superposition width $2$.
So, it is still possible---although we conjecture otherwise---larger superposition width may bring out intractability.

Succinct representations of games and optimization problems have been
at the center of complexity studies, as they fundamentally impact the measure of input size.
For example, the number of bits in
the binary representation of an integer is
logarithmic in the magnitude of the number.
Thus, binary representation, as oppose to
unary representation, is far more succinct for numbers.
Consider the fundamental problem {\sc Factoring}.
While it is {\em pseudo}-polynomial time solvable---i.e., if numbers are given in
the Unary Numeral System, then the method that we learned at high school has time complexity polynomial in the input size---it remains an outstanding open problem whether or not {\sc Factoring} can be solved in time polynomial in the size of its succinct binary representation.
In combinatorial games, for example, while \ruleset{Unary Nim} is a polynomially-short game measured by its input size, the height of game tree of \ruleset{Nim} where the numbers of stones in piles are specified by binary representation could be exponential in the input size.

The succinctness contributes to the possibility  that \ruleset{Quantum Nim} may have complexity outside \cclass{PSPACE}.
Other games involves graphs \& numbers  are also good candidate for beyond \cclass{PSPACE} quantumized complexity.
For instance, consider the following \ruleset{MaxCut} game:
An instance (position) is defined by a weighted undirected graph $G=(V,E,w)$ and a {\em partition} $(S, V\setminus S)$ of $V$,
where $w: E \rightarrow \mathbb{Z}$ assigns weights to edges.
In this impartial game,
  two players takes turns to move a vertex across the cut.
A move is {\em feasible} if the new partition has total {\em cut weight} larger than that of the current partition.
With a potential-function-based argument, one can show that this game has to terminate in finite---but possibly exponential in the input size---steps.
In fact, \ruleset{MaxCut} is a form of
succinctly-represented instance of \ruleset{Geography}:
Consider the {\em Partition Graph} $H$ of $G=(V,E,w)$, where vertices of $H$ are $V$'s partitions, and there is a directed edge in $H$ from partition $(S,V\setminus S)$ to $(S',V\setminus S)$ if
$S$ and $S'$ differ by one vertex, and the total cut weight of $S'$ is larger than that of $S$.
It is well-known that $H$ is a DAG.
Thus, playing \ruleset{MaxCut} on $G$ is the same as playing \ruleset{Geography} on $H$, which is an exponentially-sized DAG given in succinct representation by $G$.

\begin{openquestion}
Is \ruleset{Quantum MaxCut} with a classical start \cclass{PSPACE}-solvable?
\end{openquestion}

We can ask the same question for similar succinctly-represented games based on polynomial local search (PLS) of other combinatorial substructures such as flows, spanning trees, matchings, linear programs, {\em etc., etc}.

` 


\subsection{Science for Fun}

More than a decade ago when we designed the \cclass{PSPACE}-complete combinatorial game \ruleset{Atropos} \cite{Burke2008,10.5555/1714301}, we implemented our game.
At the time, we were drawn to its colorful yet elegant two-dimensional game positions.
\ruleset{Atropos} is an impartial game, played on the Sperner's triangle and is based on the beautiful, celebrated Sperner's Lemma \cite{SPE28}.
Like \ruleset{Hex} \cite{NashHex,Gale:1979},
\ruleset{Atropos} is also built on the mathematical foundation of Brouwer fixed-point theorem \cite{Brouwer}.
We shared our implementation with our colleagues and students at Boston University.
The game soon was incorporated in the class projects for computer science department's AI class.
Many students who built AI search programs for playing \ruleset{Atropos} in their class projects acknowledged the educational value of the game beyond the AI class.
They told us that they enjoyed their learning of Sperner's lemma (the discrete Brouwer's fixed-point theorem instrumental to computational understanding \cite{PAP94,DGPJournal,NashJACM} of Nash equilibrium in game theory \cite{NAS50} and Arrow-Debreu's market equilibrium in economics \cite{AD}).\footnote{It has always been a special memory for us that Herb Scarf called us in 2011 to load our \ruleset{Atropos} program onto his laptop, and then demonstrated it to Kenneth Arrow during his 90th birthday celebration at Jerusalem.}
Our own experience with \ruleset{Atropos} contributed to our great admiration for Allan Goff for creating \ruleset{Quantum Tic-Tac-Toe} as an educational tool for learning quantum mechanics~\cite{goff2006quantum}.

It is our desire to
 set up an on-line suite of elegant quantum games for us and more people to explore,
both as educational tools to learn quantum combinatorial game theory, and as subjects for potential projects in AI search.
We anticipate that practical implementation
of quantum games
 and visualization of superposition on a computer will be also be a wonderful source for interesting, and possibly challenging,
 research problems intersecting geometry, design, and optimization.
We are looking forward to identifying more basic research problems from our on-going ``Science for Fun'' project.

\bibliographystyle{plain} 

\begin{thebibliography}{10}

\bibitem{akl2010importance}
Selim~G Akl.
\newblock On the importance of being quantum.
\newblock {\em Parallel processing letters}, 20(03):275--286, 2010.

\bibitem{AD}
K.J. Arrow and G.~Debreu.
\newblock Existence of an equilibrium for a competitive economy.
\newblock {\em Econometrica}, 22:265--290, 1954.

\bibitem{WinningWays:2001}
Elwyn~R. Berlekamp, John~H. Conway, and Richard~K. Guy.
\newblock {\em Winning Ways for your Mathematical Plays}, volume~1.
\newblock A K Peters, Wellesley, Massachsetts, 2001.

\bibitem{Bouton:1901}
Charles~L. Bouton.
\newblock Nim, a game with a complete mathematical theory.
\newblock {\em Annals of Mathematics}, 3(1/4):pp. 35--39, 1901.

\bibitem{Brouwer}
L.~E.~J. Brouwer.
\newblock {\mbox{\"{U}}} ber {A}bbildung von {M}annigfaltigkeitn.
\newblock {\em Math. Annale}, 71:97--115, 1912.

\bibitem{Burke2008}
Kyle~W. Burke and Shang-Hua Teng.
\newblock Atropos: A \cclass{PSPACE}-complete \mbox{Sperner} triangle game.
\newblock {\em Internet Math.}, 5(4):477--492, 2008.

\bibitem{10.5555/1714301}
Kyle~Webster Burke.
\newblock {\em Science for Fun: New Impartial Board Games}.
\newblock PhD thesis, USA, 2009.

\bibitem{NashJACM}
Xi~Chen, Xiaotie Deng, and Shang-Hua Teng.
\newblock Settling the complexity of computing two-player nash equilibria.
\newblock {\em J. ACM}, 56(3), May 2009.

\bibitem{ONAG:2001}
John~H. Conway.
\newblock {\em On numbers and games {(2.} ed.)}.
\newblock A {K} Peters, 2001.

\bibitem{DGPJournal}
C.~Daskalakis, P.W. Goldberg, and C.H. Papadimitriou.
\newblock The complexity of computing a {N}ash equilibrium.
\newblock {\em SIAM Journal on Computing}, to appear. A version of this paper
  might be found at
  http://www.eecs.berkeley.edu/{\mbox{$\sim$}}costis/journal\_ver10.pdf.

\bibitem{dorbec2017toward}
Paul Dorbec and Mehdi Mhalla.
\newblock Toward quantum combinatorial games.
\newblock {\em arXiv preprint arXiv:1701.02193}, 2017.

\bibitem{eisert1999quantum}
Jens Eisert, Martin Wilkens, and Maciej Lewenstein.
\newblock Quantum games and quantum strategies.
\newblock {\em Physical Review Letters}, 83(15):3077, 1999.

\bibitem{Eppstein}
D.~Eppstein.
\newblock Computational complexity of games and puzzles, 2006.
\newblock http://www.ics.uci.edu/$\sim$eppstein/cgt/hard.html.

\bibitem{EvenTarjanHex}
S.~Even and R.~E. Tarjan.
\newblock A combinatorial problem which is complete in polynomial space.
\newblock {\em J. ACM}, 23(4):710–719, October 1976.

\bibitem{DBLP:journals/jct/FraenkelL81}
Aviezri~S. Fraenkel and David Lichtenstein.
\newblock Computing a perfect strategy for n x n chess requires time
  exponential in n.
\newblock {\em J. Comb. Theory, Ser. {A}}, 31(2):199--214, 1981.

\bibitem{DBLP:journals/tcs/FraenkelSU93}
Aviezri~S. Fraenkel, Edward~R. Scheinerman, and Daniel Ullman.
\newblock Undirected edge geography.
\newblock {\em Theor. Comput. Sci.}, 112(2):371--381, 1993.

\bibitem{Gale:1979}
David Gale.
\newblock The game of {H}ex and the {B}rouwer fixed-point theorem.
\newblock {\em American Mathematical Monthly}, 10:818--827, 1979.

\bibitem{glos2019role}
Adam Glos and Jaros{\l}aw~Adam Miszczak.
\newblock The role of quantum correlations in cop and robber game.
\newblock {\em Quantum Studies: Mathematics and Foundations}, 6(1):15--26,
  2019.

\bibitem{goff2006quantum}
Allan Goff.
\newblock Quantum tic-tac-toe: A teaching metaphor for superposition in quantum
  mechanics.
\newblock {\em American Journal of Physics}, 74(11):962--973, 2006.

\bibitem{Grundy:1939}
P.~M. Grundy.
\newblock Mathematics and games.
\newblock {\em Eureka}, 2:198---211, 1939.

\bibitem{2011takumiquantum}
Takumi Ishizeki and Akihiro Matsuura.
\newblock Solving quantum tic-tac-toe.
\newblock {\em International Conference on Advanced Computing and Communication
  Technologies}, 2011.

\bibitem{KoPHSeparation}
Ker-I Ko.
\newblock Relativized polynomial time hierarchies having exactly k levels.
\newblock {\em SIAM J. Comput.}, 18(2):392–408, April 1989.

\bibitem{Lander}
Richard~E. Ladner.
\newblock On the structure of polynomial time reducibility.
\newblock {\em J. ACM}, 22(1):155–171, January 1975.

\bibitem{LichtensteinSipser:1980}
David Lichtenstein and Michael Sipser.
\newblock Go is polynomial-space hard.
\newblock {\em J. ACM}, 27(2):393--401, 1980.

\bibitem{NAS50}
J.~Nash.
\newblock {Equilibrium points in n-person games}.
\newblock {\em Proceedings of the National Academy of the USA}, 36(1):48--49,
  1950.

\bibitem{NashHex}
John~F. Nash.
\newblock {\em Some Games and Machines for Playing Them}.
\newblock RAND Corporation, Santa Monica, CA, 1952.

\bibitem{PapadimitriouBook:1994}
C.~H. Papadimitriou.
\newblock {\em Computational Complexity}.
\newblock Addison Wesley, Reading, Massachsetts, 1994.

\bibitem{PAP94}
C.H. Papadimitriou.
\newblock {On the complexity of the parity argument and other inefficient
  proofs of existence}.
\newblock {\em Journal of Computer and System Sciences}, pages 498--532, 1994.

\bibitem{Reisch:1981}
S.~Reisch.
\newblock Hex ist {PSPACE}-vollst{\"a}ndig.
\newblock {\em Acta Inf.}, 15:167--191, 1981.

\bibitem{SchaeferDichotomy}
Thomas~J. Schaefer.
\newblock The complexity of satisfiability problems.
\newblock In {\em Proceedings of the Tenth Annual ACM Symposium on Theory of
  Computing}, STOC '78, page 216–226, New York, NY, USA, 1978. Association
  for Computing Machinery.

\bibitem{DBLP:journals/jcss/Schaefer78}
Thomas~J. Schaefer.
\newblock On the complexity of some two-person perfect-information games.
\newblock {\em Journal of Computer and System Sciences}, 16(2):185--225, 1978.

\bibitem{SiegelCGT:2013}
A.N. Siegel.
\newblock {\em Combinatorial Game Theory}.
\newblock Graduate Studies in Mathematics. American Mathematical Society, 2013.

\bibitem{SPE28}
E.~Sperner.
\newblock {Neuer Beweis fur die Invarianz der Dimensionszahl und des Gebietes}.
\newblock {\em Abhandlungen aus dem Mathematischen Seminar Universitat
  Hamburg}, 6:265--272, 1928.

\bibitem{Sprague:1936}
R.~P. Sprague.
\newblock \"{U}ber mathematische {K}ampfspiele.
\newblock {\em T\^{o}hoku Mathematical Journal}, 41:438---444, 1935-36.

\end{thebibliography}

\end{singlespace}

\end{document}


\section{Complexity of Concrete Quantum Games}
\label{Sec:Concrete}

Our presentation
in the body of the paper have been centered on
a general perspective rather than specific games, aiming to highlight basic questions in
quantum combinatorial game theory.
Here, for readers who are interested in concrete quantum games derived from classical ones, we summarize the complexity characterization obtained so far
in this paper.
The table below may serve as a starting basis for further advancements in computational understanding of quantum combinatorial game theory.

\begin{center}
\begin{tabular}{ |p{3cm}|c|p{2.3cm}|c|c| }
    \hline
    Ruleset & Hardness & Reduced From & Contained In & Section \\
    \hline \hline
    \ruleset{Q-Undirected Geography}
        & \cclass{P}
        &
        & \cclass{P}
        & \cref{thm:undirectedGeography} \\
        \hline
    \ruleset{Q-Selector}
        & \cclass{co-NP}
        & 3-SAT
        & \cclass{co-NP}
        & \cref{thm:selector} \\
        \hline
    \ruleset{Q-Partition-Free-Selector}
        & \cclass{NP}
        & 3-SAT & \cclass{NP}
        & TODO \\
        \hline
    \ruleset{Q-Partition-Free-Phantom}
        & $\Sigma_2 \cup \Pi_2$
        & $\Sigma_2$-SAT and $\Pi_2$-SAT
        & $\Sigma_2 \cup \Pi_2$
        & \cref{thm:partitionFreePhantom} \\
        \hline
    \ruleset{Phantom-Move}
        & $\Sigma_2 \cup \Pi_2$
        & $\Sigma_2$-SAT and $\Pi_2$-SAT
        & $\Sigma_2 \cup \Pi_2$
        &  \cref{thm:PhantomMove}  \\
        \hline
    \ruleset{Q-Avoid-True}
        & $\Sigma_2 \cup \Pi_2$
        & \ruleset{Phantom-Move}
        & \cclass{PSPACE}
        & \cref{thm:avoidTrue} \\
        \hline
    \ruleset{QQBF-Kyle}
        & \cclass{PSPACE}
        & QBF
        & \cclass{PSPACE}
        & \cref{thm:kyles} \\
        \hline
    \ruleset{QBRB}
        & \cclass{PSPACE}
        & QBF
        & \cclass{PSPACE}
        & \cref{thm:brb} \\
        \hline
    \ruleset{QMXC}
        & \cclass{PSPACE}
        & ???
        & \cclass{PSPACE}
        & \cref{thm:mxc} \\
        \hline
    \ruleset{Q-Geography}
        & \cclass{PSPACE}
        & \ruleset{Geography}
        & \cclass{PSPACE}
        & \cref{thm:geography}\\
        \hline
    \ruleset{Q-Node Kayles}
        & \cclass{PSPACE}
        & \ruleset{QBF}
        & \cclass{PSPACE}
        & \cref{thm:nodeKayles} \\
        \hline
    \ruleset{Q-BiGraph Node Kayles}
        & \cclass{PSPACE}
        & \ruleset{QBF}
        & \cclass{PSPACE}
        & \cref{thm:bigraphNK} \\
        \hline
    \ruleset{Q-Snort}
        & \cclass{PSPACE}
        & \ruleset{Q-BiGraph Node Kayles}
        & \cclass{PSPACE}
        & \cref{thm:snort} \\
        \hline
    \ruleset{Q-By-Player-Anywhere-Avoid-False}
        & \cclass{PSPACE}
        & \ruleset{Q-SNORT}
        & \cclass{PSPACE}
        & TODO \\
        \hline
\end{tabular}
\end{center}